\documentclass[superscriptaddress,a4paper,twocolumn,11pt,published, accepted=2022-08-09]{quantumarticle}
\pdfoutput=1

\usepackage[english]{babel}
\usepackage[utf8]{inputenc}

\usepackage[sort&compress,numbers]{natbib}
\usepackage{amsthm, 
color, 
mathtools, 
amsmath, 
amssymb,
setspace, 
bbm, 
tensor, 
subfigure, 
float,
mathtools, 
placeins, 
mathrsfs, 
verbatim,
makecell}
\usepackage[dvipsnames]{xcolor}
\usepackage{graphicx}
\usepackage[unicode=true, breaklinks=false, pdfborder={0 0 1}, backref=false, colorlinks=true, linkcolor=quantumviolet, citecolor=quantumviolet, urlcolor=quantumviolet]{hyperref}
\usepackage{cancel,bm}
\usepackage{enumerate}
\usepackage{soul}

\let\oldsqrt\sqrt
\def\sqrt{\mathpalette\DHLhksqrt}
\def\DHLhksqrt#1#2{%
\setbox0=\hbox{$#1\oldsqrt{#2\,}$}\dimen0=\ht0
\advance\dimen0-0.2\ht0
\setbox2=\hbox{\vrule height\ht0 depth -\dimen0}%
{\box0\lower0.4pt\box2}}

\newcommand{\tr}{\operatorname{tr}}

\newcommand{\Bcal}{\mathcal{B}}
\newcommand{\Dcal}{\mathcal{D}}
\newcommand{\Ecal}{\mathcal{E}}
\newcommand{\Fcal}{\mathcal{F}}
\newcommand{\Hcal}{\mathcal{H}}
\newcommand{\Ical}{\mathcal{I}}
\newcommand{\Mcal}{\mathcal{M}}

\newcommand{\Lcal}{\mathcal{L}}

\newcommand{\Ncal}{\mathcal{N}}
\newcommand{\Jcal}{\mathcal{J}}

\newcommand{\Pprob}{\mathbb{P}}

\newcommand{\Rcal}{\mathcal{R}}
\newcommand{\Ocal}{\mathcal{O}}
\newcommand{\ident}{\mathbbm{1}}
\newcommand{\Ads}{\mathbbm{A}}
\newcommand{\Bds}{\mathbbm{B}}
\newcommand{\Xds}{\mathbbm{X}}

\newcommand{\sbt}{\,\begin{picture}(-1,1)(-1,-3)\circle*{2.5}\end{picture}\ }

\newcommand*\xbar[1]{%
   \hbox{%
     \vbox{%
       \hrule height 0.5pt 
       \kern0.5ex
       \hbox{%
         \kern-0.2em
         \ensuremath{#1}%
         \kern-0.0em
       }%
     }%
   }%
}

\def\BraVert{\egroup\,\mid\,\bgroup}

\def\ketbra#1#2{\ket{#1\vphantom{#2}}\!\bra{#2\vphantom{#1}}}

\def\bra#1{\mathinner{\langle{#1}|}}

\def\ket#1{\mathinner{|{#1}\rangle}}

\def\braket#1{\mathinner{\langle{#1}\rangle}}

\usepackage{bbm, braket}
\usepackage[mathscr]{euscript}
\allowdisplaybreaks
\newtheorem*{theorem*}{Theorem}

\newtheorem{lemma}{Lemma}

\newtheorem{corollary}{Corollary}
\newtheorem{proposition}{Proposition}
\newtheorem{definition}{Definition}

\makeatletter
\newcommand{\neutralize}[1]{\expandafter\let\csname c@#1\endcsname\count@}
\makeatother
\newtheorem{thm}{Theorem}

\begin{document}

\title{Resource theory of causal connection}

\author{Simon Milz}
\email{simon.milz@oeaw.ac.at} 
\affiliation{Institute for Quantum Optics and Quantum Information, Austrian Academy of Sciences, Boltzmanngasse 3, 1090 Vienna, Austria}
\orcid{0000-0002-6987-5513}

\author{Jessica Bavaresco}
\affiliation{Institute for Quantum Optics and Quantum Information, Austrian Academy of Sciences, Boltzmanngasse 3, 1090 Vienna, Austria}
\orcid{0000-0002-4823-0353}

\author{Giulio Chiribella} 
\affiliation{QICI Quantum Information and Computation Initiative, Department of Computer Science, The University of Hong Kong, Pokfulam Road, Hong Kong}
\affiliation{Department of Computer Science, University of Oxford, Wolfson Building, 15 Parks Road, Oxford OX1 3QD, United Kingdom }
\affiliation{Perimeter Institute for Theoretical Physics, 31 Caroline St North, Waterloo, ON N2L 2Y5, Canada }
\orcid{0000-0002-1339-0656}

\begin{abstract}
The capacity of distant parties to send  signals to one another is a fundamental requirement  in many information-processing tasks. Such ability is determined by the causal structure connecting the parties, and more generally, by the intermediate processes carrying signals from one laboratory to another. Here we build a fully fledged resource theory of causal connection for all multi-party communication scenarios, encompassing those where the parties operate in a definite causal order and also where the order is indefinite. We define and characterize the set of free processes and three different sets of free transformations thereof, resulting in three distinct resource theories of causal connection. In the causally ordered setting,  we identify the most resourceful processes in the bipartite and tripartite  scenarios. In the general setting, instead, our results suggest that there is no global most valuable resource.   We establish the signalling robustness as a resource monotone of causal connection and provide tight bounds on it for many pertinent sets of processes. Finally, we introduce a resource theory of causal non-separability, and show that it is -- in contrast to the case of causal connection -- unique. Together our results offer a flexible and comprehensive framework to quantify and transform general quantum processes, as well as insights into their multi-layered causal connection structures. 
\end{abstract}

\maketitle

\newpage

\tableofcontents

\newpage

\section{Introduction}\label{sec::Intro}
The ability of separated parties to signal to one another is one of the most fundamental resources in information processing.  Examples are abundant even in everyday life; for instance, any group call where the participants want to communicate with each other requires a network of communication links of sufficiently high quality. Somewhat less mundanely, quantum protocols such as teleportation and superdense coding require reliable transmission of classical and quantum bits~\cite{nielsen_quantum_2000}.

The possibility of signalling between a set of parties is determined by the causal structure in which they operate. For example, if Alice's operations causally precede Bob's, then Bob cannot send any signal to Alice. More generally, the amount of signalling between Alice and Bob will depend on the background processes that connect their laboratories. When the parties operate in a definite causal order, their causal structure can be described by the framework of quantum combs~\cite{chiribella_transforming_2008, chiribella_quantum_2008, chiribella_theoretical_2009}, that is, matrices that describe  intermediate processes connecting the parties in a given order. More generally, the parties could operate in an indefinite causal structure, which can be described with the framework of process matrices~\cite{OreshkovETAL2012,araujo_witnessing_2015},  representing intermediate processes connecting the parties in an indefinite order. 
The goal of this paper is to provide a rigorous framework in which combs and process matrices can be regarded as resources for the communication between multiples parties.  

The canonical framework for characterizing resources is that of resource theories~\cite{horodecki_quantumness_2013, coecke_mathematical_2016}. 
Any resource theory consists of a set of free (i.e., non-resourceful) objects and a set of free transformations thereof (i.e., transformations that are assumed not to create resources), together with resource monotones (i.e., functions that are non-increasing under free transformations, and as such quantify the amount of resource in a given object).

In quantum mechanics, resource theories have been highly successful in characterizing resources such as entanglement~\cite{vedral_quantifying_1997, brus_characterizing_2002, plbnio_introduction_2007, horodecki_quantum_2009}, purity~\cite{horodecki_reversible_2003, streltsov_maximal_2018}, athermality~\cite{brandao_resource_2013, gour_resource_2015, lostaglio_introductory_2019}, asymmetry~\cite{gour_resource_2008, vaccaro_tradeoff_2008, gour_measuring_2009, marvian_theory_2013, marvian_extending_2014, piani_robustness_2016}, coherence~\cite{aberg_quantifying_2006, baumgratz_quantifying_2014, levi_quantitative_2014, chitambar_comparison_2016, chitambar_critical_2016,  winter_operational_2016, yadin_quantum_2016, napoli_robustness_2016, streltsov_colloquium_2017, wu_quantum_2020}, imaginarity~\cite{hickey_quantifying_2018, wu_resource_2021, wu_operational_2021} and nonclassicality~\cite{wolfe_quantifying_2020} to name but a few pertinent examples (see, e.g., Ref.~\cite{chitambar_quantum_2019} for an overview). 

In most of the existing resource theories, the resources are quantum states. Recently, there has been increasing interest in a new type of resource theories, where the resources are quantum channels (that is, valid transformations of quantum states into quantum states). For example, in Refs.~\cite{gour_dynamical_2020, gour_dynamical_res_2020, gour_entanglement_2021}, the entanglement of bipartite channels has been investigated, leading to a resource theory where the free objects are channels and the free transformations are quantum supermaps~\cite{chiribella_transforming_2008, chiribella_theoretical_2009, chiribella_quantum_2013} (that is, valid transformations of quantum channels into quantum channels). In an analogous fashion, the distillation of quantum channel resources has recently been considered in Ref.~\cite{regula_fundamental_2021}. More closely related to our present work, in Refs.~\cite{chiribella_quantum_2019, kristjansson_resource_2020}, resource theories of communication have been introduced, focusing on the possible transformations of  communication  channels. 

Here we establish a general framework to quantify the resourcefulness of the causal structure connecting an arbitrary number of parties. Our main goal is to characterize to what extent a given causal structure enables signalling between the parties. The resulting framework will be called a resource theory of \textit{causal connection}. However, a slight modification of the framework also allows us to formalize an explicit resource theory of \textit{causal non-separability}, where causal indefiniteness -- instead of causal connection -- is the resource of interest. A salient aspect of the resource theories of causal connection and causal non-separability is that, unlike in most other resource theories, the basic resources are neither quantum states nor quantum channels but rather quantum causal structures, described by quantum combs and process matrices. Accordingly, the free transformations are transformations of combs and process matrices, similar to the types of transformations considered in Refs.~\cite{bisio_theoretical_2019, castro-ruiz_dynamics_2018}. The resource-theoretic study of these higher-order processes is a relatively recent development. Transformations of combs and the resulting resource theories of (quantum) non-Markovianity have been discussed in Ref.~\cite{berk_resource_2021}. A resource theory of superpositions of causal orders has been provided in Ref.~\cite{taddei_quantum_2019}, leading to similar mathematical structures as those encountered in this work. Finally, the connection between channel capacities and causal order has been investigated in Ref.~\cite{jia_causal_2019}. In a similar vein, albeit not within a resource theoretical framework, the resources required for the implementation of different generalized quantum processes have been investigated in Refs.~\cite{nery_simple_2021, milz_genuine_2021}. 

In general, the  basic objects in our resource theories are process matrices, representing the background processes connecting the parties' laboratories. We start by establishing resource theories of causal connection. In this case the \textit{free} objects are process matrices that do \textit{not} allow for signalling between the parties. For the free transformations, we investigate three natural sets, one contained inside the other. The first set is the least restrictive one, and consists of all proper operations on process matrices that preserve the set of free processes. The second set consists of operations that cannot be used to establish signalling between the involved parties. Finally, the third, most restrictive set consists of operations that can be actively implemented by the parties using only local operations and shared correlations. The ensuing three possible sets of free transformations -- of which the latter two are the most physically motivated ones -- give rise to three distinct resource theories of causal connection, each with its own peculiar features.

Besides establishing a general resource theoretic framework, we introduce quantitative measures of causal connection. In particular, we propose the \textit{signalling robustness} as a faithful monotone of causal connection and find tight bounds for this monotone for the case of causally ordered processes as well as special sets of causally indefinite processes. As it turns out, both for the two- and three-party scenario,  the process with the highest signalling robustness in the causally ordered case is a Markovian process that only consists of unitary channels linking one party to the next. This result implies  that non-Markovian processes do not offer a higher amount of causal connection, despite their ability to create correlations between the outputs of different parties.    Perhaps more surprisingly, numerical results show that also causal non-separability does not increase the amount of causal connection, thus suggesting that the advantages of causally non-separable processes over causally ordered ones in quantum information tasks~\cite{chiribella_quantum_2013, chiribella_perfect_2012, OreshkovETAL2012, araujo_quantum_2017, quintino_inversion_2018, bavaresco_strict_2020, bavaresco_unitaries_2021} do not stem from an increased amount of causal connection. 

Like other types of robustness, the signalling robustness can be phrased as a semidefinite program (SDP). Using this SDP formulation, we provide an intuitive interpretation of the signalling robustness in terms of the `number and strength' of causal loops that could in principle be closed in a given process. This, in turn, offers an intuitive interpretation of causal connection even in cases where the underlying process is causally indefinite.

We then explore convertibility of processes under the free operations (for all three possible sets thereof) in the resource theory of causal connection. For some choices of free operations -- albeit, as we show, not all -- we find that they cannot change the causal ordering of a process, nor can they transform causally ordered processes into processes with indefinite causal order. On the other hand, causally non-separable processes can always be freely transformed into processes with definite causal order by means of free transformations. 

Finally, we adapt our results and the structural framework we establish for causal connection to formulate a  resource theory of non-separability. While this latter resource theory has been alluded to in the literature~\cite{araujo_witnessing_2015}, to our knowledge, it has not been fully worked out yet. Leveraging our considerations on causal connection, we lay out an explicit resource theory of causal non-separability and characterize the so-called causal robustness, introduced in Ref.~\cite{araujo_witnessing_2015} as a monotone in the resource theory  of causal non-separability. Unlike in the case of causal connection, as we show, there is only \textit{one} meaningful resource theory of causal non-separability, namely the maximal one, i.e., where the set of free transformations coincides with the set of all proper operations that cannot create causal indefiniteness. 

Finally, using these insights from the resource theory of causal non-separability, we return to the resource theories of causal connection, and show that, for two of them, there is no total order with respect to free operations, even when only considering the set of causally non-separable processes.

The paper is structured as follows. In Sec.~\ref{sec::Preliminaries}, we provide the necessary mathematical prerequisites for the discussion of general quantum processes and transformations thereof. Based on the notion of non-signalling (NS) constraints, in Sec.~\ref{sec::FreeStates_FreeOp}, we formally introduce the set of free processes and free transformations, the latter based both on axiomatic and operational considerations. The signalling robustness, as well as its formulation in terms of an SDP is introduced in Sec.~\ref{sec::Montones}, where we also provide its alternative interpretation in terms of the number of causal loops that can be closed in a given process. Tight bounds on the signalling robustness in the bi- and multipartite case, as well as the most resourceful process on two parties and its transformations are presented in Sec.~\ref{sec::BoundsRes}. A derivation of the unique resource theory of causal non-separability is provided in Sec.~\ref{sec::Interconvertability}, where we also provide an investigation of the interconvertibility of causally non-separable processes by means of free transformations (with respect to the resource theory of causal connection). The manuscript then concludes in Sec.~\ref{sec::Conclusion}. Further mathematical background and the derivation of all monotones of the resource theory of causal connection can be found in the Appendix.

\section{Preliminaries}\label{sec::Preliminaries}
Throughout this paper, we predominantly consider the following scenario: Two parties (Alice and Bob) each have independent laboratories in which they can receive quantum states, manipulate them, and subsequently send them forward (we will consider all involved systems to be of finite dimension). In a causally ordered framework, there are three possibilities: Alice's actions could have an influence on the state that Bob receives (denoted by $A\prec B$), Bob's actions could have an influence on the state that Alice receives (denoted by $B\prec A$) or neither party can influence the other (denoted by $A||B$). The former two cases allow for communication between the parties, while in the latter case, no signal can be sent between them. 

The aim of this paper is to provide a meaningful and operationally clear-cut framework to quantify the  `value' of the possibility of communication between the involved parties.  We will refer to this possibility as the {\em causal connection} between the parties. We emphasize that, even though we mostly work in the two-party case, most of our results hold in more generality, and we will be explicit whenever results depend on the number of parties involved. 

The causal connection between multiple parties is described by a process matrix~\cite{OreshkovETAL2012,araujo_witnessing_2015}, which encapsulates all situations where the parties have a definite causal order, as well as more exotic scenarios in which the order between the parties is indefinite. Independent of the respective underlying process, the operations Alice (A) and Bob (B) can perform on the respective quantum systems they receive and subsequently send forward are described by an instrument $\Jcal_X = \{\Mcal_X^k\}$, a collection of (trace non-increasing) completely positive (\textbf{CP}) maps, each corresponding to the implemented operation given an observed measurement outcome of party $X$, that add up to a CP trace-preserving (\textbf{TP}) map. Once all possible joint probabilities $\Pprob(i,j|\Jcal_A,\Jcal_B)$ for all potential outcomes Alice and Bob could obtain given their respective instruments are measured, the underlying causal order (or lack thereof) can be deduced.

Before providing a resource theory of causal connection, we first introduce the mathematical background needed to represent causal connection and its manipulation. 

\begin{figure}
\centering
\includegraphics[width=0.45\textwidth]{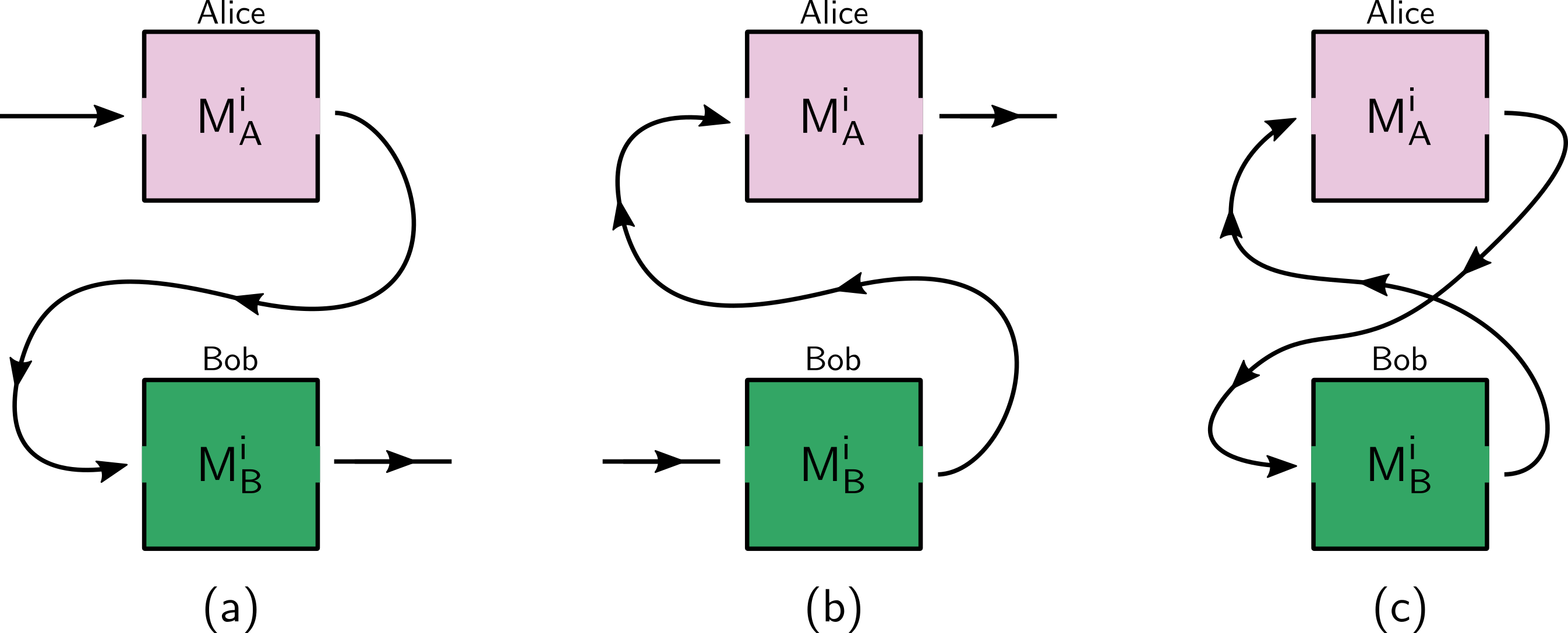}      
\caption{\textbf{Causal connection  between two parties.}  Two parties, Alice and Bob, receive, manipulate, and send forward quantum systems. The causal connection between Alice and Bob determines whether the operations performed by one party have an influence on the state the other party receives.  Specifically, we consider the situation where  Alice's manipulation can influence the input of Bob [panel (a)], or vice versa [panel (b)]. We will also consider more general scenarios that do not allow one to attribute a definite causal order, but no logical paradoxes occur [panel (c)]. Note that the case where none of the parties can influence the other can be considered a special case of (a) or (b). }
\label{fig::gen_setup}
\end{figure}

\subsection{Choi-Jamio{\l}kowski isomorphism and the link product}

For all of our considerations, it proves convenient to employ the Choi-Jamio\l{}kowski isomorphism (CJI) to represent CP maps by their Choi matrices~\cite{de_pillis_linear_1967, jamiolkowski_linear_1972, choi1975}. Denoting the input (output) space of party $X$ by $\Hcal_{X_I}$ ($\Hcal_{X_O}$), an instrument can equivalently be represented by a collection $\Jcal_X = \{M_X^k\}$ of positive semidefinite matrices $M^k_X \in \Bcal(\Hcal_{X_O} \otimes \Hcal_{X_I}) $ that satisfy $\tr_{X_O}(\sum_k M_X^k)= \ident_{X_I}$~\cite{chiribella_theoretical_2009}. Whenever there is no risk of confusion, we drop the explicit distinction between a map and its Choi matrix and we will employ the convention $X = X_IX_O$ when labelling Hilbert spaces. Whenever we want to make the distinction between maps and Choi matrices explicit, we use calligraphic letters to denote maps, and their upright version to denote the corresponding Choi matrix. 

A particular CPTP map that we encounter regularly is the partial trace $\tr_X$, with corresponding Choi matrix $\ident_X$, and the identity channel $\Ical_{X_O\rightarrow Y_I}$ with corresponding Choi matrix $\Phi^+ = \sum_{i,j}\ketbra{ij}{ij}$, where $\{\ket{ij}\}$ is a fixed product basis of $\Hcal_{X_O} \otimes \Hcal_{Y_I}$. Throughout, we will often add subscripts to the respective objects to denote the spaces they are defined on, and we assume that objects with different subscripts differ, i.e., $M_X \not\cong M_Z$, even if $\Hcal_X \cong \Hcal_Z$. Additionally, whenever applicable, we use the convention $M_X := \tr_Y M_{XY}$, etc.\footnote{Since the objects we consider are not always of unit trace, we frequently abuse this notation and implicitly assume that $M_X$ is a trace renormalized version of $\tr_Y M_{XY}$. This slight ambiguity does not affect the arguments we make, and whenever necessary we are explicit about it.}

It is convenient to introduce the \textit{link product} $\star$ which allows one to express the concatenation of maps in terms of their respective Choi matrices~\cite{chiribella_theoretical_2009}. Specifically, let $F_{XY} \in \Bcal(\Hcal_X \otimes \Hcal_Y)$ and $L_{YZ} \in \Bcal(\Hcal_Y \otimes \Hcal_Z)$ be the Choi matrices of two maps $\Fcal: \Bcal(\Hcal_X) \rightarrow \Bcal(\Hcal_Y)$ and $\Lcal : \Bcal(\Hcal_Y) \rightarrow \Bcal(\Hcal_Z)$. Then, the Choi matrix $N_{XZ} \in \Bcal(\Hcal_X \otimes \Hcal_Z)$ of the concatenation $\Ncal = \Lcal \circ \Fcal$ is given by 
\begin{gather}
\begin{split}
    N_{XZ} &= F_{XY} \star L_{YZ}  \\
    &:= \tr_Y[(F_{XY} \otimes \ident_{Z})(\ident_X\otimes L_{YZ}^{\mathrm{T}_{Y}})]\, ,
\end{split}
\end{gather}
where $\sbt^{\mathrm{T}_Y}$ denotes the partial transpose with respect to $\Hcal_Y$. Intuitively, the link product simply provides a convenient way of connecting different maps in the Choi representation. If the two constituents of the link product do not share any spaces they are defined on, then it coincides with the normal tensor product, i.e., $F_X\star L_Y = F_X\otimes L_Y $. As shown in Ref.~\cite{chiribella_theoretical_2009}, the link product of two positive semidefinite (Hermitian) matrices is again positive semidefinite (Hermitian). Additionally, the link product is associative and  commutative for all cases we consider.\footnote{Commutativity only holds up to reordering of the tensor factors of composite systems, but this subtlety will not be of importance for the considerations of this work.} In anticipation of future matters, it is worth considering these latter two properties in more detail. For example, the link product of three maps
\begin{gather}
\begin{split}
F_{X} \star L_{Y} \star K_{XY} 
&= F_{X} \star (K_{XY} \star L_{Y}) 
\\&=  (F_{X} \star K_{XY}) \star L_{Y}
\end{split}
\end{gather} 
can be understood in two different ways. On the one hand, as $K$ `acting on' $L$, yielding a new matrix on $\Hcal_X$ which is then linked with $F$; or it can be read as $K$ `acting on' $F$, yielding a new matrix on $\Hcal_Y$ which is then linked with $L$. We emphasize that the commutativity of the link product does not imply commutativity of the underlying maps. While we do not exclusively employ the link product to phrase our results, the flexibility it provides considerably simplifies the presentation of statements and proofs.

\subsection{Process matrices and causal order}

Suppose that two parties, Alice and Bob, perform local operations in two separate laboratories. Using the CJI for Alice's and Bob's instruments, the joint probability distribution of the outcomes can be expressed in terms of a {\em process matrix} $W \in \Bcal(\Hcal_{A_I} \otimes \Hcal_{A_O} \otimes \Hcal_{B_I} \otimes \Hcal_{B_O})$~\cite{OreshkovETAL2012}, as follows\footnote{To comply with the definition of the link product, this definition of the process matrix differs from the one in~\cite{OreshkovETAL2012} by a transpose. This difference is a mere notational one.}
\begin{gather}
\label{eqn::Prob}
\begin{split}
    \Pprob(i,j|\Jcal_A,\Jcal_B) &= \tr[W^\mathrm{T}(M_A^i \otimes M_B^j)] \\ &= W\star M_A^i \star M_B^j\, ,
\end{split}
\end{gather}
Mathematically, the process matrix is the Choi matrix of a higher-order map transforming Alice's and Bob's local operations into probabilities~\cite{chiribella_quantum_2013}. This matrix contains all spatio-temporal correlations of the underlying process that connects Alice's and Bob's laboratories~\cite{chiribella_memory_2008, shrapnel_updating_2017}. 
A graphical representation of this scenario is provided in Fig.~\ref{fig::Procmatrix}. 
\begin{figure}
    \centering
    \includegraphics[width=0.45\linewidth]{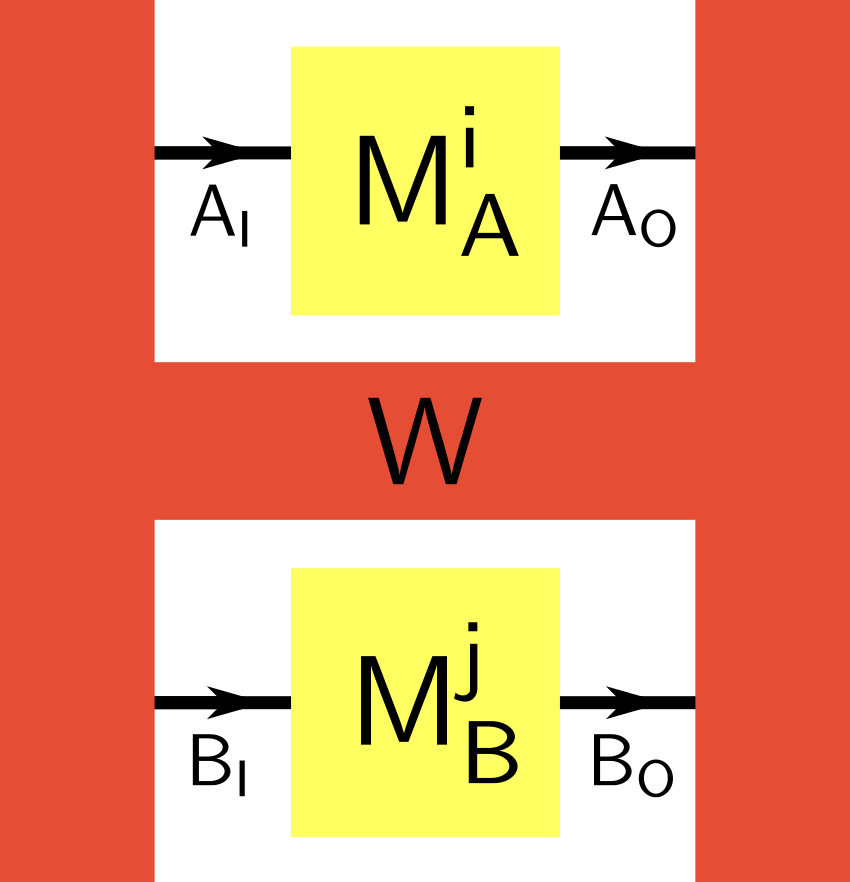}
    \caption{\textbf{Process matrix.} In their respective laboratories, Alice and Bob can perform arbitrary quantum operations $M_A^i, M_B^j$ transforming an input space $X_I$ to an output space $X_O$. All spatio-temporal correlations between them are given by $W$ and computed by `contraction' (i.e., link product) of $W$ with $M_A^i$ and $M_B^j$. Locally, causality holds (as signified by the arrows), but there is not necessarily a global ordering between the events in different laboratories.}
    \label{fig::Procmatrix}
\end{figure}
 
Demanding that quantum mechanics holds locally by requiring Alice and Bob to implement sets of quantum instruments, and that all probabilities must be positive and normalized, even when $W$ acts on correlated operations, implies
\begin{gather}
\label{eqn::DefProc}
    W\geq 0 \quad \text{and} \quad \tr[W^\mathrm{T}(M_A \otimes M_B)]=1
\end{gather}
for all Choi matrices $\{M_A, M_B\}$ of CPTP maps. These requirements can be succinctly summarized (for two parties) as~\cite{araujo_witnessing_2015}
\begin{gather}
\label{eqn::succinct}
W\geq 0, \quad W = L_V(W), \quad \tr(W)=d_{A_O}d_{B_O}\, ,
\end{gather}
where the projection operator\footnote{The subscript `V' refers to `valid'.} $L_V$  is of the form 
\begin{gather}
\label{eqn::DefProcMat}
\begin{split}
    L_V[W] = &{}_{A_O}W + {}_{B_O}W - {}_{A_OB_O}W - {}_{B_IB_O}W  \\
    &+ {}_{A_OB_IB_O}W  - {}_{A_IA_O}W  + {}_{A_IA_OB_O}W\, ,
\end{split}
\end{gather}
with ${}_{X}W \coloneqq \frac{\ident}{d_X} \otimes \tr_XW$. A generalization of $L_V$ to more parties can be found in Ref.~\cite{araujo_witnessing_2015}. While the explicit form of the projection operator $L_V$ is not intuitively obvious, individual terms in Eq.~\eqref{eqn::DefProcMat} can be given a direct meaning. For example, ${}_{A_O}W$ corresponds to the `Bob-to-Alice' one-way signalling part of $W$, while ${}_{A_OB_O}W$ coincides with the non-signalling part of $W$. The detailed forms of $L_V$ as well as all the projection operators we provide in the subsequent sections follow directly from their respective defining equation (for example, Eq.~\eqref{eqn::DefProc} for $L_V$), but only rarely allows for an intuitive interpretation -- in the sense that it is not obvious why they must contain certain combinations of individually interpretable terms beyond the algebra that yields them.

In what follows, we will refer to matrices $W$ that satisfy Eq.~\eqref{eqn::DefProc} (or the multi-party generalizations thereof) as proper process matrices (or simply processes) and denote the set of all such matrices by $\texttt{Proc}$. 

Throughout, we often make use of the fact that the operator ${}_{X}\sbt\,$ is trace-preserving, self-dual, and a projection, as well as the fact that ${}_{XY}\sbt = {}_{YX}\sbt$ holds. Self-duality means that for any two Hermitian matrices $Q$ and $R$ we have $\tr[Q({}_XR)] = \tr[({}_XQ)R]$, while the projection property implies that ${}_X({}_XR) = {}_XR$ for all matrices R. Both of these properties can be readily seen from the definition of ${}_X\sbt$ and also hold in the link product, i.e., $Q\star {}_XR = {}_X Q\star R = {}_X Q\star {}_XR$ whenever $Q$ and $R$ are both defined on the space $X$.

In anticipation of later investigations, it is worth discussing the properties of  \textit{causally ordered} processes and their relation to non-signalling constraints~\cite{piani_properties_2006, eggeling_semicausal_2002} in more detail. Following Refs.~\cite{chiribella_theoretical_2009, chiribella_quantum_2008}, we often call a process matrix that corresponds to a causally ordered process a proper comb (or simply comb). Causal order (for example, $A\prec B$, for two involved parties) means that Bob cannot signal to Alice. Put differently, in this causal order, the choice of Bob's instrument cannot influence the statistics Alice records. Using Eq.~\eqref{eqn::Prob}, this implies
\begin{gather}
    W^{A\prec B} \star M_A^i \star M_B = W^{A\prec B} \star M_A^i \star M'_B\, ,
\end{gather}
for  all CPTP maps $M_B,M_B'$. Using $\tr_{B_{O}} M_B = \tr_{B_{O}} M_B'\ = \ident_{B_I}$, it follows that
\begin{gather}
\label{eqn::AprecB}
    W^{A\prec B} = \ident_{B_{O}} \otimes W_{B_IA_OA_I}
\end{gather}
holds, where $W_{B_IA_OA_I}$ satisfies $\tr_{B_I}W_{B_IA_OA_I} = \ident_{A_O} \otimes \rho_{A_I}$ [due to the normalization constraint of Eq.~\eqref{eqn::DefProc}], and $\rho_{A_I}$ is a quantum state on $\Hcal_{A_I}$.

In the more general, multi-party, case, we distinguish between `input legs' of a process (labelled by $\{A_I, B_I, \dots\}$) and `output legs' of a process (labelled by $\{A_O,B_O, \dots\}$) (see Fig.~\ref{fig::Caus_Ordered} for a graphical representation).
\begin{figure}
    \centering
    \includegraphics[width=0.8\linewidth]{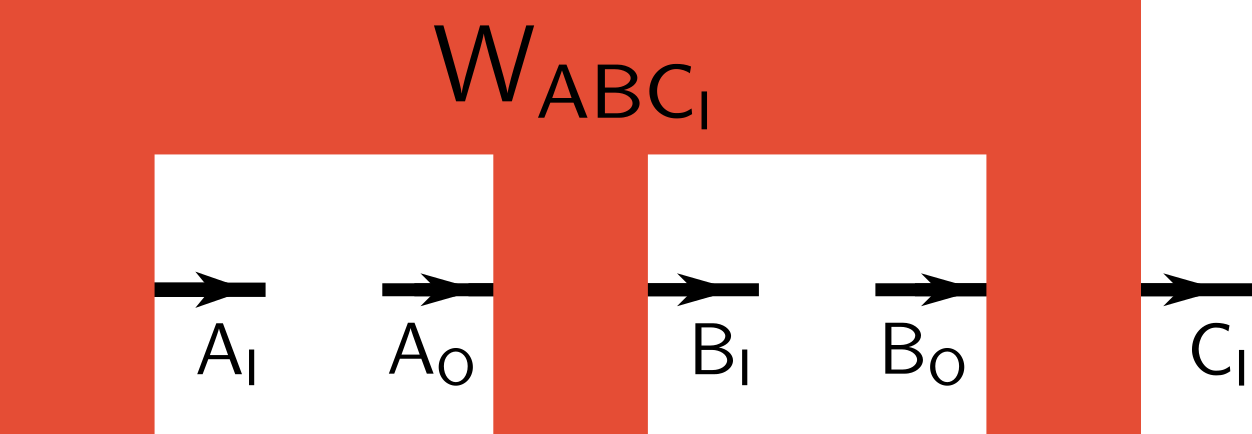}
    \caption{\textbf{Quantum comb.} For any process we distinguish between input legs -- here, $A_O$ and $B_O$ -- and output legs -- here, $A_I, B_I$, and $C_I$. If the process is causally ordered, then the hierarchy of trace conditions is satisfied (see Eqs.~\eqref{eqn::Hierarchy1} --~\eqref{eqn::Hierarchy3}).}
    \label{fig::Caus_Ordered}
\end{figure}
In slight abuse of notation, we also use $\prec$ to denote an ordering of legs (not just laboratories), e.g., $A_I\prec A_O \prec B_I \prec B_O \prec C_I$ for the example of Fig.~\ref{fig::Caus_Ordered}.  

Causal ordering then imposes a whole hierarchy of trace conditions that a comb has to satisfy~\cite{chiribella_quantum_2008,chiribella_theoretical_2009}. For example, for the process provided in Fig.~\ref{fig::Caus_Ordered}, we have $W_{ABC_I} \geq 0$, 
\begin{align}
\label{eqn::Hierarchy1}
    \tr_{C_I} W_{ABC_I} &= \ident_{B_{O}} \otimes W_{AB_I}, \\
    \label{eqn::Hierarchy2}
    \tr_{B_I}W_{AB_I} &= \ident_{A_O} \otimes W_{A_I}, \\
    \label{eqn::Hierarchy3}
    \text{and} \ \tr_{A_I} W_{A_I} &= 1\, .
\end{align}
More concisely, using the operators ${}_X\sbt$ defined above, these conditions can be subsumed in one projection operator $L_{A\prec B\prec C_I}$ as
\begin{align}
\notag
    W&= L_{A\prec B \prec C_I}[W] \\
    \notag
    &:= W- {}_{C_I}W + {}_{C_IB_O}W \\
    &\phantom{:=}- {}_{C_IB_OB_I}W + {}_{C_IB_OB_IA_O}W\, , \\
    \text{and} \ \tr W&= d_{A_O} d_{B_O}\, ,
    \label{eqn::CausHier}
\end{align}
where, for compactness, we have dropped the additional subscripts on $W$. The constraint of Eq.~\eqref{eqn::AprecB} is then just a special case of the above conditions, with a trivial final output leg. Importantly, any comb that satisfies the above conditions (or the multi-time generalization thereof) corresponds to a causally ordered process that can in principle be implemented as a standard quantum circuit~\cite{chiribella_quantum_2008,chiribella_theoretical_2009}. We will also refer to combs as \textit{deterministic operations}.

On a formal level, causal ordering amounts to signalling conditions. As mentioned, for a comb that satisfies the causal ordering $A\prec B$, the comb that Alice `sees' locally is independent of the CPTP operations that Bob can implement. On the other hand, for the same process, Alice can generally influence the local comb Bob sees, i.e., there exist CPTP maps $M_A$ and $M_A'$ such that $W^{A\prec B} \star M_A \neq W^{A\prec B} \star M'_A$. As such, Alice can inform Bob about what operation she performed, and thus signal to him. 

In our definition of free transformations of process matrices, we frequently encounter such non-signalling conditions and often express the properties of the objects we investigate by means of projectors like $L_{A\prec B \prec C_I}$ employed above. 

Finally, Eq.~\eqref{eqn::AprecB} (together with its version for the case $B\prec A$) allows us to provide the formal definition of \textit{causally separable} process matrices (for two parties) as such process matrices that can be written as 
\begin{gather}
    W = p\,W^{A\prec B} + (1-p)W^{B\prec A}\, ,
\end{gather}
for some probability $0\leq p\leq 1$, where completely non-signalling processes $W^{A||B} = \ident_{A_OB_O} \otimes \rho_{A_IB_I}$ -- which we also call parallel processes -- can be considered to belong to either the set $A\prec B$ or $B\prec A$ in the above definition. Whenever we want to express that a process is of order $X\prec Y$ but \textit{not} $X||Y$, we will denote it by $X\vec \prec Y$. We call the set of causally separable process matrices \texttt{Sep}. While it is not necessarily possible to represent all $W\in \texttt{Sep}$ by quantum circuits, they are nonetheless not overly exotic, since they can be perceived as causally ordered processes together with an initial coin flip that decides which process is chosen in an individual run.\footnote{The structure of separable processes becomes more involved in the multi-party case, where causal order could be chosen dynamically, for example the causal ordering between Bob and Charlie could depend on what Alice does~\cite{oreshkov_causal_2016, wechs_definition_2019}. These subtleties in the definition of causally separable processes are not of relevance for our considerations.}

All causally ordered ($A \prec B, B\prec A$, and $A||B$) processes as well as convex combinations thereof are described by a process matrix of the form~\eqref{eqn::DefProc}. However, there are processes that satisfy the requirements~\eqref{eqn::DefProc} (or variants of it, for the case of more parties) but do not correspond to a causally separable process~\cite{PhysRevA.88.022318, OreshkovETAL2012, oreshkov_causal_2016} (see Fig.~\ref{fig::Venn} for a graphical representation of the set of valid processes).
\begin{figure}
    \centering
    \includegraphics[width=0.7\linewidth]{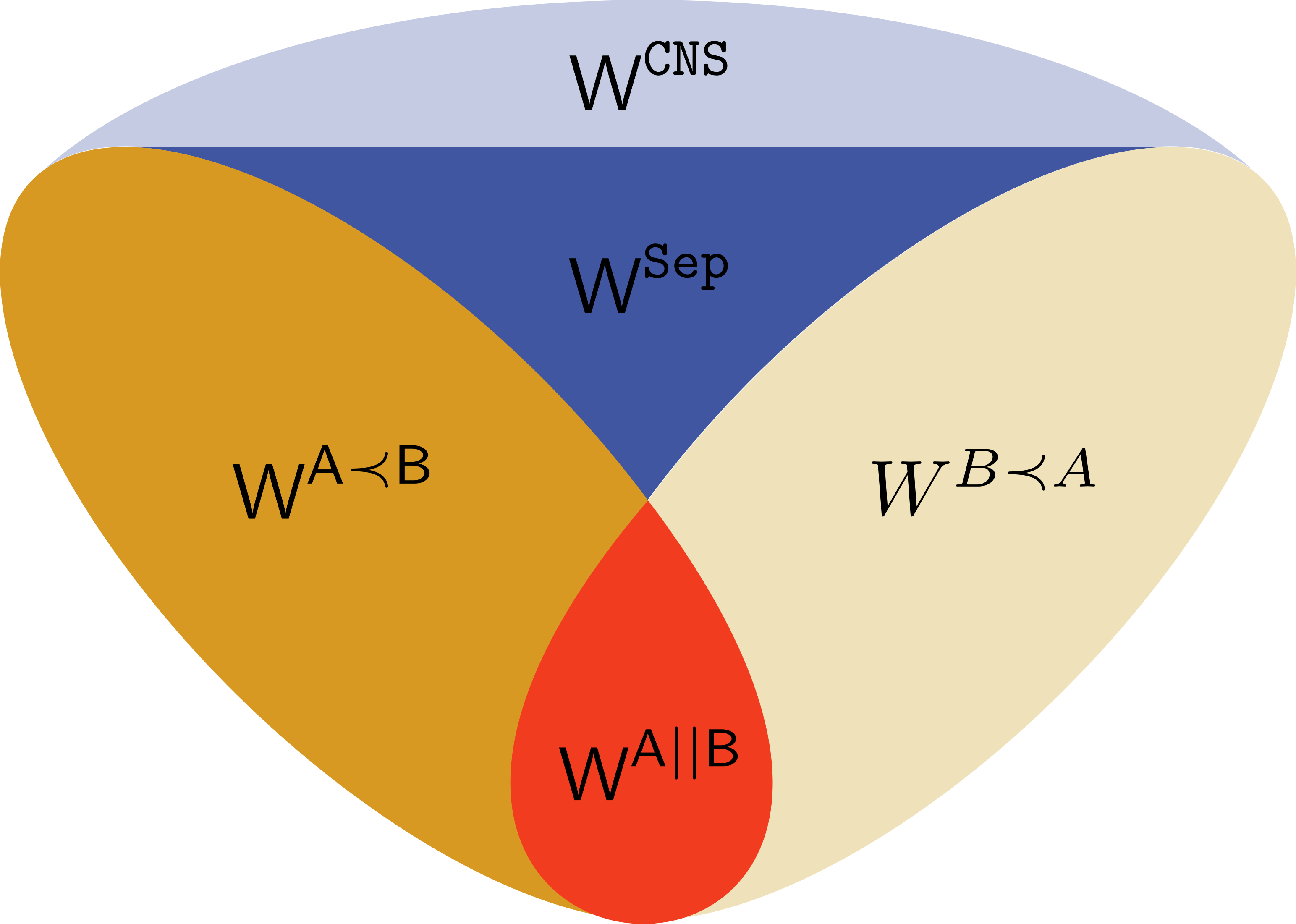}
    \caption{\textbf{Sets of process matrices.} The individual sets of process matrices that are ordered $A\prec B$ or $B\prec A$ form convex sets, with their intersection coinciding with the set of parallel/non-signalling processes ($A||B$). The convex hull of all causally ordered processes corresponds to $\texttt{Sep}$. The set $\texttt{Proc}$ of \textit{all} proper process matrices forms a strict (convex) superset of $\texttt{Sep}$, and all processes in $\texttt{Proc}\setminus \texttt{Sep}$ are causally non-separable (CNS).}
    \label{fig::Venn}
\end{figure}

Such processes have been dubbed \textit{causally non-separable}~\cite{OreshkovETAL2012,araujo_witnessing_2015}. Beyond the foundational intrigue they offer, they also widen the concept of processes and allow us to consider the largest class of conceivable resources for a resource theory of causal connection. In turn, their inclusion provides further insights into the fundamental differences between causally ordered and causally 
disordered processes; for example, as we shall see, general processes have more involved interconversion relations under free transformations than those that are causally ordered. 

In addition, causal non-separability -- instead of causal connection -- can itself be considered a resource for quantum processes. Structurally, the objects required to build up the corresponding resource theory of causal non-separability are similar to those employed in the resource theory of causal connection, and both theories share some of the same monotones. However, defining and characterizing `meaningful' sets of free transformations for the case of causal non-separability presents itself as a more layered question than for causal connection (we provide a detailed discussion in Sec.~\ref{sec::ResCausNonSep}) and requires much of the technical machinery that we provide in the subsequent sections. Consequently, here, we first start with a comprehensive discussion of the resource theory of causal connection and return to the resource theory of causal non-separability in Sec.~\ref{sec::Interconvertability}.

\section{Resource theories of causal connection}\label{sec::FreeStates_FreeOp}
Any resource theory is concerned  with the transformation of resources and the question of what resources can be transformed into each other given certain constraints. 
Here, the set of resources are process matrices, describing the background causal structure between a set of parties. Accordingly, the transformations of resources are given by supermaps (which we denote by $\Upsilon$). We will call such transformations {\em adapters}, since they can generally transform the input and output systems in the parties' laboratories, thus serving as ``adapters'' from one type of system to another, and we call them \textit{admissible adapters} if they map any proper process matrix onto a proper process matrix. A diagrammatic representation of the action of an adapter on a process matrix is shown in Fig.~\ref{fig::Adapter}.   
\begin{figure}
    \centering
    \includegraphics[width=0.98\linewidth]{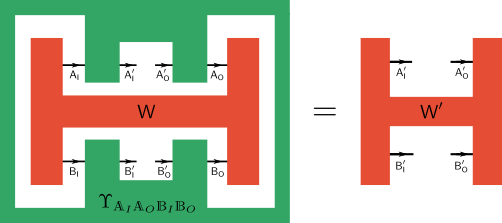}
    \caption{\textbf{General Adapters.}  A process matrix $W_{AB}$ can be mapped onto another process matrix $W'_{A'B'}$ via an adapter, such that $W'_{A'B'} = \Upsilon_{\Ads_{I}\Ads_{O}\Bds_{I}\Bds_{O}} \star W_{AB}$, where $\mathbbm{X}_{I/O} = X_{I/O}X_{I/O}'$. Depending on the considered resource theory, the adapter has to satisfy certain constraints. The green border around the process matrix $W$, belonging to the adapter $\Upsilon$, signifies that, in general, the adapter does not have to be causally ordered.}
    \label{fig::Adapter}
\end{figure}

In the following, we specify precisely the set of \textit{free objects}, and the set of \textit{free transformations} (also called `free operations' or, in what follows, `free adapters'). Since the property of interest is causal connection, the choice of free objects is rather straightforward: the free objects are process matrices that do not allow for communication between any of the parties. For the free transformations, on the other hand, we will see that -- besides the basic requirement of mapping proper processes to proper processes -- there are a few possible choices, all satisfying the basic requirement that free objects should be mapped onto free objects. In particular, we consider free transformations that leave the set of free processes invariant; that cannot be used themselves to establish communication between the parties; or that can be \textit{implemented} using exclusively non-signalling resources. We see below that these requirements are distinct from one another and lead to strictly different sets of free transformations.

\subsection{Free objects}\label{sec::FreeStates}

In a resource theory of causal connection, the set of free process matrices should not allow for any signalling between the parties. Mathematically, this means that the conditions   
\begin{align}\label{nosigprocessmatrix}
    \nonumber W\star M_A &= W\star M_A'  \\
    W\star M_B &= W\star M_B'    
\end{align} 
have to hold for all CPTP maps $M_A$, $M_A'$, $M_B$ and $M_B'$, which means that $W$ is ordered both $A \prec B$ \textit{and} $B\prec A$. These conditions imply that the set of free process matrices is the set of `non-signalling' or `parallel' processes $W^{A||B}$, that is, processes of the form 
\begin{gather}
\label{eqn::parallel}
    W^{A||B} \geq 0 \quad \text{and} \quad W^{A||B} = \rho_{A_IB_I} \otimes \mathbbm{1}_{A_OB_O}
\end{gather}
where $\rho_{A_IB_I}$ is a quantum state. Concretely, we have $W^{A||B} \star M_B = \rho_{A_I} \otimes \ident_{A_O}$ and $W^{A||B} \star M_A = \rho_{B_I} \otimes \ident_{B_O}$ for all CPTP maps $M_A$ and $M_B$, implying that no signalling between Alice and Bob can occur. Conversely, it is easy to see that processes of the form of Eq.~\eqref{eqn::parallel} are the only ones that do not enable communication between Alice and Bob. 

Hence, we define the set of free objects as the set $\texttt{Free} = \{W^{A||B}\}$ of all process matrices that satisfy Eq.~\eqref{eqn::parallel}. This definition extends naturally to multiple parties. Equivalently, using projectors, the set of free process matrices can also be characterized as the set of all matrices $W^{A||B} \in \Bcal(\Hcal_A \otimes \Hcal_B)$ that satisfy 
\begin{gather}
\begin{split}
\label{eqn::charFreeProj}
&W^{A||B} \geq 0, \\
&W^{A||B} = {}_{A_OB_O}W^{A||B}, \\
\text{and} \ &\tr W^{A||B} = d_{A_O}d_{B_{O}}\, .
\end{split}
\end{gather}
We emphasize that the set of free process matrices is convex. While convexity is not strictly necessary in a general resource theory, it will be helpful in some of applications discussed later in the paper. 

It is important to point out the difference between the above notion of non-signalling process matrix and the notion of non-signalling channel frequently used in the literature. In the case of two parties, a non-signalling channel is a bipartite channel $\Ncal:\Bcal(\Hcal_{A}\otimes \Hcal_{B}) \rightarrow \Bcal(\Hcal_{A'}\otimes \Hcal_{B'})$, jointly acting on systems in Alice's and Bob's laboratories. The channel $\Ncal$ is non-signalling from Alice to Bob if Alice's choice of input states cannot influence Bob's output state, i.e., 
\begin{gather}
    \tr_{A'}\{\Ncal[\rho_{A} \otimes \eta_{B}]\} = \tr_{A'}\{\Ncal[\xi_{A} \otimes \eta_{B}]\} 
\end{gather}
for all $\rho_{A}$, $\xi_{A}$ and $\eta_{B}$, and analogously for the case where Bob cannot signal to Alice. A bipartite channel is  non-signalling channel (NS) if it is non-signalling both from Alice to Bob and from Bob to Alice  (such channels are also called causal~\cite{beckman_causal_2001,piani_properties_2006}). 

On the other hand, a non-signalling \textit{process matrix}, as defined in Eq.\eqref{nosigprocessmatrix}, describes a causal structure where the local choice of CPTP map of one party do not have any observable effect in the other party's laboratory. These two notions of non-signalling can be related by considering process matrices as channels from $A_O \otimes B_O$ to $A_I \otimes B_I$ and realizing that, in this case, the free processes defined in Eq.~\eqref{eqn::parallel} correspond to a strict subset of all non-signalling channels. The reason for this difference is that, in the definition of non-signalling channels, the parties are allowed to signal only by preparing different states, while in the definition of non-signalling process matrix, the parties could signal by choosing local quantum channels. 

Throughout the paper, we also consider more general signalling scenarios and always understand non-signalling with respect to the operations the respective party could perform in their laboratory.

\subsection{Free transformations}

The transformations in a resource theory of causal connection are represented by adapters that transform process matrices into process matrices. We define their action through their associated Choi matrices $\Upsilon$. Let $W \in \Bcal(\Hcal_{A_I} \otimes \Hcal_{A_O} \otimes \Hcal_{B_I} \otimes \Hcal_{B_O})$ be a process matrix which is subjected to a transformation and let $W' \in  \Bcal(\Hcal_{A'_I} \otimes \Hcal_{A'_O} \otimes \Hcal_{B'_I} \otimes \Hcal_{B'_O})$ be the resulting process matrix (see Fig.~\ref{fig::Adapter}). Then, using the link product, the adapter responsible for this transformation, $\Upsilon\in\Bcal(\Hcal_{\Ads_I} \otimes \Hcal_{\Ads_O} \otimes \Hcal_{\Bds_I} \otimes \Hcal_{\Bds_O})$, with $\Xds_{I/O} := X_{I/O}X_{I/O}'$, acts on a process matrix according to
\begin{gather}
    W'_{A'B'} \coloneqq \Upsilon_{\Ads_I\Ads_O\Bds_I\Bds_O}\star W_{AB}. 
\end{gather}
The probability distribution observed by Alice and Bob when performing their local instruments $\Jcal_{A'}=\{M_{A'}^i\}$ and $\Jcal_{B'}=\{M_{B'}^j\}$ on the transformed process matrix $W'_{A'B'}$ is then given by
\begin{gather}
\begin{split}
    &\Pprob(i,j|\Jcal_{A'},\Jcal_{B'}) \\
    &= W'_{A'B'} \star M_{A'}^i \star M_{B'}^j \\
    &= (W_{AB} \star \Upsilon_{\Ads_I \Ads_O \Bds_I \Bds_O}) \star (M_{A'}^i \star M_{B'}^j).
\end{split}
\label{eqn::AdapterDef_actW}
\end{gather}
In a dual fashion, an adapter can also be understood as a transformation of the local  operations performed by the parties. This view can easily be understood thanks to the flexibility of the link product:
\begin{gather}
\begin{split}
    &\Pprob(i,j|\Jcal_{A'},\Jcal_{B'}) \\
    &= (W_{AB} \star \Upsilon_{\Ads_I \Ads_O \Bds_I \Bds_O}) \star (M_{A'}^i \star M_{B'}^j)\\
    &= W_{AB} \star [\Upsilon_{\Ads_I \Ads_O \Bds_I \Bds_O} \star (M_{A'}^i \star M_{B'}^j)] \\
    &=  W_{AB} \star M_{AB}^{ij}.
\end{split}
\label{eqn::AdapterDef_actM}
\end{gather}

\begin{figure}
    \centering
    \includegraphics[width=0.9\linewidth]{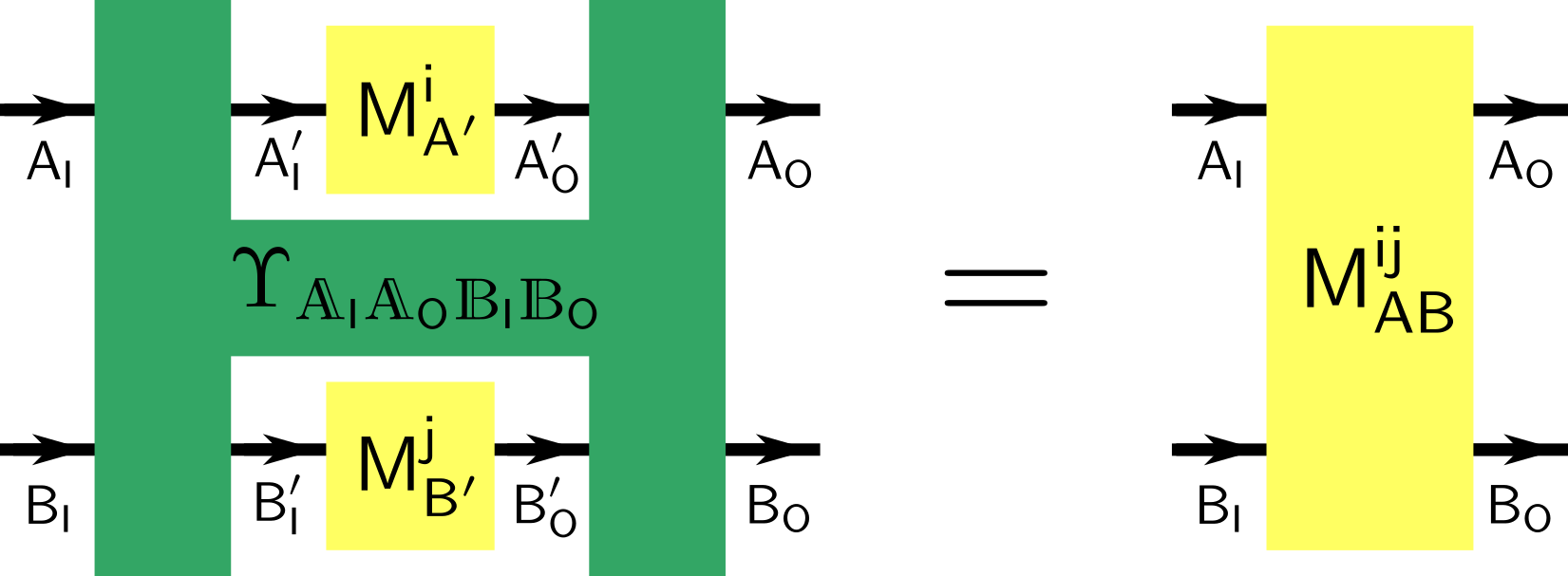}
    \caption{\textbf{Adapters as transformations of local operations.} Instead of being considered a transformation of the process matrix Alice and Bob share (depicted in Fig.~\ref{fig::Adapter}), adapters can also be understood as a transformation of their local operations.}
    \label{fig::MapstoMaps}
\end{figure}

Consequently, one can consider $\Upsilon_{\Ads_I\Ads_O\Bds_I\Bds_O}$ in terms of the mapping (see Fig.~\ref{fig::MapstoMaps})
\begin{gather}
    M_{AB}^{ij} \coloneqq \Upsilon_{\Ads_I\Ads_O\Bds_I\Bds_O} \star M_{A'}^i \star M_{B'}^j\, ,
\end{gather}
that maps the local operations of Alice and Bob to a new, possibly correlated operation. This corresponds to the dual, yet equivalent, action of the way in which we considered the action of the adapters in Eq.~\eqref{eqn::AdapterDef_actW}.

Let us now turn  to the definition of \textit{free} adapters. As we anticipated, there are several potential candidates for the set of free adapters, since which transformations are considered free ultimately depends on the level of control that Alice and Bob may have of the causal structure.   In the following, we lay down a hierarchy of physically motivated requirements corresponding to different levels of control.

Since our central concern is the ability of parties to communicate, natural requirements for free adapters are:
\begin{itemize}
    \item[\textbf{R$1$.}] To map proper process matrices onto proper process matrices.
    \item[\textbf{R$2$.}] To map free process matrices onto free process matrices.
    \item[\textbf{R$3$.}] To not allow for signalling.
    \item[\textbf{R$4$.}] To be implemented without communication between the parties.
\end{itemize}

Requirement R$1$ is not specifically a requirement on the {\em free} transformations per se, but rather a requirement for any physically meaningful transformation on the causal structure; valid causal structures should be transformed into valid causal structures.  Requirement R$2$ is the minimal requirement on \textit{free} transformations: they should map free objects onto free objects. Requirements R$3$ and R$4$ concern the physical properties of the adapters themselves and come into play as soon as Alice and Bob either have access to the adapter (Requirement R$3$) or implement the adapter themselves, by performing operations in their own laboratories (Requirement R$4$).

In the following we investigate the various sets of transformations that follow from these requirements and how they relate to each other. In particular, we show that Requirements R$2$ to R$4$ -- the requirements pertaining specifically to \textit{free} transformations -- form a strict hierarchy: the set of transformations corresponding to Requirement R$4$ is strictly smaller than the set of transformations corresponding to R$3$, which in turn is strictly smaller than the set of transformations corresponding to R$2$. Consequently these three requirements will lead to three \textit{distinct} sets of free adapters, and thus three distinct resource theories.

\subsubsection{Admissible transformations}
Let us start by analyzing Requirement R$1$, that adapters map process matrices to process matrices. Such adapters have been considered in Ref.~\cite{castro-ruiz_dynamics_2018} but here we introduce  them in a way that is dual to previous considerations. 

The set of transformations corresponding to Requirement R$1$ is the  set of \textit{all} conceivable  transformations of causal structures, free or not. Any  adapter $\Upsilon_{\Ads_I\Ads_O\Bds_I\Bds_O}$ in this set satisfies the basic condition  
\begin{gather}
\label{eqn::legal}
\Upsilon_{\Ads_I\Ads_O\Bds_I\Bds_O}\star W_{AB}\in\texttt{Proc}
\end{gather}
for all process matrices $W_{AB}\in\texttt{Proc}$. Note that here and in what follows, we always assume $\Upsilon \geq 0$ such that all probabilities are positive, even when the adapters only act on parts of process matrices~\cite{chiribella_transforming_2008,chiribella_quantum_2013,castro-ruiz_dynamics_2018}. 

We call any $\Upsilon \geq 0$ that satisfies Eq.~\eqref{eqn::legal} an \textit{admissible adapter}\footnote{In what follows, \texttt{A} will always stand for admissible, while $A$ and $\mathbbm{A}$ denote the party Alice.} and denote it by $\Upsilon^{\texttt{A}} \in \Theta_{\texttt{A}}$. Importantly,  Eq.~\eqref{eqn::legal} can be equivalently rewritten  as 
\begin{gather}
\label{eqn::legal2}
    \Upsilon_{\Ads_I\Ads_O\Bds_I\Bds_O}\star M_{A'B'} \text{ is NS channel}
\end{gather}
for all NS channels $M_{A'B'}$. This is the property that we use in this paper for the investigation of the properties of admissible adapters. 

The equivalence of Eqs.~\eqref{eqn::legal} and~\eqref{eqn::legal2} can be shown using the fact that process matrices are equivalently characterized as the most general supermaps transforming non-signalling channels into probabilities~\cite{chiribella_quantum_2013}, and by using the commutativity and associativity of the link product.  The proof of  equivalence is as follows. First, suppose that Eq.~\eqref{eqn::legal}  holds, that is, that  the matrix $\Upsilon^{\texttt{A}}_{\Ads_I\Ads_O\Bds_I\Bds_O} \star  W_{AB}$ is a valid process matrix whenever $W_{AB} \in \texttt{Proc}$.  This condition is equivalent to
\begin{gather}
\label{eqn::Req1}
\begin{split}
  &(\Upsilon^{\texttt{A}}_{\Ads_I\Ads_O\Bds_I\Bds_O} \star  W_{AB}) \star M_{A' B'} = 1
 \end{split}
\end{gather}
 for every non-signalling channel $M_{A'B'}$~\cite{chiribella_quantum_2013}.  Using the associativity  and commutativity of the link product, we then obtain the condition 
 \begin{gather}\label{eqn::aaaa}
      W_{AB} \star (\Upsilon^{\texttt{A}}_{\Ads_I\Ads_O\Bds_I\Bds_O} \star M_{A'B'}) =1  \, .
  \end{gather} 
for every process matrix $W_{AB}$ and for every non-signalling channel $M_{A'B'}$. Now, since the process matrix $W_{AB}$ is arbitrary, Eq.~\eqref{eqn::aaaa} implies that $\Upsilon^{\texttt{A}}_{\Ads_I\Ads_O\Bds_I\Bds_O} \star M_{A'B'}$ is a non-signalling channel whenever $M_{A'B'}$ is a non-signalling channel.  Hence,  Eq.~\eqref{eqn::legal2}  must hold.  Conversely, suppose that Eq.~\eqref{eqn::legal2} holds.  Then, following the above steps in the reverse order we obtain that Eq.~\eqref{eqn::legal} must hold.

Summarizing, Eqs.~\eqref{eqn::legal} and~\eqref{eqn::legal2} are  equivalent conditions. Physically, their equivalence is the higher-order analogue of the relation between Schr\"odinger and Heisenberg picture of a quantum evolution.

The  characterization  of admissible adapters in terms of NS channels can be used to derive an explicit equation for the constraints that an admissible adapter must satisfy.  Recall that a channel with Choi matrix $M_{AB}$ is non-signalling if and only if ~\cite{piani_properties_2006}
\begin{gather}
\label{eqn::NonSigTrace}
\begin{split}
    &{}_{X_O}M_{AB} = {}_{X_IX_O} M_{AB}\, ,
\end{split}
\end{gather}
where $X\in \{A,B\}$. In App.~\ref{app::Charc_LegalAdapters} we show that this is equivalent to the requirement 
\begin{gather}
\begin{split}
    M &= M - {}_{A_O}M + {}_{A_IA_O}M - {}_{B_O}M + {}_{A_OB_O}M \\
    &- {}_{A_IA_OB_O}M +  {}_{B_IB_O}M - {}_{A_OB_IB_O}M \\
    &+ {}_{A_IA_OB_IB_O}M =:L_{ns}[M]
\end{split}
\end{gather}
where we omitted the subscripts on $M_{AB}$ and introduced the self-dual trace preserving projector $L_{ns}$ onto the vector space spanned by non-signalling channels. Naturally, since it is the Choi matrix of a channel, $M_{AB}$ additionally satisfies $M_{AB} \geq 0$ and $\tr M_{AB} = d_{A_I}d_{B_I}$. Using the projector $L_{ns}$, in App.~\ref{app::Charc_LegalAdapters} we prove that the requirement that an adapter maps non-signalling maps to non-signalling maps is then equivalent to the following definition:
\begin{definition}[Admissible adapters]\label{def::LegalAdapters}
An admissible adapter is a matrix $\Upsilon^\textup{{\texttt{A}}}\in\Bcal(\Hcal_{\Ads_I}\otimes\Hcal_{\Ads_O}\otimes\Hcal_{\Bds_I}\otimes\Hcal_{\Bds_O})$ that satisfies
\begin{align}
    &\Upsilon^\textup{{\texttt{A}}} \geq 0 \\
    &\Upsilon^\textup{{\texttt{A}}} = \Upsilon^\textup{{\texttt{A}}} - L_{ns}'[\Upsilon^\textup{{\texttt{A}}}] + (L_{ns} \otimes L_{ns}')[\Upsilon^\textup{{\texttt{A}}}]\, , \\ 
    &{}_{AB}(L_{ns}'[\Upsilon^\textup{{\texttt{A}}}]) = {}_{ABA'B'}\Upsilon^\textup{{\texttt{A}}}\, ,\\
&\tr(\Upsilon^\textup{{\texttt{A}}}) = d_{A_I}d_{B_I}d_{A_O'}d_{B_O'},
\end{align}
where $L'_{ns}$ is the same projector as $L_{ns}$ but acts on the primed degrees of freedom. The set of admissible adapters is denoted $\Theta_\textup{{\texttt{A}}}$.
\end{definition}
Naturally, this definition can be extended to the multi-party setting in a similar vein.
  Note that this definition is equivalent to  that provided in Ref.~\cite{castro-ruiz_dynamics_2018}, the only difference being that we characterized the admissible adapters by their action on NS channels, rather than their action on process matrices.   While both approaches lead to the same definition of admissible adapters, the dual route we took here will make it easier to directly decide whether or not a given adapter is admissible.

Finally, it is worth stressing that the admissible adapters describe {\em all} logically conceivable transformations of process matrices, not just the free transformations.  In the next sections we will show concrete examples of admissible adapters that  map free processes to non-free processes (i.e., they violate Requirement R$2$).

\subsubsection{Free-preserving transformations}
Let us now analyze Requirement R$2$, that adapters map free objects to free objects. This constitutes the minimal requirement on free adapters. The ensuing set of \textit{free-preserving adapters}, denoted by $\Upsilon^\texttt{FP}$, is the set of adapters that satisfy
\begin{gather}
    \Upsilon^\texttt{FP}_{\Ads_I\Ads_O\Bds_I\Bds_O}\star W^{A||B}\in\texttt{Free},
\end{gather}
for all process matrices $W^{A||B}\in\texttt{Free}$. Using the characterization of free processes provided in Eq.~\eqref{eqn::charFreeProj} this implies 
\begin{gather}
\begin{split}
   {}_{A_O'B_O'} (\Upsilon^\texttt{FP} \star W^{A||B}) &= (\Upsilon^\texttt{FP} \star W^{A||B}) \\
   \text{and} \quad \tr(\Upsilon^\texttt{FP} \star W^{A||B}) &= d_{A'_O}d_{B_O'}
\end{split}   
\end{gather}
In App.~\ref{app:FreePres} we show that this implies the following characterization of the linear maps that transform free processes to free processes. 

\begin{definition}[Free-preserving adapters]\label{def::FreePreservAdapters}
A free-preserving adapter is a matrix $\Upsilon^\textup{\texttt{FP}}\in\Bcal(\Hcal_{\Ads_I}\otimes\Hcal_{\Ads_O}\otimes\Hcal_{\Bds_I}\otimes\Hcal_{\Bds_O})$ that satisfies
\begin{align}
    \Upsilon^\textup{\texttt{FP}} &\geq 0\, \\
    {}_{A_OB_O}\Upsilon^\textup{\texttt{FP}} &= {}_{A_O'B_O'A_OB_O}\Upsilon^\textup{\texttt{FP}}\, , \\
    {}_{A'B'A_OB_O}\Upsilon^\textup{\texttt{FP}} &= {}_{A'B'AB}\Upsilon^\textup{\texttt{FP}}\,, \\
    \tr(\Upsilon^\textup{\texttt{FP}}) &=  d_{A_I}d_{B_I}d_{A_O'}d_{B_O'}.
\end{align}
The set of free-preserving adapters is denoted $\Theta_\textup{\texttt{FP}}$.
\end{definition}

We emphasize that although the requirement to be free-preserving guarantees that proper \textit{free} processes are mapped to proper \textit{free} processes, it alone does not guarantee that the resulting set of operations is admissible, i.e., maps \textit{any} proper (non-free) process to a proper process. In fact, there are free-preserving adapters that can indeed map proper process matrices to objects outside of the set $\texttt{Proc}$. 

To see this, consider, for example, the free-preserving adapter 
\begin{gather}
    \Upsilon^\mathrm{1SW} = \Phi_{A_IB_I'}^+ \otimes \Phi_{B_IA_I'}^+ \otimes \Phi_{A_OA_O'}^+ \otimes \Phi_{B_OB_O'}^+\, ,
\end{gather}
which simply swaps Alice's and Bob's input lines (see Fig.~\ref{fig::Swap}). This adapter maps process matrices of the form $\rho_{A_IB_I} \otimes \ident_{A_OB_O}$ to $\rho_{B_I'A_I'} \otimes \ident_{A_O'B_O'}$, where $\rho_{B_I'A_I'}$ is simply a permutation of $\rho_{A_IB_I}$. On the other hand, starting with the non-signalling map $N = \Phi_{A_I'A_O'}^+ \otimes \Phi_{B_I'B_O'}^+$ which corresponds to independent identity channels  on Alice's and Bob's side, respectively, we see that $\Upsilon^\mathrm{1SW} \star N = \Phi_{B_IA_O}^+ \otimes \Phi_{A_IB_O}^+$ holds, where the latter is a channel that allows for perfect communication from Alice to Bob and from Bob to Alice, and hence does not represent a NS channel, implying $\Upsilon^\mathrm{1SW} \notin \Theta_{\texttt{A}}$. This fact can equivalently be understood in the dual perspective. To this end, we note that free-preserving adapters are equivalently characterized as transformations that map general channels $A_I'B_I' \rightarrow A_O'B_O'$ to channels $A_IB_I \rightarrow A_OB_O$. Since this latter set of transformations is strictly larger than the set of transformations that map NS channels to NS channels, it follows that there exist free-preserving adapters that are not admissible. 

On the other hand, there also exist admissible adapters that are not free-preserving. One such example is the admissible adapter
\begin{gather}
    \Upsilon = W'_{A'B'}\otimes\ident_{AB}/d_{A_O}d_{B_O},
\end{gather}
which discards any input process matrix and substitutes it for some non-free process $W'\in\texttt{Proc}\setminus\texttt{Free}$. This simple adapter is clearly admissible -- since it replaces any process matrix by a proper process $W'_{A'_IA'_OB'_IB'_O}$ -- but maps free processes to non-free ones.

We thus conclude that neither Requirement R$1$ nor R$2$ alone define good candidates for free adapters. In the next section we provide a more suitable set by imposing both requirements simultaneously.

\subsubsection{Admissible and Free-preserving transformations}
\label{sec::legal_free_pres}
Naturally, the combination of both Requirements R$1$ and R$2$ then leads us to the maximal set of operations that both map proper processes to proper processes and free processes to free processes. The adapters in this set, which we denote by $\Upsilon^{\texttt{A}\texttt{FP}}$, are the ones in the intersection $\Theta_{\texttt{A}\texttt{FP}}\coloneqq\Theta_{\texttt{A}}\cap\Theta_\texttt{FP}$. As we will see, adapters in this set -- while mathematically reasonable -- have some undesirable properties; for one, they can can change the causal ordering of processes. Additionally, they can create causal non-separability, making them somewhat `too powerful' to be good candidates for the free operations in a resource theory of causal connection.\footnote{In Sec.~\ref{sec::ResCausNonSep}, in the context of the resource theory of causal non-separability, we will discuss the set $\Theta_{\texttt{ASP}}$ of admissible adapters that preserve the set \texttt{Sep} of causally separable process matrices, which turn out to be the \textit{only} meaningful set of free adapters in the resource theory of causal non-separability. As we shall see there, $\Theta_{\texttt{ASP}} \setminus \Theta_{\texttt{AFP}} \neq \emptyset$, and no inclusion property holds amongst them.}

For a simple characterization, we note that some of the properties that make an adapter free-preserving are already implied by the corresponding stronger requirements on admissible adapters; both the trace of adapters $\Upsilon^{\texttt{A}\texttt{FP}} \in \Theta_{{\texttt{A}\texttt{FP}}}$ as well as their trace-rescaling properties are given by the requirements on admissible adapters, making the corresponding conditions for free-preserving adapters superfluous. In particular, for any admissible adapter we have 
\begin{gather}
\begin{split}
    {}_{A'B'A_OB_O}\Upsilon^{\texttt{A}} &= {}_{A'B'A_OB_O}(L_{ns}'[\Upsilon^{\texttt{A}}]) \\
    &= {}_{A'B'A_OB_O}(L_{ns} \otimes L_{ns}'[\Upsilon^{\texttt{A}}]) \\
    &= {}_{A'B'A_OB_O}(L_{ns}[\Upsilon^{\texttt{A}}]) \\
    &= {}_{A'B'AB}\Upsilon^{\texttt{A}} \, ,
\end{split}
\end{gather}
where we used the fact that $L_{ns}'$ is trace preserving (such that ${}_{A'B'}(L_{ns}'[\Upsilon^{\texttt{A}}]) = {}_{A'B'}\Upsilon^{\texttt{A}}$) as well as the fact that $\Upsilon^{\texttt{A}} \in \Theta_{\texttt{A}}$ (such that $L_{ns}'[\Upsilon^{\texttt{A}}] = L_{ns} \otimes L_{ns}'[\Upsilon^{\texttt{A}}]$). The last equality in the above equation is obtained by inserting the definition of $L_{ns}$. Consequently, two of the properties of free-preserving adapters are already implied by requirements on admissible adapters such that the definition of free-preserving admissible adapters is equivalent to 

\begin{definition}[Admissible and Free-preserving adapters]\label{def::LegalFreePreserv}
An admissible and free-preserving adapter is a matrix $\Upsilon^\textup{{\texttt{A}\texttt{FP}}}\in\Bcal(\Hcal_{\Ads_I}\otimes\Hcal_{\Ads_O}\otimes\Hcal_{\Bds_I}\otimes\Hcal_{\Bds_O})$ that satisfies
\begin{align}
        \Upsilon^\textup{{\texttt{A}\texttt{FP}}} &\in \Theta_{\textup{{\texttt{A}}}}\, ,\\
        {}_{A_OB_O}\Upsilon^\textup{{\texttt{A}\texttt{FP}}} &= {}_{A_O'B_O'A_OB_O} \Upsilon^\textup{{\texttt{A}\texttt{FP}}}.
\end{align}
The set of admissible and free-preserving adapters is denoted $\Theta_\textup{{\texttt{A}\texttt{FP}}}$.
\end{definition}

Interestingly, despite leaving the set of process matrices that do not enable signalling between the involved parties invariant, adapters in $\Theta_{\texttt{A}\texttt{FP}}$ can both change the causal order of processes they are applied to, and even create indefinite causal order.  

\begin{figure}
    \centering
    \includegraphics[width=0.99\linewidth]{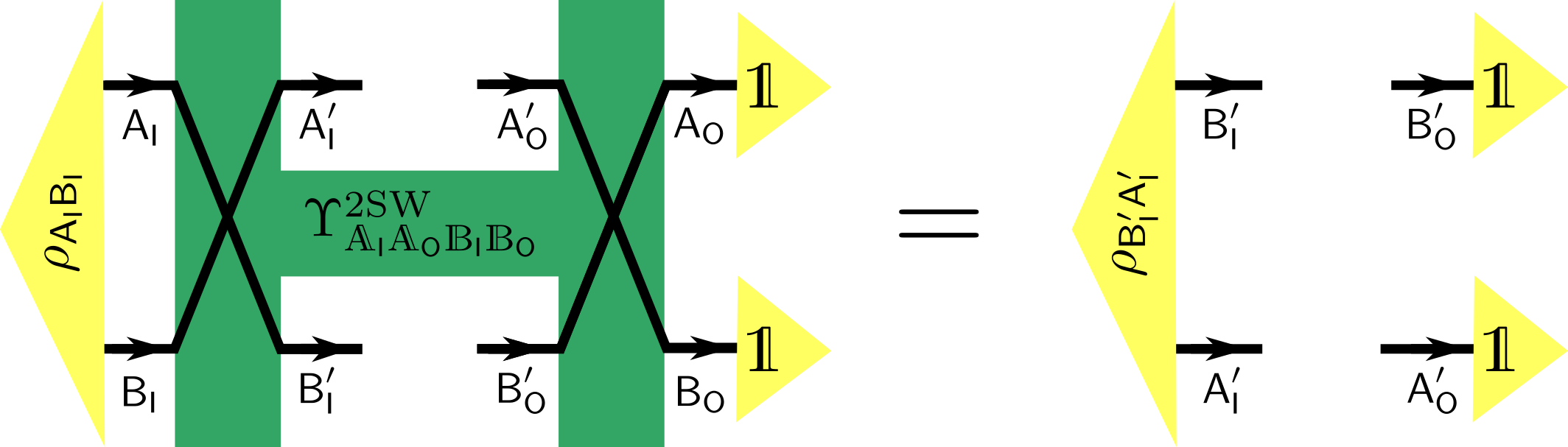}
    \caption{\textbf{Swapping Adapter.} An adapter $\Upsilon^{2\mathrm{SW}}$ that merely swaps the spaces $A_I$ and $B_I$, and $A_O$ and $B_O$, respectively, leaves the set of free process matrices invariant. However, it obviously allows for signalling between Alice and Bob, and  lies thus in $\Theta_{{\texttt{A}\texttt{FP}}}$ but not in $\Theta_{\texttt{NS}}$. Note that the adapter $\Upsilon^{1\mathrm{SW}}$ that is mentioned in the main text as an example of an adapter that is free-preserving but not admissible has the same action but without the second swap.}
    \label{fig::Swap}
\end{figure}

To see  that admissible free-preserving adapters can modify the causal order,   consider an adapter $\Upsilon^{\mathrm{2SW}}$ of the form (see Fig.~\ref{fig::Swap})
\begin{gather}\label{eqn::2swap}
    \Upsilon^{\mathrm{2SW}} = \Phi^+_{A_IB_I'} \otimes \Phi^+_{B_IA_I'} \otimes \Phi^+_{A_O'B_O} \otimes \Phi^+_{B_O'A_O}\, ,
\end{gather}
where we assume the dimensions of all involved spaces to be equal. This adapter amounts to two swap operations, which lead to the relabelling $A_I \mapsto B_I'$, $B_I \mapsto A_I'$, $A_O \mapsto B_O'$, and $B_O \mapsto B_O'$. It is easy to see that this adapter indeed maps free processes to free processes as we have 
\begin{gather}
\begin{split}
    \Upsilon^{\mathrm{2SW}}_{\Ads_I\Ads_O\Bds_I\Bds_O} &\star (\rho_{A_IB_I} \otimes \ident_{A_OB_O}) \\
    &= \rho_{B'_IA'_I} \otimes \ident_{A'_OB'_O}\, ,
\end{split}
\end{gather}
where the inverted the labels on $\rho_{B'_IA'_I}$ signify that the resulting state is the same as $\rho_{A_IB_I}$ but with its spaces swapped. Unlike the previous example  $\Upsilon^{\mathrm{1SW}}$, however, the adapter, consisting of two swap operations, is admissible.  To see this, we first consider  its effect on a product $M_{A'_IA_O'} \otimes M_{B_I'B_O'}$ of two CPTP maps. We have 
\begin{align}\label{eqn::bbbb}
\begin{split}
    \Upsilon^{\mathrm{2SW}}_{\Ads_I\Ads_O\Bds_I\Bds_O} &\star (M_{A'_IA_O'} \otimes M_{B_I'B_O'}) \\
    &= M_{A_IA_O} \otimes M_{B_IB_O}\, ,
\end{split}
\end{align}
where $M_{A_IA_O} = M_{B_I'B_O'}$ and $M_{B_IB_O} = M_{A'_IA_O'}$ (up to a relabeling of the involved spaces). Now, we recall that a matrix $M_{AB}$ is the Choi matrix of a non-signalling map $\Mcal:\Bcal(\Hcal_{A_I}\otimes\Hcal_{B_I}) \rightarrow \Bcal(\Hcal_{A_O}\otimes\Hcal_{B_O})$ if and only if  $M_{AB} \geq 0$ and 
\begin{gather}
\label{eqn::DecompNonSig}
    M_{AB} = \sum_i \lambda_i M_{A}^{(i)} \otimes M_{B}^{(i)}, \ \ \   \sum_i\lambda_i=1\, ,
\end{gather}
where all $\{M^{(i)}_A,M^{(i)}_B\}$ are Choi matrices of CPTP maps and $\lambda_i\in\mathbb{R}$~\cite{Gutoski09, chiribella_quantum_2013}. Hence, Eq.~\eqref{eqn::bbbb} implies that $\Upsilon^{\mathrm{2SW}}$ maps non-signalling maps to non-signalling maps, and, therefore, since it is positive semidefinite,  $\Upsilon^{\mathrm{2SW}} \in \Theta_{\texttt{A}\texttt{FP}}$. However, given that it swaps the involved spaces, a process that has causal ordering $A\vec \prec B$ is mapped to a process of causal ordering $B' \vec \prec A'$ by $\Upsilon^{\mathrm{2SW}}$. 

Physically, the supermap $\Upsilon^{\mathrm{2SW}}$ implements a transformation of causal structures that inverts the signalling relations between Alice and Bob: for instance, if originally signalling is only possible from Alice to Bob, then after the transformation signalling is only possible from Bob to Alice.

A slight modification of the above example shows the existence of admissible and free-preserving adapters that can map a process of order $A\prec B$ to a process that is a convex mixture of processes with orderings $A\prec B$ and $B\prec A$. The simplest example is an adapter that is a convex mixture of $\Upsilon^\mathrm{2SW}$ and an identity adapter $\Upsilon^\mathrm{Id}$, which leaves any process matrix it acts upon unchanged (except for adding primes to all involved spaces). Choosing $\Upsilon^\mathrm{mix} = p\Upsilon^{\mathrm{2SW}} + (1-p) \Upsilon^{\mathrm{Id}} \in \Theta^{\texttt{A}\texttt{FP}}$, with $p \in (0,1)$, then yields an adapter that maps a process matrix $W^{A\prec B}$ to 
\begin{gather} 
\Upsilon^\mathrm{mix} \star W^{A\prec B} = pW^{B'\prec A'} + (1-p) W^{A'\prec B'}\, ,
\end{gather}
which is a convex mixture of causal orders. 

However, the above process is still causally separable, and it remains to investigate whether admissible and free-preserving adapters can transform causally separable to causally non-separable processes. In App.~\ref{app::Maps_SeptoNonsep}, we work out a numerical method to search for such adapters and to verify the causal non-separability of the resulting processes. As a result of this investigation, we indeed find admissible and free-preserving adapters that  map causally separable to causally non-separable processes. For details, see App.~\ref{app::Maps_SeptoNonsep}.

The above observations suggest that, while the set $\Theta^{\texttt{A}\texttt{FP}}$  satisfies the  first two requirements for a set of free transformations, it may be  too big for a  physically meaningful  resource theory of causal connection. On the one hand, it allows one to change between causal orders, which seems to require internal signalling in the adapter, and on the other hand, perhaps more surprisingly, it also allows one to generate causally  non-separable processes from causally separable ones.  On top of this, it is intuitively clear that the realization of  many of the adapters in $\Theta_{\texttt{A}\texttt{FP}}$ requires some form of signalling between Alice and Bob. 
In Secs.~\ref{sssec::NS_Adapters} and~\ref{sss:lose} we will make this intuition rigorous, showing that the adapters that can be physically realized without signalling are a strict subset of $\Theta^{\texttt{A}\texttt{FP}}$.

The difference  between adapters that do not generate signalling and adapters that can be implemented without signalling operations can be better understood by referring to two different interpretation of the adapters. 
The first interpretation is that the adapters represent physical transformations implemented by a third party, other than Alice and Bob, who can alter the causal structure between them. In this setting, the natural requirement is that the operation should not create causal connections where such connections were not present. In this case, one may be interested in the set of admissible and free-preserving adapters. 

The second interpretation is that the adapters represent physical processes that are accessible to the parties. That is, the adapter is a multiport device, to which Alice and Bob can connect their inputs and outputs.  In this setting, Alice and Bob have full access to the adapter at hand and can, potentially, use it to exchange information with one another, independent of the process matrix they initially shared. In this setting, however, the correct physical requirement is  Requirement R$3$, which demands that free adapters should not allow for Alice and Bob to communicate with each other, even if they had full access to the inputs and outputs of the adapter at hand. As it will turn out, demanding this additional requirement then removes many of the surprising (and overly powerful) properties of admissible and free-preserving adapters. 

\subsubsection{Non-signalling transformations}\label{sssec::NS_Adapters}
While Requirements R$1$ and R$2$ impose restrictions on the kinds of process matrices that can be achieved by the action of an adapter, Requirement R$3$ focuses on  the  properties  of the   adapter   when   seen   itself   as   a   potential resource  shared  by  the  involved  parties. In particular, demanding Requirement R3 amounts to giving Alice and Bob complete access to their respective ‘sides’ of the adapter and demanding that they cannot use it for communication.  We  will  call adapters that satisfy this requirement \textit{non-signalling adapters}, denoted by $\Upsilon^\texttt{NS} \in \Theta_\texttt{NS}$.

\begin{figure}
    \centering
    \includegraphics[width=0.99\linewidth]{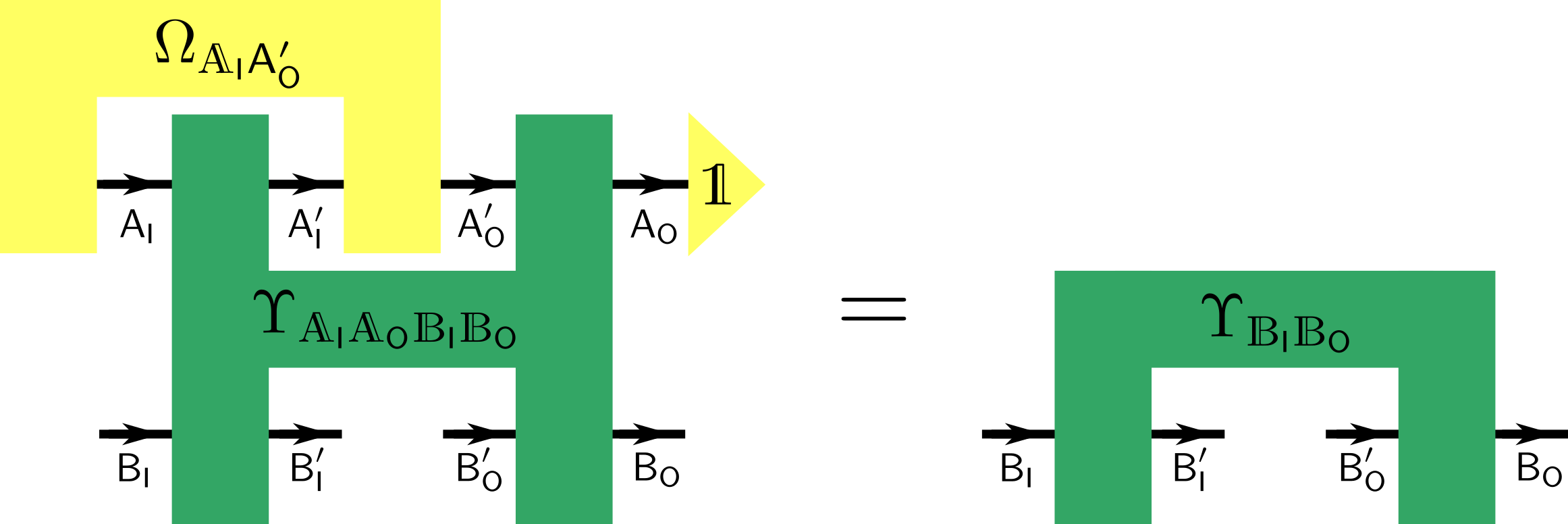}
    \caption{\textbf{Local operations on an adapter.} Given access to an adapter $\Upsilon_{\Ads_I \Ads_O\Bds_I\Bds_O}$, Alice and Bob could, in principle, use it to transmit information between them. To do so, Alice would perform deterministic operations on her end of the adapter, and if Bob's resulting part of the adapter depended on said operations, the adapter would allow for communication between the parties. Non-signalling adapters are exactly those adapters that do not allow for communication and the local adapters are independent of the respective operations on the other side (depicted above for Bob's side). Note that, for the depicted case, a deterministic operation is of the form $\Omega_{\Xds_IX_O'} \otimes \ident_{X_O}$.}
    \label{fig::AdLocalAc}
\end{figure}

Analogously to the case of channels and processes, no signalling from Alice to Bob by means of the adapter $\Upsilon^\texttt{NS}$ means that, whatever deterministic operation  Alice performs on her side of the adapter, once Alice's part is discarded, the remaining part  on Bob's side  must be independent of Alice's operation, and vice versa (see Fig.~\ref{fig::AdLocalAc}). 
The fact that the most general deterministic operation $\Omega_{\Xds_I\Xds_O}$ a party $X\in \{A,B\}$ could perform is a comb~\cite{chiribella_theoretical_2009} with ordering $X_I \prec X_I' \prec X_O' \prec X_O$ implies that it has the special form $\Omega_{X_I  X_I'  X_O'} \otimes \ident_{X_O}$, with $\tr_{X_O'}\Omega_{X_I X_I'  X_O'} = \Omega_{X_I} \otimes \ident_{X_I'}$  and $\tr  \Omega_{X_I}  =1$.  For short, we say that the combs $\Omega_{X_I X_I'  X_O'}$ are of type $X_I \prec X_I' \prec X_O'$.  

Now,  a  non-signalling adapter $\Upsilon^\texttt{NS}_{\Ads_I\Ads_O\Bds_I\Bds_O}$ has to satisfy the condition that Alice's operations cannot signal to Bob, namely 
\begin{align}
\notag
    \Upsilon^\texttt{NS}_{\Bds_I\Bds_O} &= \Upsilon^\texttt{NS}_{\Ads_I\Ads_O\Bds_I\Bds_O} \star \Omega_{A_IA_I' A_O'} \otimes \ident_{A_O} \\
    \label{eqn::non_signal1}
    &= \Upsilon^\texttt{NS}_{\Ads_I\Ads_O\Bds_I\Bds_O} \star \widetilde{\Omega}_{A_I A_I'A_O'} \otimes \ident_{A_O}
    \end{align}
 for all combs $\Omega_{A_IA_I'A_O'}$ and $\widetilde \Omega_{A_IA_I'A_O'}$  of type $A_I\prec A_I' \prec A_O'$. 
 Similarly, for non-signalling from Bob to Alice, it must hold that
 \begin{align}
\notag
    \Upsilon^\texttt{NS}_{\Ads_I\Ads_O} &= \Upsilon^\texttt{NS}_{\Ads_I\Ads_O\Bds_I\Bds_O} \star \Omega_{B_IB_I'B_O'} \otimes \ident_{B_O} \\
    \label{eqn::non_signal2}
    &= \Upsilon^\texttt{NS}_{\Ads_I\Ads_O\Bds_I\Bds_O} \star \widetilde{\Omega}_{B_IB_I' B_O'} \otimes \ident_{B_O}\, ,
\end{align}
for all combs $\Omega_{B_IB_I'B_O'}$ and $\widetilde \Omega_{B_IB_I'B_O'}$ of type $B_I\prec  B_I' \prec B_O'$. 
 
Eqs.~\eqref{eqn::non_signal1} and~\eqref{eqn::non_signal2} can be explicitly  characterized by a set of trace conditions, similar to the trace conditions for causally ordered combs~\cite{chiribella_theoretical_2009}. To derive the trace conditions for non-signalling adapters, let us assume that Alice tries to send a signal to Bob via the adapter $\Upsilon$ by simply feeding states into it. Her corresponding comb would have the form 
\begin{gather}
    \Omega_{A_IA_I' A_O'} \otimes \ident_{A_O} = \rho_{A_I} \otimes \eta_{A_O'} \otimes \ident_{A_I'A_O}, 
\end{gather}
where $\rho_{A_I}$ and $\eta_{A_O'}$ can be arbitrary quantum states. Now, requiring that the resulting comb
\begin{gather}
    \Upsilon_{\Ads_I\Ads_O\Bds_I\Bds_O} \star (\rho_{A_I} \otimes \eta_{A_O'} \otimes \ident_{A_I'A_O})
\end{gather}
on Bob's side is independent of $\eta_{A_O'}$ implies $\tr_{A_O}\Upsilon_{\Ads_I\Ads_O\Bds_I\Bds_O} = \ident_{A_O'} \otimes \Upsilon_{\Ads_I\Bds_I\Bds_O}$, while additional independence of $\rho_{A_I}$ implies $\tr_{A_I'}\Upsilon_{\Ads_I\Bds_I\Bds_O} = \ident_{A_I} \otimes \Upsilon_{\Bds_I\Bds_O}$. The same reasoning applies for operations on Bob's side. 

While these are not the most general operations that Alice and Bob can perform, these trace conditions are already sufficient to ensure the impossibility of communication by means of a non-signalling adapter (see App.~\ref{app::NSAdapt} for a proof). Combining them into the self-dual, trace preserving projectors $L_A$ and $L_B$, we arrive at the general definition of a non-signalling adapter, equivalent to that of non-signalling operations introduced in Ref.~\cite{taddei_quantum_2019}:

\begin{definition}[Non-signalling adapters]\label{def::free_adapters}
A non-signalling adapter is a matrix $\Upsilon^\textup{\texttt{NS}} \in \Bcal(\Hcal_{\Ads_I} \otimes \Hcal_{\Ads_O} \otimes \Hcal_{\Bds_I} \otimes \Hcal_{\Bds_O})$ that satisfies
\begin{align}
&\Upsilon^\textup{\texttt{NS}} \geq 0\, ,\\
\label{eqn::defFree}
&\Upsilon^\textup{\texttt{NS}} = (L_A \otimes L_B)[\Upsilon^\textup{\texttt{NS}}], 
\\\label{eqn::trFree} \mathrm{and} \quad &\tr\Upsilon^\textup{\texttt{NS}} = d_{A_O'}d_{A_I'} d_{B_O'}d_{B_I'}\, ,
\end{align}
where $L_X[\Upsilon] = \Upsilon - {}_{X_O}\Upsilon + {}_{X_OX_O'}\Upsilon - {}_{X_I'X_OX_O'}\Upsilon + {}_{X_IX_I'X_OX_O'}\Upsilon$. The set of non-signalling adapters is denoted $\Theta_\textup{\texttt{NS}}$.
\end{definition}
We show that this definition indeed yields adapters that satisfy Eqs.~\eqref{eqn::non_signal1} --~\eqref{eqn::non_signal2} in App.~\ref{app::NSAdapt}. While here we focus on the two-party case, this definition, as well as the remaining properties of non-signalling adapters we derive in this section, straightforwardly extend to the multi-party scenario. Since $L_A$ and $L_B$ are projectors, any $\Upsilon^{\texttt{NS}}$ that satisfies the above definition is invariant under both $L_A$ and $L_B$; invariance under $L_A$ encapsulates non-signalling from Alice to Bob, and vice versa for invariance under $L_B$. The trace condition in Eq.~\eqref{eqn::trFree} ensures that the resulting process matrices $\Upsilon^\texttt{NS}_{\Ads_I\Ads_O\Bds_I\Bds_O}\star W_{AB}$ are always properly normalized. 

Interestingly, it can be verified that the set of non-signalling adapters, defined by Requirement R$3$, automatically satisfies Requirements R$1$ and R$2$ as well. That is, non-signalling adapters not only cannot be used by Alice and Bob to signal to each other, but they also constitute transformations of process matrices that map $\texttt{Proc}\mapsto\texttt{Proc}$ and $\texttt{Free}\mapsto\texttt{Free}$.  It is important to stress, however, that Requirement R$3$ is more restrictive than the combination of Requirements R$1$ and R$2$; as we will see in the following, not all admissible and free-preserving adapters are non-signalling.

Let us first show that  non-signalling adapters satisfy  Requirement R$1$, that is, they  map valid process matrices into valid  process matrices. 
This could be done by insertion, i.e., by showing that the conditions of Def.~\ref{def::free_adapters} imply those of Def.~\ref{def::LegalAdapters}. Equivalently, here we show that any $\Upsilon^\texttt{NS} \in \Theta_\texttt{NS}$ maps non-signalling channels on $A'B'$ to non-signalling channels on $AB$. To this end, we first note that, due to the decomposition provided by Eq.~\eqref{eqn::DecompNonSig}, we can restrict our discussion to non-signalling maps of the product form $M_{A'B'} = M_{A'} \otimes M_{B'}$. For these maps, we have
\begin{gather}
\label{eqn::non-Signalling_axiom}
    \begin{split}
{}_{A_O}M_{AB} &= {}_{A_O}(\Upsilon^\texttt{NS}_{\Ads_I\Ads_O \Bds_I \Bds_O} \star M_{A'} \star M_{B'}) \\
    &= {}_{A_OA_O'}\Upsilon^\texttt{NS}_{\Ads_I\Ads_O \Bds_I \Bds_O} \star M_{A'} \star M_{B'} \\
    &= {}_{A_OA_O'}\Upsilon^\texttt{NS}_{\Ads_I\Ads_O \Bds_I \Bds_O} \star {}_{A_O'A_I'}M_{A'} \star M_{B'} \\
    &= {}_{A_OA_O'A_I'A_I}\Upsilon^\texttt{NS}_{\Ads_I\Ads_O \Bds_I \Bds_O} \star M_{A'} \star M_{B'} \\
    &= {}_{A_OA_I}({}_{A_O'A_I'}\Upsilon^\texttt{NS}_{\Ads_I\Ads_O \Bds_I \Bds_O} \star M_{A'} \star M_{B'}),
\end{split}
\end{gather}
where we have used the fact that ${}_X\sbt$ is a self-dual projection (i.e., it can be moved around freely in the link product) and the condition in Eq.~\eqref{eqn::NonSigTrace} ${}_{A_O'}M_{A'} = {}_{A_O'A_I'}M_{A'}$ of CPTP maps. Since ${}_{A_I}\sbt$ is a projection, the above implies ${}_{A_O}M_{AB} = {}_{A_IA_O}M_{AB}$. Thus, given that $M_{AB}$ satisfies Eq.~\eqref{eqn::NonSigTrace}, it is non-signalling from Alice to Bob. Analogously, we can show that $M_{AB}$ is non-signalling from Bob to Alice. Both the positivity of $M_{AB}$ and $\tr(M_{AB}) = d_{A_I}d_{B_I}$ are easy to see, implying that an adapter $\Upsilon^\texttt{NS} \in \Theta_\texttt{NS}$ maps any non-signalling channel $M_{A'B'}$ to a non-signalling channel $M_{AB}$, thus satisfying Requirement R$1$. Naturally, this result can be extended to more than two parties.

Now we show that non-signalling adapters also satisfy Requirement R$2$, i.e., they map free processes to free processes. This follows by application of the properties of non-signalling adapters $\Upsilon^{\texttt{NS}}$ and free processes $W^{A||B}$, which satisfy $W^{A||B} = {}_{A_OB_O}W^{A||B}$. With this, we have
\begin{gather}
\begin{split}
    W^{A'||B'}&:= \Upsilon^\texttt{NS}_{\Ads_I\Ads_O\Bds_I\Bds_O} \star W^{A||B} \\
    &= {}_{A_OB_O}\Upsilon^\texttt{NS}_{\Ads_I\Ads_O\Bds_I\Bds_O} \star W^{A||B} \\
    &= {}_{A'_{O}B'_{O}} (\Upsilon^\texttt{NS}_{\Ads_I\Ads_O\Bds_I\Bds_O} \star W^{A||B}) \\
    &= {}_{A'_{O}B'_{O}}W^{A'||B'},
\end{split}
\end{gather}
where, like in the previous derivations, we have used the properties of the operator ${}_{X}\sbt$ as well as those of non-signalling adapters. Again, positivity and proper trace normalization follow directly. Hence, non-signalling adapters are free-preserving.  

Summarizing,  non-signalling adapters satisfy Requirements R$1$ and R$2$;  we have $\Theta_{\texttt{NS}} \subseteq \Theta_{{\texttt{A}\texttt{FP}}}$.  We now show that this inclusion is strict, i.e., $\Theta_\texttt{NS} \subset \Theta_{\texttt{A}\texttt{FP}}$ (in what follows, `$\subset$' will always denote \textit{strict} inclusion). For this purpose, we  provide a concrete example of an adapter that satisfies Requirements R$1$ and R$2$ but not R$3$.

Again, we consider the adapter $\Upsilon^\mathrm{2SW} \in \Theta_{{\texttt{A}\texttt{FP}}}$ from Eq.~\eqref{eqn::2swap}, that implements two swap operators on the process matrices that it acts upon. We have previously shown that this adapter is admissible and free-preserving. However, as we have also seen, it can change the causal order of processes, a property that, as we show in the following, cannot be exhibited by non-signalling adapters, proving that $\Upsilon^{2\mathrm{SW}} \notin \Theta_\texttt{NS}$.

Let us see that every  non-signalling adapter must  preserve the causal order of process matrices, a property which is interesting  in its own right.    This statement is proven by direct insertion. Let $W_{AB}$ be a process matrix with causal ordering $B\prec A$. From Eq.~\eqref{eqn::AprecB} we see that this implies $W_{AB} = {}_{A_O}W_{AB}$ and ${}_{A_OA_I}W_{AB} = {}_{A_OA_IB_O}W_{AB}$. With this, we have
\begin{gather}
\begin{split}
    &W'_{A'B'} := \Upsilon^\texttt{NS}_{\Ads_I\Ads_O\Bds_I\Bds_O} \star W_{AB} \\
    &= {}_{A_O}\Upsilon^\texttt{NS}_{\Ads_I\Ads_O\Bds_I\Bds_O} \star W_{AB} \\
    &= {}_{A_O'} \Upsilon^\texttt{NS}_{\Ads_I\Ads_O\Bds_I\Bds_O} \star W_{AB} =  {}_{A_O'} W'_{A'B'},  
\end{split}
\end{gather}
where we have used the causality constraints on $W_{AB}$, the property ${}_{A_O}\Upsilon^\texttt{NS} = {}_{A_OA_O'}\Upsilon^\texttt{NS}$ [which can be seen by direct insertion into Eq.~\eqref{eqn::defFree}], as well as the self-duality of the operators ${}_X\sbt$. In the same vein, one shows that ${}_{A_O'A_I'} W'_{A'B'} = {}_{A_O'A_I'B_O'}W'_{A'B'}$, implying that $W'_{A'B'}$ has causal order $B' \prec A'$. Consequently, since the same can be shown for processes that are ordered $A\prec B$, adapters $\Upsilon^\texttt{NS} \in \Theta_{\texttt{NS}}$ do not change causal order, which means that $\Upsilon^\mathrm{2SW} \in \Theta_{\texttt{AFP}} \setminus \Theta_{\texttt{NS}}$. We thus have $\Theta_\texttt{NS} \subset \Theta_{\texttt{A}}\texttt{FP}$.

From the linearity of the action of $\Upsilon^\texttt{NS}$, one can also conclude that non-signalling adapters map causally separable process matrices to causally separable process matrices, a property that, as we have seen, is not satisfied by all adapters $\Upsilon^{\texttt{A}\texttt{FP}} \in \Theta_{\texttt{A}\texttt{FP}}$ either.

Hence, non-signalling adapters form a strict subset of the set of admissible and free-preserving adapters and not only satisfy the additional desirable property of not allowing for signalling between the involved parties, but moreover preserve causal order. However, the fact that non-signalling adapters do not allow for signalling does not necessarily mean that they can be \textit{implemented} with non-signalling resources alone. The implementation of adapters and the final requirement R$4$ are discussed in the next section.

\subsubsection{Transformations from local operations and shared entanglement}\label{sss:lose}
Requirement R$4$ demands that adapters should be implementable using only  non-signalling resources, such as local operations and shared entanglement. On the hierarchy of control over the adapters, we are now at the point where Alice and Bob do not just have access to their respective parts of the adapter, but, in order for the adapter to be free, they have to be able to implement it without employing signalling resources. Such non-signalling resources are local operations and shared entanglement and, following the literature on such operations~\cite{Gutoski09, schmid_postquantum_2021}, we denote the corresponding adapters by $\Upsilon^\texttt{LOSE} \in \Theta_\texttt{LOSE}$.

As an example of such an adapter, consider the case where Alice pre-processes the quantum state she obtains by means of a CPTP map $\Lambda_{\Ads_I}$, then performs her CP map $M_{A'}^i$, and post-processes the resulting state by means of a CPTP map $\Gamma_{\Ads_O}$ before sending it forward. Such processing would change Alice's and Bob's respective probabilities, but would not constitute any additional signalling (besides the one given by $W$) between both parties (see Fig.~\ref{fig::mini_combs}). 
With this, if Bob also pre- and post-processes (using the CPTP 
maps $\Lambda_{\Bds_I}$ and $\Gamma_{\Bds_O}$), the joint probabilities for measurements in their respective laboratories are given by:
\begin{gather}
\begin{split}
    &\Pprob(i,j|\Jcal_{A'},\Jcal_{B'}) \\
    &= W_{AB} \star (\Lambda_{\Ads_I} \star M_{A'}^i \star \Gamma_{\Ads_O}) \\
    &\phantom{=} \ \star (\Lambda_{\Bds_I} \star M_{B'}^j \star \Gamma_{\Bds_O}) \\
    &= (W_{AB} \star \Lambda_{\Ads_I} \star  \Gamma_{\Ads_O}  \star \Lambda_{\Bds_I} \star \Gamma_{\Bds_O}) \\
    &\phantom{=} \ \star (M_{A'}^i \star M_{B'}^j) \\
    &=: (W_{AB} \star \Upsilon^\texttt{LOSE}_{\Ads_I \Ads_O \Bds_I \Bds_O}) \star (M_{A'}^i \star M_{B'}^j)\, ,
    \end{split}
        \label{eqn::AdapterDef}
\end{gather}
where we have defined the adapter $\Upsilon^\texttt{LOSE}_{\Ads_I \Ads_O \Bds_I \Bds_O} \in \Bcal(\Hcal_{\Ads_I} \otimes \Hcal_{\Ads_O} \otimes \Hcal_{\Bds_I} \otimes \Hcal_{\Bds_O})$ which encapsulates how the original $W_{AB}\in \Bcal(\Hcal_A \otimes \Hcal_B)$ gets transformed when Alice and Bob perform their local pre- and post-processing. Such local operations (and the adapters $\Upsilon^\texttt{LOSE}_{\Ads_I \Ads_O \Bds_I \Bds_O}$ they implement) have been considered in Ref.~\cite{araujo_witnessing_2015} as a set of operations under which causal robustness (a measure of causal non-separability, see Sec.~\ref{sec::Montones_CausalRobust}) is monotone. As above, here, we have employed the flexibility of the link product; while the pre- and post-processing operations $\Lambda$ and $\Gamma$ can be understood as transforming the CP maps $M_{A'}^i$ and $M_{B'}^j$, respectively, they can also be understood as a transformation of $W$, leading to the above definition of the adapter $\Upsilon^\texttt{LOSE}$. 

Additionally, Alice and Bob can perform more general operations without signalling to each other. For example, instead of independent pre- and post-processing operations, they could correlate their respective operations, both in a classical or quantum way, effectively implementing local combs $C_{\Ads_I\Ads_O}$ ($C_{\Bds_I\Bds_O}$) with causal ordering $A_I\prec A_I' \prec A_O' \prec A_O$ ($B_I\prec B_I' \prec B_O' \prec B_O$) (see Fig.~\ref{fig::mini_combs} for a graphical representation). 
\begin{figure}
    \centering
    \includegraphics[width=0.7\linewidth]{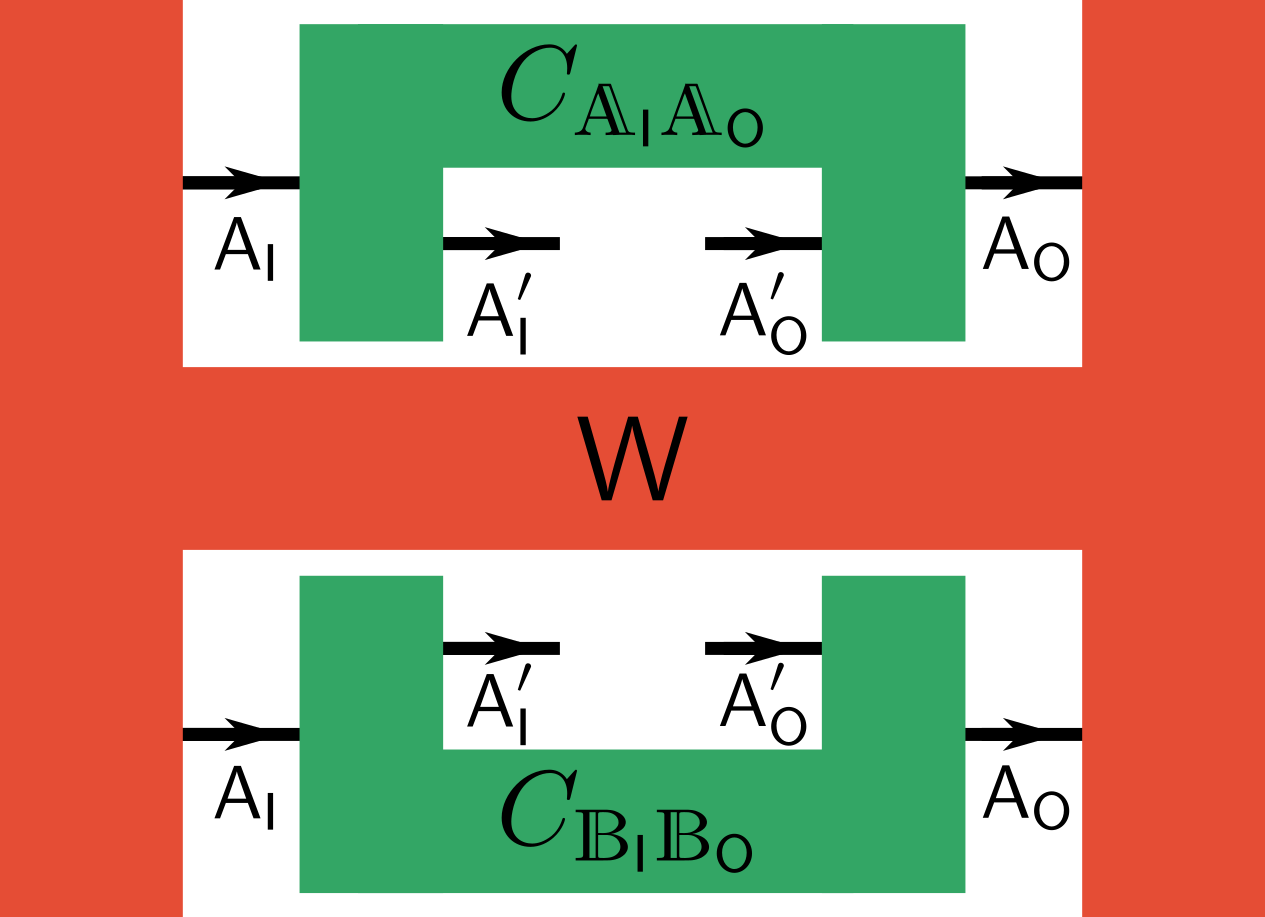}
    \caption{\textbf{Adapters from local combs.} Alice and Bob can effectively change the underlying process matrix by applying local, mutually independent combs. As a special, memoryless, case, this includes the scenario where Alice and Bob merely pre- and post-process the state they receive and feed forward, respectively, by means of independent maps $\Lambda_{\Ads_I}, \Gamma_{\Ads_O}, \Lambda_{\Bds_I}$, and $\Gamma_{\Bds_O}$.}
    \label{fig::mini_combs}
\end{figure}

Due to their causal ordering, these combs satisfy~\cite{chiribella_theoretical_2009}:
\begin{gather}
    \tr_{X_O}C_{\Xds_I\Xds_O} = \ident_{X_O'} \otimes  C_{\Xds_I}, \quad \tr_{X_{I}'}C_{\Xds_I} = \ident_{X_I}\, ,
\end{gather}
and $C_{\Xds_I\Xds_O} \geq 0$ for $X\in\{A,B\}$. Following the same reasoning that led to Eq.~\eqref{eqn::AdapterDef} we see that the corresponding adapter is given by 
\begin{gather}
\begin{split}
    \Upsilon^\texttt{LOSE}_{\Ads_I \Ads_O \Bds_I \Bds_O} &= C_{\Ads_I\Ads_O}  \star C_{\Bds_I\Bds_O} \\
    &= C_{\Ads_I\Ads_O}  \otimes C_{\Bds_I\Bds_O} \,.
\end{split}
\end{gather}
Finally, Alice and Bob could also share parts of an entangled state $\rho_{\widetilde A_I \widetilde A'_O \widetilde B_I \widetilde B'_O}$ which, together with local CPTP operations, they can use to implement an adapter (see Fig.~\ref{fig::Shared_initial} for a graphical representation). 
\begin{figure}
    \centering
    \includegraphics[width=0.9\linewidth]{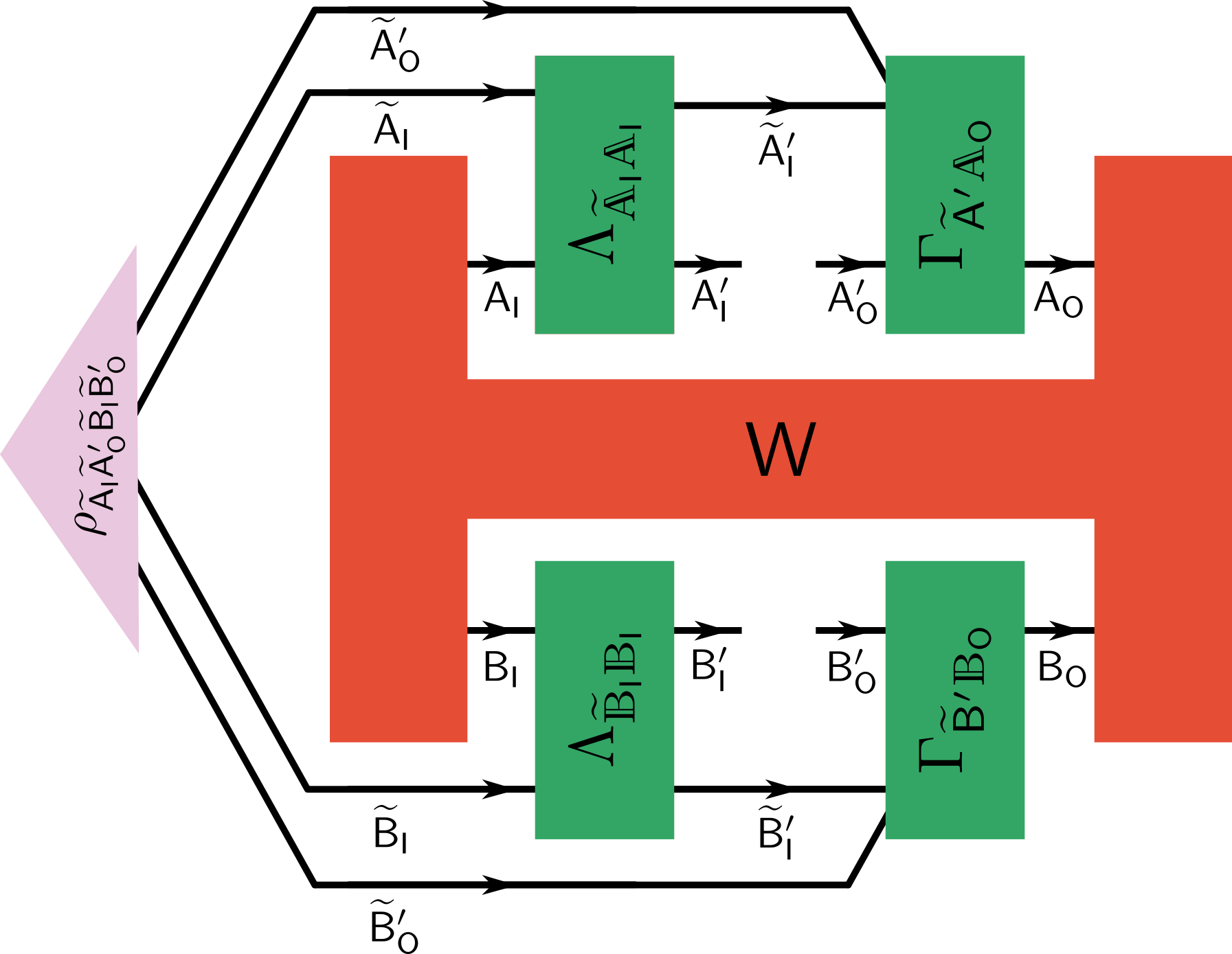}
    \caption{\textbf{LOSE Adapters.} Alice and Bob can share an entangled state and use it for local processing by means of the local maps $\Lambda_{\widetilde{\Xds}_I,\Xds_I}$ and $\Gamma_{\widetilde{X}'\Xds_{O}}$. The resulting adapters do not allow for signalling between the two parties and contain the adapters of Fig.~\ref{fig::mini_combs} as a special case. Note that the additional spaces $\widetilde X_O'$ could be absorbed -- here and in the main text -- into $\widetilde X_I$. The additional `lines' have been made explicit to emphasize that the entanglement can be shared between \textit{all} parts of the adapter.}
    \label{fig::Shared_initial}
\end{figure}
Concretely, the most general adapter $\Upsilon^\texttt{LOSE}_{\Ads_I\Ads_O\Bds_I\Bds_O}$ they could create in this way is of the form 
\begin{gather}
\begin{split}
    &\Upsilon^\texttt{LOSE}_{\Ads_I\Ads_O\Bds_I\Bds_O} \\
    &= \rho_{\widetilde A_I \widetilde A'_O \widetilde B_I \widetilde B'_O} \star \Lambda_{\widetilde \Ads_I \Ads_I} \star \Gamma_{\widetilde A' \Ads_O} \star \Lambda_{\widetilde \Bds_I \Bds_I} \star \Gamma_{\widetilde B' \Bds_O}
\end{split}
  \label{eqn::local_adapter}
\end{gather}

This motivates the definition of $\texttt{LOSE}$ adapters as those transformation that can be implemented by means of shared entanglement and local operations: 

\begin{definition}[LOSE adapters]
\label{def::Free_Adapters_prime}
An adapter from local operations and shared entanglement  $\Upsilon^{\textup{\texttt{LOSE}}} \in \Bcal(\Hcal_{\Ads_I} \otimes \Hcal_{\Ads_O} \otimes \Hcal_{\Bds_I} \otimes \Hcal_{\Bds_O})$ is a linear operator that can be written in the form of Eq.~\eqref{eqn::local_adapter}. The set of these free adapters is denoted $\Theta_{\textup{\texttt{LOSE}}}$. 
\end{definition}

This definition is in the spirit of local operations and ancillary entanglement provided in Ref.~\cite{taddei_quantum_2019}. Evidently, it can be extended to the multiparty case, without any added difficulties. Note that satisfaction of Eq.~\eqref{eqn::local_adapter} automatically implies $\Upsilon^\texttt{LOSE} \geq 0$, as it is the link product of positive semidefinite operators, and it guarantees that $\Upsilon^{\texttt{LOSE}}$ is properly normalized.

Any such adapter forbids signalling between Alice and Bob, and therefore it automatically satisfies Requirement R$3$, i.e., $\Theta_\texttt{LOSE} \subseteq \Theta_{\texttt{NS}}$. As a consequence, it also  preserves causal order and satisfies Requirements R$1$ and R$2$  (see App.~\ref{app::pre_shared} for a rigorous proof of these statements). 

However, Requirement R$4$ is more restrictive than Requirement R$3$, as there are non-signalling adapters that cannot be implemented by local operations and shared entanglement. To show that the inclusion $\Theta_\texttt{LOSE} \subset \Theta_\texttt{NS}$ is strict, we employ a result of Ref.~\cite{beckman_causal_2001}, where the existence of non-signalling channels that cannot be implemented by means of shared entanglement and local operations was shown (expressed in the nomenclature of Ref.~\cite{beckman_causal_2001}, there are bipartite channels that are \textit{causal} but not \textit{localizable}). Using two such maps $\Ecal: \Bcal(\Hcal_{A_I} \otimes \Hcal_{B_I}) \rightarrow \Bcal(\Hcal_{A_I'} \otimes \Hcal_{B_I'})$ and $\Fcal: \Bcal(\Hcal_{A_O'} \otimes \Hcal_{B_O'}) \rightarrow \Bcal(\Hcal_{A_O} \otimes \Hcal_{B_O})$ with corresponding Choi matrices $E_{\Ads_I\Bds_I}$ and $F_{\Ads_O\Bds_O}$, we can construct a non-signalling adapter 
\begin{gather}
   \Upsilon^\texttt{NS}_{\Ads_I\Ads_O\Bds_I\Bds_O} =  E_{\Ads_I\Bds_I} \otimes F_{\Ads_O\Bds_O}. 
\end{gather}
The fact that this indeed defines a non-signalling adapter can be seen by direct insertion into Eq.~\eqref{eqn::defFree} and using the fact that $E$ and $F$ are the Choi matrices of non-signalling channels. If this adapter were of the form of Def.~\ref{def::Free_Adapters_prime}, then we would have 
\begin{gather}
\label{eqn::ConjLOSE}
\begin{split}
    &E_{\Ads_I\Bds_I} \otimes F_{\Ads_O\Bds_O}  \\
    &= \rho_{\widetilde A_I \widetilde A'_O \widetilde B_I \widetilde B'_O} \star \Lambda_{\widetilde \Ads_I \Ads_I} \star \Gamma_{\widetilde A' \Ads_O} \star \Lambda_{\widetilde \Bds_I \Bds_I} \star \Gamma_{\widetilde B' \Bds_O}\, .
    \end{split}
\end{gather}
Now, tracing out the degrees of freedom ${\Ads_O\Bds_O}$ (i.e., tracing out $F_{\Ads_O\Bds_O}$) and using the properties of the maps on the RHS of the above equation, we obtain
\begin{gather}
\label{eqn::LOSE}
    E_{\Ads_I\Bds_I} = \rho_{\widetilde{A}_I\widetilde{B}_I} \star \Lambda_{A_I\widetilde{A}_I A_I'} \star \Lambda_{B_I\widetilde{B}_I B_I'}\,,
\end{gather}
where $\Lambda_{X_I\widetilde{X}_I X_I'}$ is the Choi matrix of a CPTP map $\Lcal: \Bcal(\Hcal_{\widetilde{X}_I} \otimes \Hcal_{X_I}) \rightarrow \Bcal(\Hcal_{{X}'_I})$. However, Eq.~\eqref{eqn::LOSE} would provide a decomposition of $E_{\Ads_I\Bds_I}$ in terms of local channels and shared entanglement (which, by assumption, it does not admit), thus implying that the adapter $\Upsilon^\texttt{NS}_{\Ads_I\Ads_O\Bds_I\Bds_O} =  E_{\Ads_I\Bds_I} \otimes F_{\Ads_O\Bds_O}$ cannot possess a decomposition of the form of Eq.~\eqref{eqn::ConjLOSE}. Consequently, we have shown the full hierarchy of strict inclusions $\Theta_\texttt{LOSE} \subset \Theta_{\texttt{NS}} \subset \Theta_{{\texttt{A}\texttt{FP}}}$ (see Fig.~\ref{fig::Adapters_sets}).

\begin{figure}
    \centering
    \includegraphics[width=0.9\linewidth]{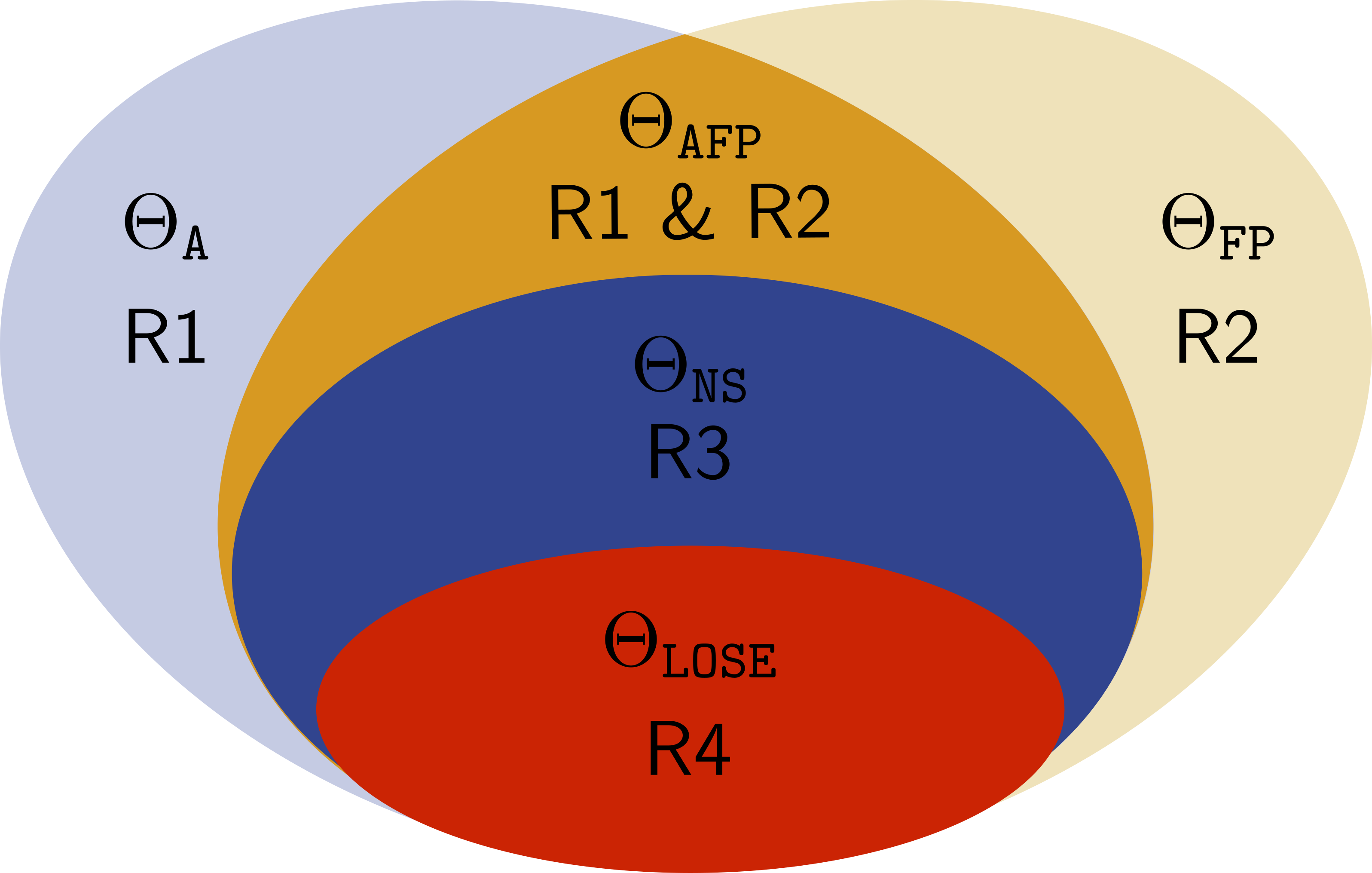}
    \caption{\textbf{Sets of adapters.} The individual sets of adapters follow from the Requirements R$1$ -- R$4$. Requirement R$1$ leads to the set $\Theta_{\texttt{A}}$ of admissible adapters (Def.~\ref{def::LegalAdapters}); Requirement R$2$ leads to the set $\Theta_\texttt{FP}$ of free-preserving adapters (Def.~\ref{def::FreePreservAdapters}), of which the operators inside the beige area are not proper transformations of process matrices; the combination of both Requirements R$1$ and R$2$ leads to the set $\Theta_{\texttt{A}\texttt{FP}}$ of admissible and free-preserving adapters (Def.~\ref{def::LegalFreePreserv}); Requirement R$3$ leads to the set $\Theta_\texttt{NS}$ of non-signalling adapters (Def.~\ref{def::free_adapters}), which coincides with the set $\Theta_{\texttt{CA}}$ of completely admissible adapters (Prop.~\ref{thm::comp_leg_ad}); and Requirement R$4$ leads to the set $\Theta_\texttt{LOSE}$ of adapters from local operations and shared entanglement (Def.~\ref{def::Free_Adapters_prime}).}
    \label{fig::Adapters_sets}
\end{figure}

While operationally clear-cut, the definition of adapters from local operations and shared entanglement is, however, too narrow to lend itself to a simple characterization. Given a positive semidefinite matrix $\Upsilon$, it is unclear how to show whether it is of the form of Eq.~\eqref{eqn::local_adapter} (already for channels, the membership to the set of \texttt{LOSE} operations is NP hard~\cite{Gutoski09}). This is reminiscent of the difficulty to characterize, e.g., LOCC operations~\cite{chitambar_everything_2014} in the resource theory of entanglement, or the problem of characterizing the set of entanglement-breaking supermaps that are obtained from entanglement-breaking pre- and post-processing~\cite{chen_entanglement-breaking_2019}. In either of these cases, the original set of transformations is extended to the set of completely resource non-generating operations (separable maps~\cite{bennett_quantum_1999} in the former and entanglement-breaking supermaps~\cite{chen_entanglement-breaking_2019} in the latter case). Here, as we have seen, there are at least two distinct possibilities of extending the set $\Theta_{\texttt{LOSE}}$ to a more mathematically amenable set of free operations, namely $\Theta_\texttt{NS}$ and $\Theta_\texttt{AFP}$.

\subsubsection{Defining free transformations}

The analysis of potential free operations for a resource theory of causal connection yields (at least) three distinct sets of free adapters, $\Theta_\texttt{LOSE} \subset \Theta_{\texttt{NS}} \subset \Theta_{{\texttt{A}\texttt{FP}}}$, each corresponding to a different level of control that the involved parties have over the adapter at hand. Both of the sets $\Theta_\texttt{LOSE}$ and $\Theta_{\texttt{NS}}$ respect Requirements R$1$, R$2$ and R$3$: they map non-signalling maps to non-signalling maps (and hence proper processes to proper processes), they map free processes to free processes, and they do not allow for signalling between the parties that share them. Additionally, they both preserve causal order. In particular these latter two properties -- both of which are not necessarily satisfied by adapters $\Upsilon^{{\texttt{A}\texttt{FP}}} \in \Theta_{{\texttt{A}\texttt{FP}}}$ -- are properties one would expect from free transformations in a resource theory of causal connection (see Tbl.~\ref{tab::CollProp} for a collation of the properties of different sets of adapters). Hence, the adapters in $\Theta_{{\texttt{A}\texttt{FP}}}$, while mathematically clear-cut, seem overly powerful. Consequently, in the following, we will predominantly focus on the resource theories of causal connection based on the free adapters $\Theta_\texttt{NS}$ and $\Theta_{\texttt{LOSE}}$. Whenever applicable, we will be explicit about whether or not properties we derive are also valid for a resource theory based on $\Theta_{{\texttt{A}\texttt{FP}}}$.

The sets $\Theta_\texttt{LOSE}$ and $\Theta_{\texttt{NS}}$ we focus on from now on differ in the following aspect: on the one hand, the set of \texttt{NS} adapters possesses a simpler mathematical characterisation, while not necessarily having a straightforward implementation in terms of non-signalling resources; on the other hand, adapters from local operations and shared entanglement are equipped with a simple physical implementation but have a cumbersome mathematical characterization.

From now on, our default choice of free transformations will be  the non-signalling adapters (see Def.~\ref{def::free_adapters}).  That is,  we set
\begin{gather}
    \Theta_\texttt{Free}\coloneqq\Theta_\texttt{NS}.
\end{gather}
  Whenever free adapters are mentioned -- unless explicitly stated otherwise -- it will be understood that we are talking about  non-signalling adapters and in order to emphasize the non-signalling property of the free adapters, we will continue to denote the set of free adapters by $\Theta_\texttt{NS}$ instead of $\Theta_\texttt{Free}$.

\begin{table}[t]
    \centering
    \begin{tabular}{l|c|c|c|c|c}
    & \thead{Pres. \\ \texttt{Proc}} & \thead{Pres. \\ \texttt{Free}} & \thead{Pres. \\ \texttt{Sep}} & \thead{non- \\sign.} & \thead{no sign. \\ required} \\
    \hline
     $\Theta_{{\texttt{A}}}$ & yes & no & no & no & no\\
     $\Theta_{\texttt{FP}}$ & no & yes & no & no & no\\
     $\Theta_{{\texttt{A}\texttt{FP}}}$ & yes & yes & no & no & no\\
     $\Theta_{\texttt{NS}}$ & yes & yes & yes & yes & no\\
     $\Theta_{\texttt{LOSE}}$ & yes & yes & yes &yes & yes\\
    \end{tabular}
    \caption{Sets of Adapters and their properties. The $\Theta_{{\texttt{A}}}$ and $\Theta_\texttt{FP}$ constitute the sets of adapters that preserve $\texttt{Proc}$ and $\texttt{Free}$, respectively. Adapters in $\Theta_{{\texttt{A}\texttt{FP}}}$, $\Theta_{\texttt{NS}}$ (which coincides with $\Theta_{\texttt{CA}}$), and $\Theta_{\texttt{LOSE}}$ preserve both $\texttt{Proc}$ and $\texttt{Free}$ and can thus meaningfully be considered as free. However, adapters in $\Theta_{{\texttt{A}\texttt{FP}}}$ could be used for signalling purposes and can map causally separable to causally non-separable processes. Adapters in $\Theta_{\texttt{LOSE}}$, the smallest set of free processes we consider, do not even require signalling resources for their implementation.}
    \label{tab::CollProp}
\end{table}

\subsection{Completely admissible adapters}
\label{sec::comp_Leg}
Up to this point, when considering adapters $\Upsilon_{\Ads_I\Ads_O\Bds_I\Bds_O}$ we have always considered their action on process matrices $W_{AB}$ (or, equivalently, on non-signalling channels $M_{A'B'}$). In principle, though, an adapter $\Upsilon_{\Ads_I\Ads_O\Bds_I\Bds_O}$ could also act on a non-signalling channel $\bar M_{A'\bar A B'\bar B}$ (with $A_I'\bar A_I \not \rightarrow B_O'\bar B_O$ and $B_I'\bar B_I \not \rightarrow A_O'\bar A_O$), i.e., act non-trivially only on parts of it (here, the degrees of freedom denoted by primed labels), and trivially on the rest (here, the degrees of freedom denoted by labels with a bar). A natural requirement on adapters is then that they are \textit{completely admissible}, i.e., they map non-signalling channels to non-signalling channels, even when acting non-trivially only on parts of them. Somewhat surprisingly, this additional requirement already removes all of the peculiarities we have encountered with respect to adapters in $\Theta_\texttt{A}$ and $\Theta_{\texttt{AFP}}$, and, as we show in this section, the set of completely admissible adapters coincides with $\Theta_\texttt{NS}$, providing yet another reason to choose the latter as the set of free adapters. 

To show this, let us first define completely admissible adapters:  \begin{definition}[Completely admissible adapters]
\label{def::CompLeg}
A completely admissible adapter is a positive semidefinite matrix $\Upsilon^{\textup{\texttt{CA}}} \in \Bcal(\Hcal_{\Ads_I} \otimes \Hcal_{\Ads_O} \otimes \Hcal_{\Bds_I} \otimes \Hcal_{\Bds_O})$ that satisfies 
\begin{gather}
    \Upsilon^{\textup{\texttt{CA}}}_{\Ads_I\Ads_O\Bds_I\Bds_O} \star \bar M_{A'\bar A B'\bar B}
\end{gather}
is non-signalling $A_I\bar A_I \not \rightarrow B_O\bar B_O$ and $B_I\bar B_I \not \rightarrow A_O\bar A_O$ for all non-signalling channels $\bar M_{A'\bar A B'\bar B}$ ($A_I'\bar A_I \not \rightarrow B_O'\bar B_O$ and $B_I'\bar B_I \not \rightarrow A_O'\bar A_O$) and arbitrary additional spaces $\{\bar A_I, \bar A_O, \bar B_I, \bar B_O\}$. The set of completely admissible adapters is denoted $\Theta_{\textup{\texttt{CA}}}$.
\end{definition}
Naturally, the above definition is equivalent to demanding that a completely admissible adapter $\Upsilon^{\texttt{CA}}_{\Ads_I\Ads_O\Bds_I\Bds_O}$ maps proper process matrices to proper process matrices even when only acting non-trivially on a part of them. 

It is easy to see that Def.~\ref{def::CompLeg} directly excludes adapters like the two-swap $\Upsilon^{2SW} \in \Theta_{\texttt{AFP}}$, which we discussed in Sec.~\ref{sec::legal_free_pres} as an example of an admissible and free preserving adapter that can change the causal order of a process matrix. More generally, it turns out that \textit{any} adapter that allows for communication between Alice and Bob maps some non-signalling map to a signalling map when only acting non-trivially on a part of it. Since the set of non-signalling adapters is exactly $\Theta_{\texttt{NS}}$, we have the following Proposition: 
\begin{proposition}
\label{thm::comp_leg_ad}
For any choice of spaces $\{\Ads_I, \Ads_O, \Bds_I, \Bds_O\}$ the corresponding sets of completely admissible adapters and non-signalling adapters coincide, i.e., 
\begin{gather}
    \Theta_{\textup{\texttt{CA}}} =  \Theta_{\textup{\texttt{NS}}} \, .
\end{gather}
\end{proposition}
The proof of this theorem can be found in App.~\ref{app::CV_NS}. Importantly, it provides a more fundamental interpretation of non-signalling adapters beyond the considerations of the resource theory of causal connection; they correspond exactly to the largest set of adapters that is admissible as soon the adapters can also non-trivially act on additional degrees of freedom. 

Naturally, the question arises, if, instead of considering the set of completely admissible adapters, one could have also considered the set of \textit{completely free preserving} adapters instead, to arrive at a similar result. However, it is easy to see that all free preserving adapters are already \textit{completely} free preserving: any matrix $\Upsilon^{\texttt{FP}}_{\Ads_I\Ads_O\Bds_I\Bds_O} \geq 0$ for which $\Upsilon^{\texttt{FP}}_{\Ads_I\Ads_O\Bds_I\Bds_O} \star (\rho_{A_IB_I} \otimes \ident_{A_OB_O}) = \rho_{A_I'B_I'} \otimes \ident_{A_O'B_O'}$ holds for all quantum states $\rho_{A_IB_I}$ directly satisfies $\Upsilon^{\texttt{FP}}_{\Ads_I\Ads_O\Bds_I\Bds_O} \star (\rho_{A_I\bar A_IB_I\bar B_I} \otimes \ident_{A_O\bar A_OB_O\bar B_O}) = \rho_{A_I'\bar A_IB_I'\bar B_I} \otimes \ident_{A_O'\bar A_OB_O'\bar B_O}$ for all quantum states $\rho_{A_I\bar A_IB_I\bar B_I}$. Consequently, requiring adapters to be \textit{completely} free preserving does not add any additional constraints on the adapters in $\Theta_{\texttt{FP}}$. With these final remarks on the structure of adapters out of the way, we now return to the set of free process matrices and show that it already follows entirely from the properties of the set $\Theta_\texttt{NS}$ of free adapters that we chose.

\subsection{Free objects revisited}\label{sec::FreeStRev}

So far it seemed natural to define the free process matrices as those  of the form $\rho_{A_IB_I}\otimes \ident_{A_OB_O}$ since they constitute the set of non-signalling process matrices. However, except for the derivation of $\Theta_{{\texttt{A}\texttt{FP}}}$, the set of free process matrices did not crucially enter our consideration of free adapters, since both $\Theta_{\texttt{NS}}$ and $\Theta_{\texttt{LOSE}}$ are derived from considerations of the (non-) signalling properties of the adapters instead of the free processes themselves.  

In principle, all free process matrices should be obtainable from free adapters via contraction of the excessive degrees of freedom. This property does indeed hold for all the sets $\Theta_{\texttt{A}\texttt{FP}}$, $\Theta_\texttt{NS}$, and $\Theta_\texttt{LOSE}$.  To prove this fact, it is enough to show that every free process matrix can be generated from adapters in the most restrictive set $\Theta_\texttt{LOSE}$.  

To this purpose, note that $\rho_{A_I'B_I'} \otimes \ident_{A_I\Ads_O B_I \Bds_O}/d_{A_O} d_{B_O}$ is a \texttt{LOSE} adapter for arbitrary quantum states $\rho_{A_I'B_I'}$. 
  In particular, this holds even if the input of the adapter is trivial, that is, if  $A_I$,  $A_O$, $B_I$,  and $B_O$ are one-dimensional systems.  In this case, the \texttt{LOSE} adapter simply generates the free process matrix  $\rho_{A_I'  B_I'} \otimes  \ident_{A_O'  B_O'}$.    

The above  observation suggests an alternative way to motivate the resource theory of causal connection, starting from the free operations instead of starting from the free objects.   The argument is as follows:   (1) the operations  (free or not)  of the resource theory are admissible adapters, that is, valid transformations of process matrices,  (2) the free operations are those that cannot be used to signal from Alice to Bob, or vice-versa,   that is, the adapters satisfying Requirement R$3$, (3) the free objects are the free   transformations with trivial input systems $A_I$,  $A_O$, $B_I$,  and $B_O$. This alternative argument leads to the resource theory of causal connection with the non-signalling adapters as free operations. Naturally, the same argument can be employed for the set of \texttt{LOSE} adapters (yielding the same set $\texttt{Free}$ of free process matrices), leading to a resource theory of causal connection with \texttt{LOSE} adapters as the free transformations and non-signalling process matrices as the free objects.

\section{Measures of causal connection}\label{sec::Montones}
Having the set $\texttt{Free}$ of free objects and the set $\Theta_{\texttt{NS}}$ of free transformations at hand, we can now introduce a measure of resourcefulness for process matrices in terms of their \textit{generalized robustness}, i.e., their robustness against `worst-case' general noise~\cite{steiner_generalized_2003}. Specifically, we analyse two kinds of generalized robustness: with respect to the set $\texttt{Free}$ of free (or non-signalling) processes -- called generalized \textit{signalling} robustness $\Rcal_s(W)$ -- and with respect to the set $\texttt{Sep}$ of causally separable processes -- called generalized \textit{causal} robustness $\Rcal_c$. This latter robustness is a natural resource monotone for the resource theory of causal non-separability (see Sec.~\ref{sec::ResCausNonSep}). Both of these robustness measures will provide us with monotones of the resource theory of causal connection and allow us to analyse the interconvertibility of process matrices under free adapters. 

We start with a discussion of the signalling robustness and its physical interpretation, followed by an investigation of the interconvertibility of causally indefinite processes.

\subsection{Robustness of signalling}\label{sec::Montones_ParallelRobust}

A generally used concept for the resourcefulness of an element of a resource theory is that of robustness against worst-case noise. Here, we follow this program and introduce the \textit{(generalized) robustness of signalling}, or simply signalling robustness, $\Rcal_s(W)$, which measures the maximal robustness of a process matrix $W$ under worst-case general mixing with respect to the set of free objects (i.e., the non-signalling process matrices in $\texttt{Free}$). We thus define $\Rcal_s(W)$ of a process $W$ as 
\begin{gather}
\label{eqn::ParRobust}
\Rcal_s(W) = \underset{T \in \texttt{Proc}} \min \left\{ s\geq 0 \left| \frac{W + sT}{1+s} = C \in \texttt{Free} \right.\right\}\, ,
\end{gather}
where $T\in \texttt{Proc}$ is a proper process matrix. Notice that, while we mostly work in the two-party scenario, this definition is valid for any number of parties. The signalling robustness is a faithful measure of causal connection -- clearly, $\Rcal_s(W) = 0$ for all $W\in \texttt{Free}$ and $\Rcal_s(W) > 0$ for all $W\notin \texttt{Free}$. 

Additionally, $\Rcal_s(W)$ is a convex function on $\texttt{Proc}$. While this property is not necessary for a resource measure it sure is desirable. Convexity can be proven in the same vein as the analogous proof for the robustness of coherence~\cite{napoli_robustness_2016} (see App.~\ref{app:conv} for an explicit proof). 

More importantly, $\Rcal_s(W)$ is non-increasing under the free adapters $\Theta_{\texttt{NS}}$, making it a monotone of the resource theory of causal connection. For the signalling robustness to be non-increasing under the action of an adapter $\Upsilon \in \Theta_{\texttt{NS}}$, it is sufficient for $\Upsilon$ to map free process matrices to free process matrices, a fact we have already shown above. In detail, let a process matrix $W$ have an optimal [with respect to the signalling robustness, see Eq.~\eqref{eqn::ParRobust}] decomposition $W = (1+\Rcal_s(W))C^* - \Rcal_s(W)T^*$, where $C^*\in \texttt{Free}$ and $T^*\in \texttt{Proc}$. Applying the adapter $\Upsilon$ to $W$ yields
\begin{gather}
\begin{split}
    &\Upsilon \star W \\
    &= (1+\Rcal_s(W))(\Upsilon \star C^*) - \Rcal_s(\Upsilon \star T^*) \\
    &=: (1+\Rcal_s(W))C' - \Rcal_s(W)T',
    \end{split}
    \label{eqn::decomp}
\end{gather}
where $C'\in \texttt{Free}$ and $T'\in \texttt{Proc}$. Since Eq.~\eqref{eqn::decomp} is a valid -- but potentially not optimal -- decomposition for $\Upsilon \star W$, this implies that its signalling robustness is at most $\Rcal_s(W)$. Hence,
\begin{gather}
    \Rcal_s(\Upsilon \star W)\leq\Rcal_s(W).
\end{gather}

Since the above proof only relies on the invariance of $\texttt{Free}$ under the considered adapters, $\Rcal_s$ is also a monotone under the action of adapters in $\Theta_{\texttt{A}\texttt{FP}}$. Being a monotone of causal connection, $\Rcal_s(W)$ provides a necessary condition for the interconvertibility of process matrices by means of free adapters; if $\Rcal_s(W') > \Rcal_s(W)$, then there is no free adapter $\Upsilon$ such that $W' = \Upsilon \star W$. As we shall see, unsurprisingly, this condition is not sufficient. For comprehensiveness, using recent results presented in Ref.~\cite{gour_dynamical_res_2020}, in App.~\ref{app::Montones_allmonotones} we derive \textit{all} monotones of the resource theory of causal connection, i.e., all functions $f(W_{AB})$ that are non-increasing under free adapters. This, in turn, provides -- in principle -- an unambiguous way to decide whether or not a process matrix can be transformed into another by means of free transformations. However, the resulting monotones are rather abstract, and we will not use them in what follows to decide interconvertibility between processes. 

Conveniently, analogously to other robustness measures in the literature~\cite{napoli_robustness_2016, piani_robustness_2016, uola_quantifying_2019, uola_quantification_2020} the computation of $\Rcal_s(W)$ can be phrased as a semidefinite program (SDP). Na\"ively, for two parties, we obtain $\Rcal_s(W)$ as the solution of the optimization 
\begin{align}
\begin{split}
 \textbf{given} \ \ & W \\ 
 \textbf{minimize} \ \ & s \\ 
 \textbf{subject to} \ \ & \frac{W + sT}{1+s} = \rho_{A_IB_I} \otimes \ident_{A_OB_O}\\
 & L_V(T) = T\\
 & \tr(T) = d_{A_O}d_{B_O}\\
 & \tr(\rho_{A_IB_I}) = 1\\
 & T\geq 0, \ \ \rho_{A_IB_I} \geq 0, \ \ s\geq 0, 
\end{split}
\label{eqn::primalRob}
\end{align}
which is not an SDP per se. However, by setting $\widetilde{T} := sT$ and $\widetilde{\rho}_{A_IB_I}:=(1+s)\rho_{A_IB_I}$ and making the substitution $\widetilde{T}=\widetilde{\rho}_{A_IB_I} \otimes \ident_{A_OB_O} - W$, the above problem can be stated as the following SDP:
\begin{align}
\begin{split}\label{sdp::parallel_primal}
 \textbf{given} \ \ & W \\ 
 \textbf{min} \ \  & \tr(\widetilde \rho_{A_IB_I}) - 1 \\ 
 \textbf{s.t.} \ \  &  \widetilde \rho_{A_IB_I}\otimes \ident_{A_OB_O} - W \geq 0 \\
 & \widetilde \rho_{A_IB_I} \geq 0.
\end{split}
\end{align}
This generalizes easily to more than two parties. We can use the above SDP to obtain a better operational understanding of the signalling robustness via the corresponding dual problem. As it is easy to see that the above primal problem is strictly feasible (choose, e.g., $\widetilde \rho_{A_IB_I} = \lambda \ident_{A_IB_I}$ and $\lambda > \lambda_{\max}$, where $\lambda_{\max}$ is the maximal eigenvalue of $W$), this dual provides the same value for the signalling robustness as the original SDP. By assigning the dual variable $S$ to the first inequality constraint of the primal, we obtain $\Rcal_s(W)$ as the solution of the dual problem
\begin{align}
\begin{split}\label{sdp::parallel_dual}
 \textbf{given} \ \ & W \\ 
 \textbf{max} \ \ & \tr(WS) - 1 \\ 
 \textbf{s.t.} \ \ &  S \geq 0 \\
 & \ident-\tr_{A_OB_O}(S)\geq 0.
\end{split}
\end{align}
On the one hand then, the SDP~\eqref{sdp::parallel_primal} and its dual provide the numerical tools to evaluate the signalling robustness of some interesting process matrices (and we will use them throughout to compute numerical values for the signalling robustness). On the other hand, the dual problem provides us with a concrete physical interpretation of the signalling robustness and connects it to the concept of witnesses of causal connection.

\subsection{Operational interpretation of the signalling robustness}

\begin{figure}[t]
    \centering
    \includegraphics[width=0.95\linewidth]{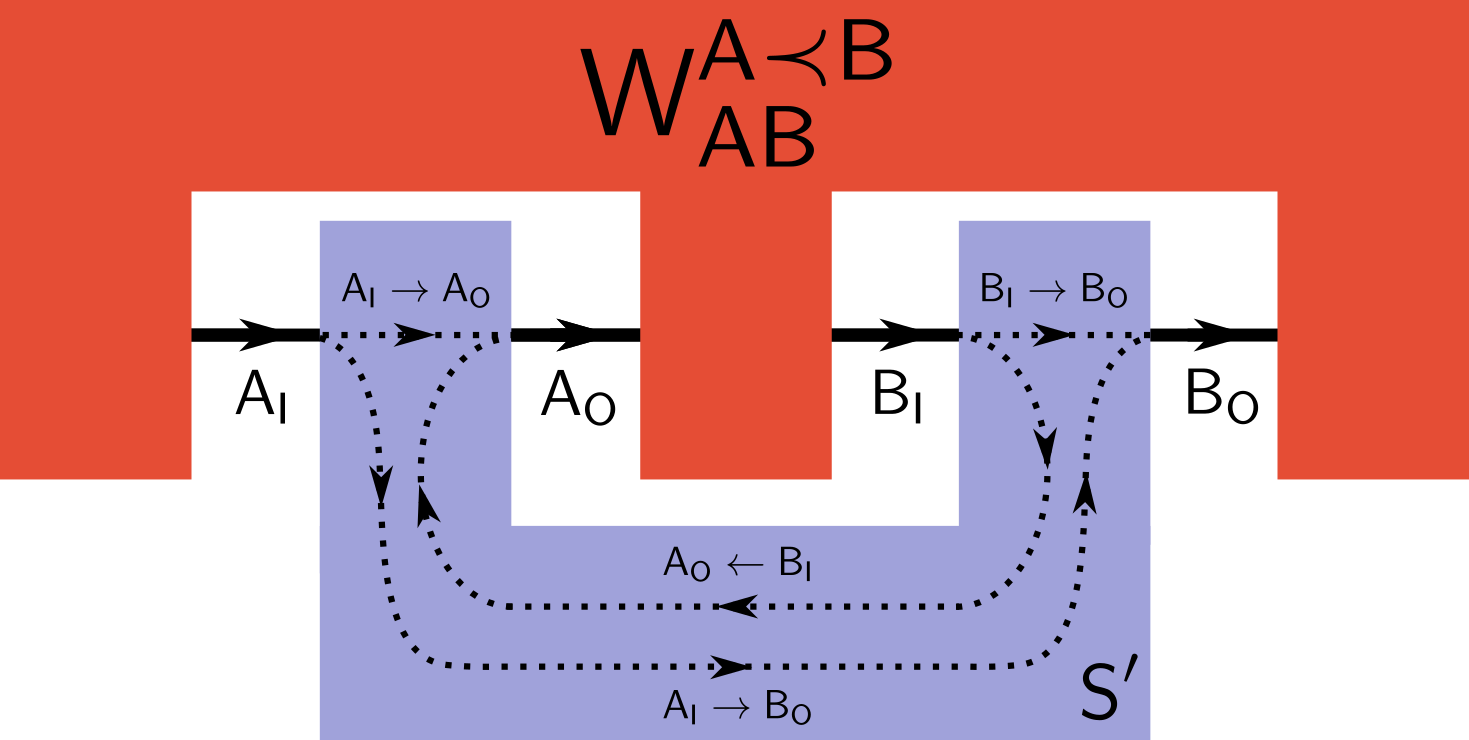}
    \caption{\textbf{Contracting an $A\prec B$ comb with $S'$.} A channel with corresponding Choi matrix $S'$ generally allows for signalling from $B_I$ to $A_I$, depicted by the dotted loop. Concatenating such a channel with, for example, a comb $W^{A\prec B}_{AB}$ with causal ordering $A\prec B$ then closes causal loops. The remaining dotted lines depict the other possible communication lines $S'$ allows for.}
    \label{fig::S_action}
\end{figure}
Below, we will see that the signalling robustness can be interpreted in terms of witnesses for causal connection. Besides this somewhat technical interpretation, the dual problem~\eqref{sdp::parallel_dual} of the SDP~\eqref{sdp::parallel_primal} offers a much more direct and operational interpretation of $\Rcal_s(W)$.

To see this, first note that the second condition in the dual SDP~\eqref{sdp::parallel_dual} can be replaced by $\tr_{A_OB_O}S = \ident$ without loss of generality. Indeed, for any $S \geq 0$ that maximizes $\tr(WS)-1$ and satisfies $\ident - \tr_{A_{O}B_{O}}S = D \geq 0$, we can define $S' = S+ D\otimes \frac{\ident_{A_OB_O}}{d_{A_O}d_{B_O}}$, which satisfies $S' \geq 0$ and $\tr_{A_{O}B_{O}}S' = \ident$, as well as $\tr(WS') \geq \tr(WS)$. The property $\tr_{A_OB_O}S' = \ident$ together with $S'\geq 0$ implies that $S'$ is the Choi matrix of a general -- potentially two-way signalling -- channel from $A_IB_I$ to $A_OB_O$. 

Generally, `plugging' such a channel into $W$, i.e., computing $\tr(WS')  = W\star S^{\prime\mathrm{T}}$ leads to causal loops. To see this, we note that the property $L_V(W) = W$ implies ${}_{A_IB_I}W = {}_{A_IA_OB_IB_O}W = \ident/(d_{A_I}d_{B_I})$, which means that $\tr_{A_IB_I}W = \ident_{A_OB_O}$, making $W$ the Choi matrix of a channel from $A_OB_O$ to $A_IB_I$. In general then, concatenating a channel from $A_OB_O$ to $A_IB_I$ (with corresponding Choi matrix $W$) with a channel from $A_IB_I$ to $A_OB_O$ creates causal loops (see Fig.~\ref{fig::S_action} for a graphical depiction). 

In this sense, $\tr(WS') - 1$ is a quantifier for the `amount' of causal loops that would be closed in $W$ if one could contract it with a general $A_IB_I \rightarrow A_OB_O$ channel $S'$. Importantly, this intuitive interpretation is compatible with our previous considerations; if $W\in \texttt{Free}$ then $\tr(WS') - 1 = 0$ and no causal loops can be closed. 

\subsection{Witnesses of signalling robustness and transformations under free adapters}

In addition to providing an intuitive interpretation of the signalling robustness, the above dual problem and, in particular, the new variable $S$ are directly related to \textit{witnesses} of causal connection. Here, we clarify this connection. The set of free process matrices is convex, and as such, there exist Hermitian matrices $\widetilde S$ (called witnesses) that satisfy $\tr(\widetilde SW) \geq 0$ for all $W = \rho_{A_IB_I} \otimes \ident_{A_OB_O}$ and $\tr(\widetilde SW) < 0$ for at least one $W\notin \texttt{Free}$. This allows us to characterize the set of witnesses of causal connection. Since
\begin{gather}
    \tr[\widetilde S (\rho_{A_IB_I}\otimes \ident_{A_OB_O})] = \tr(\rho_{A_IB_I}\tr_{A_OB_O}\widetilde S )
\end{gather}
has to be non-negative for all states $\rho_{A_IB_I}$, we see that a Hermitian matrix $\widetilde S$ can only be a witness of causal connection if $\tr_{A_OB_O}\widetilde{S} \geq 0$. 

Additionally, for $\widetilde S$ to be a proper witness, there should be at least one  process matrix $W\notin \texttt{Free}$ such that $\tr(W\widetilde S) < 0$, implying that generally $\widetilde S \ngeq 0$. For any witness $\widetilde S$, the matrix $\lambda \widetilde S$ is also a witness for $\lambda \geq 0$, implying that $-\tr(\widetilde S W)$, which can be considered a measure for causal connection, is in principle unbounded. To obtain meaningful numerical values, the set of witnesses needs to be restricted in a way that does not limit the set of causally connected process matrices $W$ that can be `detected' by the witnesses. The respective restriction is a priori somewhat arbitrary. Here, we choose it such that we obtain a connection to the above dual problem~\eqref{sdp::parallel_dual}. That is, we impose the restriction that all witnesses we consider are of the form 
\begin{gather}
    \label{eqn::decompWitness}
    \widetilde{S} = \ident/(d_{A_O}d_{B_O}) - S,
\end{gather}
where $S\geq 0$ [which coincides with the second condition of the dual~\eqref{sdp::parallel_dual}]. As $\tr_{A_OB_O}\widetilde{S} \geq 0$, this restriction then implies $\ident - \tr_{A_OB_O}S \geq 0$ (which coincides with the first condition of the dual problem). A priori, the above requirement~\eqref{eqn::decompWitness} excludes many potential witnesses and thus might restrict the set of causally connected processes that can be detected. However, for \textit{any} witness $\widetilde S$ we can find a $\lambda \geq 0$ and a matrix $S\geq 0$ such that $\lambda \widetilde S = \ident/(d_{A_OB_O}) - S$. As $\lambda \widetilde S$ is still a witness (for $\lambda \geq 0$) and detects the same set of causally connected processes as $\widetilde S$, only allowing for witnesses of the form $\ident/(d_{A_O}d_{B_O}) - S$ is thus merely a rescaling but does not restrict the range of detectable processes. 

Since $\widetilde S$ is a witness, we have $-\tr(\widetilde SW) \leq 0$ for all $W = \rho_{A_IB_I}\otimes \ident_{A_OB_O}$. This upper bound is tight; for any $\Lambda_{AB} \geq 0$ with $\tr_{A_OB_O}\Lambda_{AB} = \ident_{A_{I}B_{I}}$, $\widetilde S = \ident/(d_{A_O}d_{B_O}) - \Lambda_{AB}$ is a witness that satisfies Eq.~\eqref{eqn::decompWitness}. For this witness then, we have $\tr[\widetilde S(\rho_{A_IB_I}\otimes \ident_{A_OB_O})] = 0$. 

Importantly, with this restriction on the witnesses $\widetilde{S}$, we see that the second line in the dual problem~\eqref{sdp::parallel_dual} corresponds to 
\begin{align}
    \tr(WS) - 1 &= \tr[W(\ident/(d_{A_O}d_{B_O}) - \widetilde S)] - 1 \nonumber \\
    &= -\tr(W\widetilde S), 
\end{align}
implying that the dual problem maximizes the measure $-\tr(W\widetilde S)$, which, as we have seen, unsurprisingly yields $0$ on the set of free processes. We thus see that the dual problem~\eqref{sdp::parallel_dual} which yields the signalling robustness of a given process matrix $W$ can equivalently be understood as an optimization of the measure $-\tr(\widetilde S W)$ of causal connection over the set of witnesses of causal connection.

With these two interpretations of the signalling robutstness at hand, we will now bound this monotone both for the general two-party case, as well as special multi-party cases and use these bounds to discuss the notion of most resourceful processes. 

\section{Bounds on the signalling robustness and most resourceful processes} \label{sec::BoundsRes}

\subsection{Maximal signalling robustness on two parties}\label{subsec::max_par_robust}

An important question in any resource theory is that of the existence of a \textit{most valuable} resource, i.e., an object that allows one to obtain all other objects via free transformations. While such an object does generally not exist, in our case one can ask the related (but strictly weaker question) of an upper bound on the signalling robustness, and whether this bound is tight (at least in many relevant cases). Here we show that $\Rcal_s(W) \leq d_{\bar O}^2 - 1$ for all two-party process matrices $W$, where $d_{\bar O}$ is the maximum of the output dimensions of the two parties (Alice and Bob), and this bound is tight. 

As it turns out, a process $W^{A\rightarrow B}$ that satisfies this bound is given by a causally ordered process that consists of an initial state preparation (say, in Alice's laboratory), an identity channel between Alice and Bob, and a final discarding of Bob's output (see Sec.~\ref{sec::MostRes}). While it is somewhat self-evident that an identity channel between two parties has high causal connection, it is nonetheless surprising that neither additional memory, nor causal indefiniteness allow for larger amounts of causal connection (in the two-party scenario). Similar results hold in the multi-party scenario, suggesting that, even beyond the case of two parties, neither memory nor causal indefiniteness lead to improved causal connection. Here, we first comprehensively discuss the two-party case and then provide partial results for the multi-party scenario.

The proof that in the two-party case $\Rcal_s(W)$ is upper bounded by $d_{\bar O}^2 - 1$ has two steps. First, we show that 
\begin{gather}
\label{eqn::Cruc_Ineq}
    d_{\bar O}^2\cdot{}_{A_OB_O}W - W \geq 0
\end{gather}
for all proper two-party process matrices $W$. Then, using the dual SDP~\eqref{sdp::parallel_dual}, we see that 
\begin{gather}
\begin{split}
    \Rcal_s(W) &= \max(\tr(WS)) - 1 \\
    &\leq  d_{\bar O}^2\max(\tr({}_{A_OB_O}W S)) - 1 \leq d_{\bar O}^2 - 1\, , 
\end{split}
\end{gather}
where we have used Eq.~\eqref{eqn::Cruc_Ineq}, the self-duality of the operator ${}_{A_OB_O}\sbt$, as well as the fact that $\tr_{A_OB_O}(S) \leq \ident$ and $\tr(W) = d_{A_OB_O}$. Consequently, showing that Eq.~\eqref{eqn::Cruc_Ineq} holds provides the upper bound we aim for. 
\begin{lemma}
\label{lem::UpBound}
For any proper two-party process matrix $W$, 
\begin{gather}
    d_{\bar O}^2\cdot{}_{A_OB_O}W - W \geq 0,
\end{gather}
where $d_{\bar O}\coloneqq\max(d_{A_O},d_{B_O})$, holds.
\end{lemma}

\begin{proof}
For the proof, we first note that the map $\Dcal[W] = d\tr(W)\ident - W$ is CP, where $d$ is the dimension of the space $W$ lives on. Indeed, it is easy to see that the Choi matrix $\eta_\Dcal$ of $\Dcal$ is given by 
\begin{gather}
    \eta_\Dcal = d\ident \otimes \ident - \Phi^+\, .
\end{gather}
Since $\tr(\Phi^+) = d$, we see that the above matrix $\eta_\Dcal$ is positive semidefinite, and $\Dcal$ thus a CP map. Consequently, applying $\Dcal$ to only a part of a positive matrix still yields a positive output, i.e., 
\begin{gather}
\begin{split}
\Dcal^{(X_O)}[W] &= d_{X_O} \tr_{X_O}W \otimes \ident_{X_O} - W \\
&=  d_{X_O}^2 \cdot{}_{X_O}W - W\geq 0
\end{split}
\end{gather}
holds for all positive semidefinite matrices $W$ (note that this relation has been proven independently in~\cite{bavaresco_semi-device-independent_2019}, but we provide an explicit proof for self-containedness). Exchanging $d_{X_O}$ for $d_{\bar O} = \max_{X=A,B}\{d_{X_O}\}$ does not change the positivity of the map, such that 
\begin{gather}
    \bar{\Dcal}^{(X_O)}[W] := d_{\bar{O}}^2\cdot {}_{X_O}W  - W \geq 0.
\end{gather}
Now, applying $\bar{\Dcal}^{(A_O)} \circ \bar{\Dcal}^{(B_O)}$ to a positive matrix $W$, we see that 
\begin{gather}
\begin{split}
\label{eqn::doubleDepol}
    &\bar{\Dcal}^{(A_O)} \circ \bar{\Dcal}^{(B_O)}[W] \\
    &= d_{\bar O}^4 \cdot{}_{A_OB_O}W - d_{\bar O}^2 ({}_{A_O}W  + {}_{B_O}W)  + W \geq 0\, .
\end{split}
\end{gather}
Up to this point, we have not yet employed the fact that $W$ is a proper process matrix, and the above equation thus holds for all positive matrices $W$. Now, using that for proper process matrices $L_V(W) = W$ holds [see Eq.~\eqref{eqn::DefProcMat}], it is easy to see that we have ${}_{A_O}W  + {}_{B_O}W = W + {}_{A_OB_O}W$~\cite{araujo_witnessing_2015}. With this, Eq.~\eqref{eqn::doubleDepol} reads
\begin{gather}
    (d_{\bar O}^4 - d_{\bar O}^2) {}_{A_OB_O}W - (d_{\bar O}^2-1) W \geq 0\, .
\end{gather}
Since the dimension $d_{\bar O}$ is always assumed to be at least $2$, this equation directly implies the assertion of the Lemma. 
\end{proof}
As outlined above, this Lemma then leads to the desired bound on the signalling robustness for two-party matrices. 
\begin{proposition}
\label{prop::Upbound}
For any two-party process matrix $W$, we have 
\begin{gather}
    \Rcal_s(W) \leq d_{\bar O}^2 -1\, ,
\end{gather}
where $d_{\bar O} = \max(d_{A_O},d_{B_O})$.
\end{proposition}
This bound can be achieved by a causally ordered  comb  if  certain  dimensional  conditions are satisfied. For example, without loss of generality, let $d_{A_O} \geq d_{B_O}$. Then, if $d_{A_O} = d_{B_I}$, the bound $d_{\bar O}^2 -1 = d_{A_O}^2 -1 $ is saturated by a process of the form 
\begin{gather}
\label{eqn::MaxVal}
    W^{A \rightarrow B} \coloneqq \frac{\ident_{A_I}}{d_{A_I}}\otimes \Phi_{A_OB_I}^+ \otimes \ident_{B_O},
\end{gather} 
where $\Phi_{A_OB_I}^+ = \sum_{ij} \ketbra{ii}{jj}$ is the Choi matrix of the identity channel. Choosing the proper witness $S = \ident_{A_I} \otimes \Phi_{A_OB_I}^+ \otimes \frac{\ident_{B_O}}{d_{B_O}}$, we see that $\tr(W^{A \rightarrow B}S) - 1 = d_{A_O}^2 -1$. 

Importantly, while this bound holds for general two-party processes, it is only tight for scenarios, where the dimension of the largest output space ($A_O$ above) is at most as large as that of the other party's input space ($B_I$ above). Intuitively, this is due to the fact that there is no identity channel from $A_O$ to $B_I$ when $d_{A_O} > d_{B_I}$, and the causal connection between the parties is thus bounded by the upper bound we gave in Prop.~\ref{prop::Upbound} but not necessarily tight anymore (the influence of the different dimensions of the involved spaces on the tightness of the bounds we provide is discussed in more detail in Sec.~\ref{subsec::robust_multipartite}).

We emphasize that the fact that $W$ is a proper process matrix (and not just a positive semidefinite matrix which satisfies $\tr(W) = d_{A_O} d_{B_O}$) is crucial for the derivation of the above bounds. For positive semidefinite matrices $W$ with $\tr(W) = d_{A_O} d_{B_O}$, it is easy to see that the term $\tr(WS)$ can achieve its algebraic maximum of $d_{A_I} d_{A_O}d_{B_I} d_{B_O}$ and analogously for more parties. This, then, also implies that Lem.~\ref{lem::UpBound} holds for proper process matrices, but not for general positive semidefinite matrices with $\tr(W) = d_{A_O}d_{B_O}$.

\subsection{Most resourceful process on two parties}
\label{sec::MostRes}

Let us, for the moment, assume that all dimensions of $W$ are the same, i.e., $d_{A_I} = d_{A_O} = d_{B_I} = d_{B_O}$. As we have seen, in this case, a process with the maximal signalling robustness is given by 
\begin{gather}\label{eqn::fullsigprocess}
    W^{A \rightarrow B} \coloneqq \frac{\ident_{A_I}}{d_{A_I}}\otimes \Phi_{A_OB_I}^+ \otimes \ident_{B_O},
\end{gather}
which we will call the \textit{fully signalling} process. The process $W^{B\rightarrow A}$ is defined analogously.

In principle, $W^{A \rightarrow B}$ could be a good candidate for the most resourceful process of the resource theory of causal connection, in the sense that, starting from $W^{A \rightarrow B}$, all other processes may be reachable by means of free adapters. This, however, is not the case for two distinct reasons. On the one hand, if $W'_{A'B'}$ is defined on spaces $\{A_I',A_O',B_I',B_O'\}$ with $d_{X_I'} > d_{X_I}$ and $d_{X_O'} > d_{X_O}$ for $X\in \{A,B\}$, then the signalling robustness of $W'_{A'B'}$ can exceed that of $W^{A \rightarrow B}$ (since $d_{\bar{O}'} > d_{\bar{O}}$), implying that $W^{A \rightarrow B}$ cannot be transformed to all $W'_{A'B'}$ by means of free transformations. 

On the other hand, even when focusing on transformations that do not change (or only decrease) the respective dimensions, a free adapter cannot change the causal ordering of the process matrix it acts on, such that $W^{A \rightarrow B}$ cannot be transformed to, for example, a process that is ordered $B\vec{\prec} A$, or a process that is causally non-separable. At best, then, $W^{A \rightarrow B}$ and $W^{B \rightarrow A}$ can be the most resourceful processes for all process matrices that are of ordering $A\prec B$ and $B\prec A$, respectively, and that do not exceed them with respect to the dimensions of the involved spaces (we will henceforth assume this latter property and not mention it explicitly anymore). It will turn out that this is indeed the case, both for two- and the three-party case. 

Interestingly, the causally separable process
\begin{gather}\label{eqn::Wmix}
    W^\text{mix}_q\coloneqq (1-q)\,W^{A\to B}+q\,W^{B\to A},
\end{gather}
which is a convex mixture of two fully signalling processes, one in each direction, also satisfies $\Rcal_s(W^\text{mix}_q)=d_{\bar O}^2-1$, which can be seen by setting $S = \Phi^+_{A_OB_I} \otimes \Phi^+_{B_OA_I}$ in the dual SDP~\eqref{sdp::parallel_dual}. 

Processes of this form do not have a definite causal order (they are causally separable, though), but since free adapters cannot create causal non-separability, they cannot be transformed to \textit{all} process matrices. More interestingly, as we will see, there are causally non-separable processes that have a lower signalling robustness than  $W^\text{mix}_q$, implying that the signalling robustness does not impose a total order in the set of all process matrices. 

Let us now show that, within the set of two-party ordered processes $W^{A\prec B}$, the above process $W^{A\rightarrow B}$ is indeed the most valuable in the sense that \textit{all} two-party ordered processes $W^{A\prec B}$ can be obtained from it. This can be easily seen by considering that every $W^{A\prec B}$ can be written as a concatenation of an initial pure state $\Psi_{A_IE}$ on $A_I$ and some appropriate ancillary system $E$ and a CPTP map $A_OE \rightarrow B_I$ with corresponding Choi matrix $\Gamma_{A_OEB_I}$, such that $W^{A\prec B}_{AB} = \Psi_{A_IE} \star \Gamma_{A_OEB_I} \star \ident_{B_O}$~\cite{chiribella_theoretical_2009}. It is straightforward to see that, by performing a local adapter in her laboratory, Alice can prepare all states $\Psi_{A_IE}$ and implement all maps $\Gamma_{A_OEB_I}$ which, when applying such adapters to $W^{A\rightarrow B}$, allows her to produce all possible two-party ordered processes $W^{A\prec B}$ (see Fig.~\ref{fig::Most_val_A_prec_B} for a graphical representation). We thus have the following proposition:
\begin{proposition}
\label{prop::TransfBipartite}
Every process matrix $W'_{X'Y'}$ with causal ordering $X'\prec Y'$ can be obtained from $W^{X\rightarrow Y}$ via free adapters if the dimensions of the spaces $W'_{X'Y'}$ is defined on do not exceed those of the corresponding spaces of $W^{X\rightarrow Y}$.
\end{proposition}
The same holds true in the three-party case, where, for example, the process matrix
\begin{gather}
    W^{A\rightarrow B \rightarrow C} = \eta_{A_I} \otimes \Phi^+_{A_OB_I} \otimes \Phi^+_{B_OC_I} \otimes \ident_{C_O}
\end{gather}
can be transformed to all processes of ordering $A\prec B \prec C$. We provide the proof in App.~\ref{app::tri_part_trans}, where we also present a conjecture for an ordered process that cannot be reached in the four-party case. Importantly, this latter four-party conjecture -- or any statement of this kind -- only make sense if we put restrictions on the respective dimensions; having identity channels on arbitrary dimensions at hand, will allow one to implement all causally ordered processes on (sufficiently) lower dimensions by simply using the additional dimensions as a memory carrier. This is similar to the analogous case in the resource theories of entanglement, coherence and purity, where all states on lower dimensions can be achieved from a single resourceful state of sufficiently large dimension~\cite{plbnio_introduction_2007, horodecki_quantum_2009, streltsov_colloquium_2017, gour_resource_2015}. 

\begin{figure}[t]
    \centering
    \includegraphics[width=0.95\linewidth]{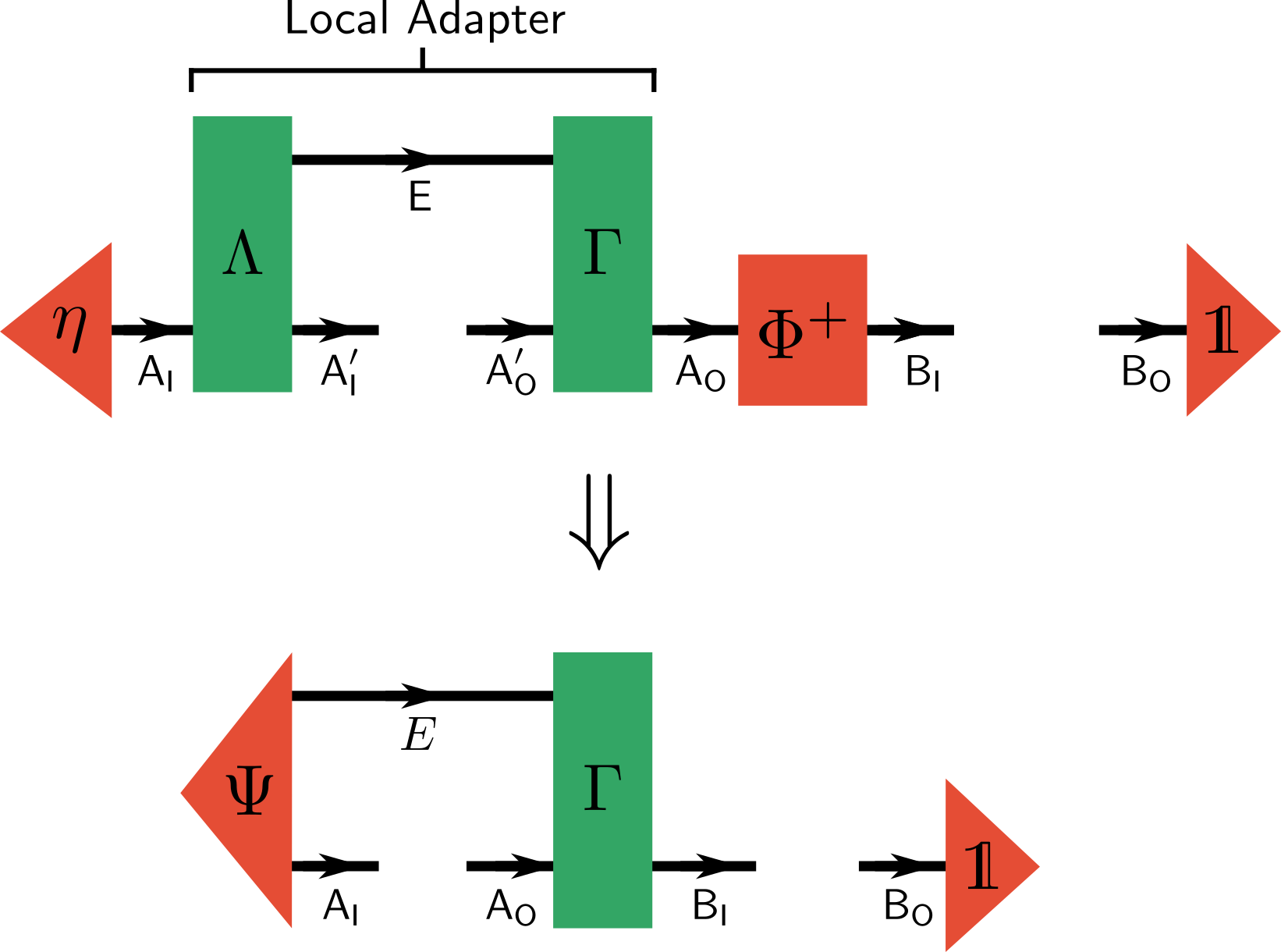}
    \caption{\textbf{Constructing all $A\prec B$ two-party combs from $W^{A \rightarrow B}$.} Any causally ordered comb $W^{A\prec B}$ can be written as a circuit with an initial system-environment state and a CPTP map that acts on the system \textit{and} the environment. By applying a local adapter, Alice can thus transform $W^{A\rightarrow B}$ to any process $W^{A\prec B}$ (with system dimensions smaller or equal than those of $W^{A\rightarrow B}$, that is). For convenience, we have dropped the primes in the system labeling in the bottom panel.}
    \label{fig::Most_val_A_prec_B}
\end{figure}

\subsection{Transformations of bipartite processes}\label{sec::TransfBipar}

As we have seen, some process matrices cannot be reached by transforming the most robust bipartite process with free adapters. 
The first example of this kind we discuss is the process $W^\text{OCB}_{A'B'}$ (named after Oreshkov, Costa and Brukner~\cite{OreshkovETAL2012}), a two-party, causally non-separable process matrix introduced in Ref.~\cite{OreshkovETAL2012} and defined as
\begin{gather}\label{eqn::w_ocb}
\begin{split}
    &W^\text{OCB} \\
    &= \frac{1}{4}\left[\ident_{AB} +\frac{1}{\sqrt{2}}\left(\sigma_{A_O}^z \sigma_{B_I}^z + \sigma_{A_I}^z \sigma_{B_I}^x  \sigma_{B_O}^z \right)\right]\, ,
\end{split}
\end{gather}
where $\sigma^z$ and $\sigma^x$ are Pauli matrices and we have omitted the respective identity matrices and tensor products. This causally non-separable process cannot be reached from $W^{A\rightarrow B}$ or $W^{B\rightarrow A}$, since free adapters cannot map causally ordered to causally non-separable processes. By using the SDP~\eqref{sdp::parallel_primal}, we find that the signalling robustness of  $W^\text{OCB}$ is $\Rcal_s(W^\text{OCB})=1$. Consequently, it is also impossible to transform $W^\text{OCB}$ to $W^{A\rightarrow B}$ or $W^{B\rightarrow A}$ by means of free adapters, since the latter two processes have -- when all systems are qubits -- a signalling robustness of $3$. Expectedly, there is no total order in the set of processes with respect to the free transformations $\Theta_{\texttt{NS}}$ (this was already clear from the fact that processes of order $A\vec{\prec} B$ cannot be transformed to processes of order $B\vec{\prec} A$ and vice versa by means of free adapters). 

While $W^\text{OCB}$ cannot be transformed to $W^{A\rightarrow B}$ by means of free transformations, there exist processes $W^{A\vec{\prec} B}$ with causal order $A\vec{\prec} B$ such that $W^{A\vec{\prec} B} = W^\text{OCB} \star \Upsilon$ for some free adapter $\Upsilon \in \Theta_{\texttt{NS}}$. More generally, in the two-party case, any causally non-separable process matrix $W$ can be `degraded' to a process of causal order $A\vec{\prec} B$ or $B\vec{\prec} A$, i.e., there exists a free adapter $\Upsilon$ such that $W \star \Upsilon$ is of causal order $A\vec{\prec} B$ or $B\vec{\prec} A$. This can be seen from the following, more general Proposition:
\begin{proposition}
\label{prop::Degrad1}
For any process matrix $W$ with $\Rcal_s(W) >0$ there exists a free adapter $\Upsilon$ such that $W \star \Upsilon$ has a causal ordering $A\vec{\prec} B$ or $B\vec{\prec} A$.
\end{proposition}
\begin{proof}
For the proof, note that $\Rcal_s(W) >0$ implies $W \neq {}_{A_OB_O} W$. Now, consider the free adapter $\Upsilon^{(X_O)}$ that traces out the output space $X_O$ and replaces it by a maximally mixed state, i.e., $\Upsilon^{(X_O)} \star W = {}_{X_O}W$. The resulting process is causally ordered $A\prec B$ if $X = B$ and $B\prec A$ if $X = A$. It remains to show that for at least one of these possibilities we do not have $A||B$. To see this, imagine that ${}_{X_O}W$ satisfies $A||B$ for all $X$. In this case we have 
\begin{gather}
    {}_{A_O}W = {}_{A_OB_O}W \quad \text{and} \quad  {}_{B_O}W = {}_{A_OB_O}W\, . 
\end{gather}
Inserting this into $L_V[W] = W$ yields $W = {}_{A_OB_O} W$ which contradicts the assumption $\Rcal_s(W) > 0$.
\end{proof}
Naturally, the above proposition is trivial if $W$ already has a definite causal order, but becomes more insightful for the case where $W$ corresponds to a mixture of causal orders or a causally non-separable process. 

Finally, one might wonder if there are processes that can be transformed into both definite causal orderings $A\vec \prec B$ \textit{and} $B\vec \prec A$ by means of free adapters. An answer to this question is provided by the following Proposition: 
\begin{proposition}
    Any causally non-separable process matrix on two parties can be transformed into a process matrix with ordering $A\vec \prec B$, as well as into one with ordering $B\vec \prec A$, by means of free adapters. 
\end{proposition}
Before proving the above statement, we emphasize that it is in line with Prop.~\ref{prop::Degrad1}, since we have $\Rcal_s(W) > 0$ for all causally non-separable process matrices, but goes beyond it, since now transformation to \textit{both} definitive causal orderings is possible. 
\begin{proof}
For the proof, consider a free trace-and-replace adapter that traces out the output of Alice's operation and feeds forward a fixed state $\ketbra{\Psi}{\Psi}^\mathrm{T}$ (where the additional transpose is for notational convenience), i.e., $W' = \Upsilon \star W = \braket{\Psi_{A_O}|W|\Psi_{A_O}} \otimes \ident_{A_O'}$. In what follows, for ease of notation, without risk of confusion, we identify $A_O$ and $A_O'$, thus replacing  $\ident_{A_O'}$ by $ \ident_{A_O}$. Since the corresponding adapter is a local operation in Alice's laboratory, it corresponds to a free adapter, and the resulting $W'$ is causally ordered $B\prec A$ (but potentially $B||A$ holds). Now, we show that assuming that $W'$ is ordered $B||A$ for all possible $\ket{\Psi}$ contradicts the assumption that the original $W$ was causally non-separable. 

To this end, we choose a basis $\{\ketbra{\Psi^{(i)}_{A_O}}{\Psi^{(i)}_{A_O}}\}_{i=1}^{d_{A_O}^2}$ of $\Bcal(\Hcal_{A_O})$ consisting of pure states. Assuming that $W'_i:= \braket{\Psi^{(i)}_{A_O}|W|\Psi^{(i)}_{A_O}} \otimes \ident_{A_O'}$ is causally ordered $B||A$ for all $i$ then implies 
\begin{gather}
    W_{i}' = \rho_{A_IB_I}^{(i)} \otimes \ident_{A_OB_O} \quad \forall i\, ,
\end{gather}
where $\{\rho_{A_IB_I}^{(i)}\}$ are quantum states. With this, we can express $W$ as 
\begin{gather}
\label{eqn::DecompDual}
    W = \sum_{i=1}^{d_{A_O}^2} \rho_{A_IB_I}^{(i)} \otimes \Delta_{A_O}^{(i)} \otimes \ident_{B_O}\, ,
\end{gather}
where $\{\Delta_{A_O}^{(i)}\}_{i=1}^{d_{A_O}^2}$ is the dual set to the basis $\{\ketbra{\Psi^{(i)}_{A_O}}{\Psi^{(i)}_{A_O}}\}_{i=1}^{d_{A_O}^2}$, in the sense that $\tr(\Delta_{A_O}^{(j)}\ketbra{\Psi_{A_O}^{(i)}}{\Psi_{A_O}^{(i)}}) = \delta_{ij}$. The validity of Eq.~\eqref{eqn::DecompDual} can be checked by direct insertion. Since the process matrix of Eq.~\eqref{eqn::DecompDual} factorizes into an identity matrix on $B_O$ and the rest, we see that it is causally ordered $A\prec B$, contradicting the assumption that $W$ is causally non-separable. Consequently, there must exist at least one $\ket{\Psi_{A_O}^{(i)}}$ for which $W_i'$ is causally ordered $B\vec \prec A$. Running the same argument for a trace-and-replace adapter on Bob's side then shows that it must also be possible to transform a causally non-separable process matrix $W$ to a process of causal order $A\vec \prec B$. 
\end{proof}
We note that, in the above proof, the requirement of causal non-separability can be relaxed. Indeed, it is sufficient if $W$ is a convex combination of causal orders but \textit{not} of the form $A\prec B$, $B\prec A$, or $A||B$. 

The two above propositions provide an a posteriori justification for the inclusion of causally indefinite processes as valid objects in a resource theory of causal connection; while free adapters do not allow the transitioning from one fixed causal order to another, causally indefinite processes in a sense `bridge the gap' between the two sets, since each causally indefinite process can be transformed to both $A\vec\prec B$ as well as $B\vec \prec A$ processes. This, then, both raises the question, whether there exists a most resourceful causally non-separable process, as well as whether there exists a causally indefinite process that allows one to reach \textit{all} causally ordered ones by means of free adapters. We return to these questions in Sec.~\ref{sec::Montones_CausalRobust}, where we introduce a second monotone of causal connection that is tailored to the distinction between causally separable and causally non-separable processes. First though, we extend some of the results we found above for the two-party case to the multi-party scenario.

\subsection{Signalling robustness of multipartite causally ordered processes}\label{subsec::robust_multipartite}

Some of the results on signalling robustness translate straightforwardly to the multi-party case. Here, we provide a couple of relevant cases for which a bound on the signalling robustness can be found. The somewhat technical proofs are relegated to App.~\ref{app::robust_multipartite}. 

For arbitrary many parties with causal order $A\prec B \prec C \prec \cdots$ a generalization of the fully signalling process $W^{A\to B}$ indeed maximizes $\Rcal_s(W^{A\prec B\prec C\prec \cdots})$. To see this, we require the following proposition:
\begin{proposition}
\label{prop::maxRordered}
Let $\textup{\texttt{Proc}}_{1:N}$ be the set of all causally ordered processes on $N$ parties with causal order $X^{(1)}\prec X^{(2)} \prec \cdots \prec X^{(N)}$. For any $W\in \textup{\texttt{Proc}}_{1:N}$ we have 
\begin{gather}
    \Rcal_s(W) \leq \prod_{i=1}^{N-1} d_{X^{(i)}_O}^{2} - 1:= d_{\bar{O}}^2-1,
\end{gather} 
where  $d_{X^{(i)}_O}$ is the output dimension of party $X^{(i)}_O$. 
\end{proposition}
We emphasize that here, unlike in the two-party case, we restrict our attention to the causally ordered case, such that there is no maximization over the output dimensions, but rather the bound on the signalling robustness only depends on all but the last output space, which is intuitively clear, since the last output space cannot signal to any of the other parties. The set of processes for which the above bound (or a variant thereof) holds can be slightly enlarged to the set of all processes with a definite last party and convex mixtures thereof (see App.~\ref{app::robust_multipartite}).

For the case where all involved dimensions are equal (the more general case is discussed in App.~\ref{app::robust_multipartite}), it is then straightforward to show the following corollary:  
\begin{corollary}
\label{cor::maxProcOrdered}
If all involved dimensions are equal, then the signalling robustness on the set $\textup{\texttt{Proc}}_{1:N}$ is maximized by the process $W^{X^{(1)} \rightarrow  \cdots \rightarrow X^{(N)}}$.
\end{corollary}

As we show in App.~\ref{app::tri_part_trans}, for three parties, $W^{X^{(1)} \rightarrow \cdots \rightarrow X^{(N)}}$ can also be transformed by means of free adapters to all other tripartite processes that are ordered in the same way, making it not only the most robust, but also the most resourceful process.

For the interpretation of the provided result, for simplicity, consider the case where all dimensions are equal. Then, it appears to be most advantageous from a causal connection perspective to perfectly pass information in one direction. Neither memory, nor convex mixture with different causal orders yield better causal connections. Upon numerical investigation, it appears like this bound can also not be beaten by causally non-separable processes in the tripartite case, and we conjecture that this holds for arbitrary numbers of parties. As mentioned, a more thorough investigation of the multipartite case can be found in App.~\ref{app::robust_multipartite}.

After this short discussion of the multi-party case, let us now return to the two-party scenario and analyse the interconvertibility of processes by means of free adapters in more detail, particularly within the set of causally non-separable processes, where it presents itself more layered than in the causally ordered one.

\section{Resource theory of causal non-separability}\label{sec::Interconvertability}
In the discussion of the resource theory of causal connection, causally non-separable processes formed the natural `bridge' between the sets of processes with different causal orders, and allowed for a discussion of the signalling structure of \textit{all} processes that quantum mechanics (in principle) allows for. Beyond this role for the investigation of signalling properties, causal non-separability can be considered a resource in its own right -- as evidenced, for example, by the fact that causal non-separability can allow for the violation of causal inequalities~\cite{OreshkovETAL2012, branciard_simplest_2015}, or outperform causally ordered processes when it comes to distinguishing non-signalling channels~\cite{chiribella_perfect_2012}. Somewhat unsurprisingly, the corresponding resource theory of causal non-separability is conceptually closely related to the resource theory of causal connection. Here, we first detail how the resource theory of causal non-separability can be developed following the same reasoning as above and show that, unlike in the case of causal connection, there is only \textit{one} meaningful set of free adapters when causal non-separability is the resource of interest. In order to quantify the resourcefulness of causal non-separability, we provide a corresponding resource monotone, the causal robustness (originally introduced in~\cite{araujo_witnessing_2015}), which also turns out to be a resource monotone for the resource theory of causal connection. Consequently, it will allow us to investigate the interconvertibility of causally non-separable processes under free adapters with respect to causal connection. 

\subsection{Uniqueness of the resource theory of causal non-separability}\label{sec::ResCausNonSep}

For simplicity, we limit our discussion of resource theories of causal non-separability to the two-party case. Naturally, the set of free objects in such a resource theories is the set \texttt{Sep} of causally separable processes. Following the same path we took for the resource theory of causal connection, the different sets of possible free transformations could be derived by imposing qualitatively different requirements on them. Such requirements on free operations\footnote{While mathematically they are the same type of objects, for better distinction, we use two different symbols, $\Xi$ and $\Upsilon$ for adapters in the resource theory of causal non-separability and that of causal connection, respectively.} $\Xi$ are:
\begin{itemize}
    \item[\textbf{R$1$'.}] To map process matrices to process matrices.
    \item[\textbf{R$2$'.}] To map causally separable to causally separable processes.
    \item[\textbf{R$3$'.}] To not allow for the creation of causal non-separability.
    \item[\textbf{R$4$'.}] To be implemented with only causally ordered resources between the parties.
\end{itemize}
Requirement R$1$' is the same as R$1$ and simply ensures that free operations are also admissible. Requirement R$4$' seems like the natural analogue of R$4$. As we have seen, there are adapters that are composed of signalling channels and still yield admissible adapters, one such example being $\Upsilon^{2SW}$.  However, unlike R$4$ for the resource theory of causal connection, R$4$' does not always lead to admissible adapters. To see this, consider an adapter $\Xi$ that is implemented by Alice and Bob by making use of a resource that allows for signalling from Bob to Alice. In many cases, the resulting adapter cannot be applied to a process matrix of causal ordering $A\prec B$ since causal loops would be closed. An example for such an adapter is $\Upsilon^{1SW} \notin \Theta_{{\texttt{A}}}$. In addition to potentially yielding `improper adapters', from a conceptual point of view, it is questionable why (and how), in a situation where Alice and Bob share a comb of a given causal order (say $A\prec B$), they should be able to signal in, say, the opposite direction when making an adapter. Intuitively then, the only types of adapters that Alice and Bob can implement without potentially running into logical contradiction seem to be those that do not require signalling, i.e., adapters in $\Theta_{\texttt{LOSE}}$. While this might be too restrictive -- as we mentioned, $\Upsilon^{2SW} \notin \Theta_{\texttt{LOSE}}$ would be an admissible adapter that can be implemented by means of signalling resources -- $\Theta_{\texttt{LOSE}}$ appears like the largest set of adapters that Alice and Bob can implement themselves that is not plagued with interpretational issues. Requirement R$4$' should thus be replaced by R$4$, implying that, at the operational level, where Alice and Bob are the ones to implement the free adapters, there is no fundamental distinction between the resource theory of causal non-separability and that of causal connection. 

Adapters that satisfy both requirements  R$1$' and R$2$' are the largest set of possible free transformations in a resource theory of causal non-separability, and we denote the corresponding set of admissible separability preserving adapters by $\Theta_{\texttt{ASP}}$. It is natural to compare this maximal set of free operation to $\Theta_{{\texttt{A}\texttt{FP}}}$, the corresponding maximal set of free operations for the resource theory of causal connection. As we have seen, there are free adapters $\Upsilon \in \Theta_{{\texttt{A}\texttt{FP}}}$ that map causally ordered processes to causally non-separable processes. Consequently, such adapters would not lie in $\Theta_{\texttt{ASP}}$. On the other hand, it is easy to construct adapters that satisfy $\Xi \in \Theta_\texttt{ASP}$ but $\Xi \notin \Theta_{\texttt{A}\texttt{FP}}$. An example for such an adapter is given by 
\begin{gather}
    \xi = \frac{1}{d_{A_O}d_{B_O}}\ident_{AB} \otimes W_{A'B'}\, ,
\end{gather}
where $W_{A'B'} \in \texttt{Sep}\setminus\texttt{Free}$. This adapter, when acting on a process matrix simply discards it and replaces it by a fixed, signalling process matrix (and thus does not lie in $\Theta_{\texttt{AFP}}$). With this, we see that, while they naturally share many common adapters, there is no clear hierarchy or inclusion relation between the two maximal sets of free adapters for the resource theory of causal connection and that of causal non-separability.

Requirement R$3$' appears like the natural analogy to the corresponding requirement R$3$ for the resource theory of causal connection. However, on its own, it does not even guarantee that the resulting set of adapters preserves $\texttt{Sep}$. Basically, R$3$' demands that, no matter what Alice and Bob do locally, when they have access to the adapter, they cannot use it to create a process matrix that is causally non-separable. Locally, the only operations they can perform to transform an adapter to a process matrix is to feed in a state and discard the final degrees of freedom, i.e., given an adapter $\Xi$, they can transform it to 
\begin{gather}
    W' = \rho_{A_I} \otimes \eta_{B_I} \otimes \ident_{A_OB_O} \star \Xi. 
\end{gather}
Now, since $\rho_{A_I} \otimes \eta_{B_I} \otimes \ident_{A_OB_O}$ corresponds to a non-signalling process, the requirement that $W' \in \texttt{SEP}$ is weaker than demanding that non-signalling processes are mapped to non-signalling processes, which, besides being admissible, is the restriction on adapters that lie in $\Theta_{{\texttt{A}\texttt{FP}}}$. Since adapters in $\Theta_{{\texttt{A}\texttt{FP}}}$ can create causally non-separable process matrices when acting on a causally separable one, this means that the set of free adapters emerging from requirement R$3$' is too big to be a meaningful set of free adapters for the resource theory of causal non-separability. On the other hand, combining requirements R$2$' and R$3$' (and R$1$') to restrict ourselves to adapters that preserve $\texttt{Sep}$ does not introduce anything new, since the resulting set of free adapters simply coincides with $\Theta_{\texttt{ASP}}$ (i.e., the set already obtained from the combination of R$1$' and R$2$'): Since $\rho_{A_I} \otimes \eta_{B_I} \otimes \ident_{A_OB_O}$ is a process matrix in $\texttt{Sep}$, satisfaction of R$2$' ensures that it cannot be mapped to a causally non-separable one, thus already implying satisfaction of requirement R$3$'. Consequently, while R$3$ yields an interesting and physically meaningful set of free adapters, its na{\"i}ve translation R$3$' for the resource theory of causal non-separability does not lead to additional restrictions/a meaningful new theory. 

We conclude thus, given that R$4$ coincides with R$4$', there seems to only be a \textit{single} meaningful resource theory of causal non-separability, namely the one where the free transformations are given by $\Theta_{\texttt{ASP}}$. This statement notwithstanding, one could, in principle, introduce other sets of adapters that preserve $\texttt{Sep}$. For example, $\Theta_{\texttt{NS}}$, the set of all non-signalling adapters preserves $\texttt{Sep}$, since all such adapters preserve causal order. However, there is no additional operational motivation to restrict oneself to this set of adapters as a natural set of free adapters in the resource of causal non-separability. 

It remains to briefly comment on the structure of adapters $\Xi \in \Theta_\texttt{ASP}$. As we have seen above, for the resource theory of causal connection, it is possible to phrase the properties of free adapters in terms of linear constraints (besides positivity that is). Mathematically, this is due to the fact that the non-signalling requirement on free process matrices, i.e., that they are of the form $\rho_{A_IB_I} \otimes \ident_{A_OB_O}$ is -- besides the positivity constraint -- a linear one. On the other hand, this fails to hold for the condition of causal separability; linear combinations of causally separable processes can well be causally non-separable. This, in turn, makes a concrete characterization of the adapters in the set $\Theta_\texttt{ASP}$ more elusive and the question of whether a given adapter lies in $\Theta_\texttt{ASP}$ a hard one to answer in general. 

Nonetheless, the resource theory of causal non-separability can nicely be phrased within the framework for the resource theory of causal connection that we set up above. Here, we do not aim to provide a \textit{full} account of it, but merely want to emphasize the structural similarities with the resource theory of causal connection. 

Additionally, in the same vein as above, it is natural to introduce a robustness measure with respect to the set $\texttt{Sep}$, the causal robustness, as a monotone for the resource theory of causal non-separability. While this monotone is not faithful for causal connection, as it turns out, it is nonetheless monotone under adapters in $\Theta_{\texttt{NS}}$ and $\Theta_{\texttt{LOSE}}$. Introducing the causal robustness thus kills two birds with one stone -- it completes our brief discussion of the resource theory of causal non-separability and provides us with a second monotone for the resource theory of causal connection, allowing, for example, for the investigation of the interconvertibility of processes in $\texttt{Proc}\setminus \texttt{Sep}$ under adapters $\Upsilon \in \Theta_{\texttt{NS}}$.

\subsection{Robustness of causal non-separability}\label{sec::Montones_CausalRobust}

We have seen in the previous sections that the stratification of the space of process matrices with respect to free adapters of the resource theory of causal connection presents itself as somewhat layered. In particular, there is no total order, i.e., $W \overset{\Theta_{\texttt{NS}}}{\nrightarrow} W'$ does not generally imply $W' \overset{\Theta_{\texttt{NS}}}{\rightarrow} W$. Concretely, we have seen this for pairs of processes, where either both processes were causally ordered, or one was causally ordered, while the other was not. However, we have not yet considered the case where both processes are causally non-separable. 

Since the previous results were obtained based on the monotonicity of the robustness of signalling under free adapters, they hold independent of whether $\Theta_{\texttt{LOSE}}$, $\Theta_{\texttt{NS}}$, or $\Theta_{\texttt{A}\texttt{FP}}$ is considered as the set of free adapters. This will fail to hold true for the results of the present section, where we require preservation of $\texttt{Sep}$ for the derivation of our results. 

By providing a second monotone for the resource theory of causal connection, we show that the lack of total order also holds within the set of causally non-separable processes. To this end, we consider the \textit{(generalized) robustness of causal non-separability} $\Rcal_c$, inspired by analogous measures for the robustness of entanglement~\cite{steiner_generalized_2003} and coherence~\cite{napoli_robustness_2016}, and already introduced in Ref.~\cite{araujo_witnessing_2015} for process matrices (following Ref.~\cite{araujo_witnessing_2015}, we will simply call $\Rcal_c$ `causal robustness'). We define $\Rcal_c(W)$ of a process matrix $W$ as the maximal robustness under worst-case general mixing, now with respect to the set of causally separable process matrices \texttt{Sep} (instead of the set \texttt{Free}), \textit{i.e.}, 
\begin{gather}
\label{eqn::Robust}
\Rcal_c(W) = \underset{T \in \texttt{Proc}} \min \left\{ s\geq 0 \left| \frac{W + sT}{1+s} = R \in \texttt{Sep} \right.\right\}\, ,
\end{gather}
where $T$ is proper process matrix and \texttt{Sep} is the set of causally separable process matrices.\footnote{While the definition of the set \texttt{Sep} is straightforward in the bipartite case, it becomes more subtle in the multipartite setting~\cite{araujo_witnessing_2015,oreshkov_causal_2016, wechs_definition_2019}. Here, we predominantly focus on the two-party case.} Since free processes in the resource theory of causal non-separability preserve $\texttt{Sep}$, the causal robustness is a monotone of this resource theory. More importantly for our purposes, it is also a monotone for two of the resource theories of causal connection we introduced above (those based on $\Theta_{\texttt{NS}}$ and $\Theta_{\texttt{LOSE}}$), and we now use it to investigate how causally non-separable processes can be converted into one another by means of adapters in $\Theta_{\texttt{NS}}$. 

\subsection{Conversion of causally non-separable processes under free transformations}

Analogously to the signalling robustness, one can show that the causal robustness is a convex function and non-increasing under free adapters (i.e., adapters in $\Theta_{\texttt{NS}}$). This latter point follows directly from the fact that free adapters preserve causal ordering (this fails to hold for adapters in $\Theta_{\texttt{AFP}}$). 

However, the causal robustness is not a faithful measure for causal connection, since $\Rcal_c(W)=0$ for all $W\in\texttt{Sep}$, which is a strict superset of \texttt{Free}. Nonetheless, the fact that it is non-increasing under free adapters allows one to show that there are pairs $\{W,W'\}$ of causally non-separable process matrices for which $W \overset{\Theta_{\texttt{NS}}}{\nrightarrow} W'$ and $W' \overset{\Theta_{\texttt{NS}}}{\nrightarrow} W$ holds, i.e., there is no total order in the set of causally non-separable processes with respect to free adapters in $\Theta_\texttt{NS}$.

As for the signalling robustness, the causal robustness of a process matrix can be phrased as an SDP and equivalently be calculated by its primal and dual formulation. Its primal SDP problem can be derived starting from the definition of causal robustness in Eq.~\eqref{eqn::Robust} and following similar steps as for the signalling robustness in Eqs.~\eqref{eqn::primalRob} and~\eqref{sdp::parallel_primal}, to arrive at
\begin{align}
\begin{split}\label{sdp::causalprimal}
    \textbf{given} \ \ & W \\ 
    \textbf{min} \ \   & \frac{1}{d_O}\tr(\widetilde R)-1\\ 
    \textbf{s.t.} \ \  & \widetilde R-W\geq 0 \\ 
                        & \widetilde R \geq 0 \\
                       & \widetilde R \propto R \in\texttt{Sep}.
\end{split}
\end{align}

The dual problem associated to the above SDP reads
\begin{align}
\begin{split}\label{sdp::causaldual}
    \textbf{given} \ \ & W \\ 
    \textbf{max} \ \   & \tr(WS)-1\\ 
    \textbf{s.t.} \ \  & S\geq 0 \\ 
                       & \frac{1}{d_O}\ident + U - L_{A\prec B}(U) \geq S \\
                       & \frac{1}{d_O}\ident + V - L_{B\prec A}(V) \geq S,
\end{split}
\end{align}
where $L_{X\prec Y}(W)\coloneqq {}_{Y_O}W - _{Y_IY_O}W + _{X_OY_IY_O}W$ is the projector onto the subspace of processes with causal order $X \prec Y$. Analogously as for the robustness of signalling, $S$ can be interpreted as a witness of causal non-separability.

Evaluating either this SDP or its dual, the causal robustness of $W^\text{OCB}$ [see Eq.~\eqref{eqn::w_ocb}] can be computed to be $\Rcal_c(W^\text{OCB})=0.1716$, while the causal robustness of $W^{A\to B}$ is $\Rcal_c(W^{A\to B})=0$. Hence, this pair of process matrices provides us with an interesting example where
\begin{align}
\Rcal_s(W^\text{OCB})&<\Rcal_s(W^{A\to B})\\ 
&\text{and} \nonumber \\ 
\Rcal_c(W^\text{OCB})&>\Rcal_c(W^{A\to B})=0,
\end{align}
reinforcing that neither process matrix can be converted into the other with free adapters.

There are also examples of pairs of process matrices that are \textit{both} causally non-separable but that, nevertheless, cannot be converted into each other. In order to find such examples, we numerically sample process matrices and check their properties. For bipartite process matrices, this works as follows:
\begin{enumerate}
    \item Fix input dimensions $d_{A_I},d_{B_I}$ and output dimensions $d_{A_O},d_{B_O}$.
    \item Uniformly sample a positive semidefinite matrix $Q$ of size ($d_{A_I}d_{A_O}d_{B_I}d_{B_O}$)$\times$($d_{A_I}d_{A_O}d_{B_I}d_{B_O}$).
    \item Define $\widetilde{W} := L_V(Q)$ to be the projection of $Q$ on the subspace of valid process matrices, where $L_V$ is the projection operator in Eq.~\eqref{eqn::DefProcMat}.
    \item Define $W = d_{A_O}d_{B_O}\frac{\widetilde{W}}{\tr(\widetilde{W})}$.
    \item Check whether $W$ is a positive semidefinite matrix. If not, discard $W$ and repeat the process. If yes, than $W$ is a proper process matrix.
\end{enumerate}
This method can be straightforwardly extended to process matrices of multiple parties.

Sampling random processes in this way, we find an example of a process matrix $W^\#$, whose signalling robustness is $\Rcal_s(W^\#)=1.0581$, while its causal robustness is $\Rcal_c(W^\#)=0.0245$. Just like $W^\text{OCB}$, the process matrix $W^\#$ is causally non-separable. However, there is no clear hierarchy of robustness between them, since they satisfy the relations
\begin{align}
\Rcal_s(W^\text{OCB})&<\Rcal_s(W^\#)\\ 
&\text{and} \nonumber \\ 
\Rcal_c(W^\text{OCB})&>\Rcal_c(W^\#)>0.
\end{align}
Notice that, even though both process matrices in this example are causally non-separable, which means that, in principle, there could exist a free adapter that could convert one into the other, the impossibility of such conversion is nevertheless guaranteed by the lack of hierarchy between them, as shown by the different ordering in their signalling and causal robustness values. A data file containing the matrices that correspond to the processes $W^\text{OCB}$ and $W^\#$, along with the values of their robustness, can be found in the online repository~\cite{github}.

It would then be interesting to find out whether in the set of causally non-separable processes there exist one process matrix that could be transformed into all others with free adapters (in principle, this is possible despite the lack of total order). This process matrix would necessarily have higher or equal causal and signalling robustness then all other causally non-separable processes.

In order to investigate the potentially \textit{most} valuable resource with respect to causal robustness, we applied a seesaw algorithm, alternatingly optimising the witness $S$ for a given process matrix $W$ and vice versa. Concretely, the algorithm iterates the SDP~\eqref{sdp::causaldual} and the SDP given by
\begin{align}
\begin{split}\label{sdp::parallel_seesawmax}
    \textbf{given} \ \ & S \\ 
    \textbf{max} \ \ & \tr(WS) \\ 
    \textbf{s.t.} \ \ &  W \in \texttt{Proc}.
\end{split}
\end{align}
which takes the optimal witness from the solution of SDP~\eqref{sdp::causaldual} as input and outputs the process matrix that maximally violates it. The output process matrix then becomes the input of SDP~\eqref{sdp::causaldual} in the next round of iteration, and so on. As this problem is non-convex, one is not guaranteed to obtain the global optimum for $\tr(WS)$ by means of this algorithm. However, running the seesaw for a large number of different initial process matrices provides a good guess for the process matrices which maximize $\Rcal_c(W)$. 

The highest value of causal robustness that we have found for process matrices in which all subspaces are qubits with our seesaw is $\Rcal_c(W^*)=0.2104$. This value exceeds $\Rcal_c(W^\text{OCB})$ showing that $W^\text{OCB}$ cannot be transformed into $W^*$ with free adapters. This same process matrix $W^*$ has signalling robustness of $\Rcal_s(W^*)=1.4805$, which is also greater than $\Rcal_s(W^\text{OCB})=1$. Hence, these two causally non-separable process matrices satisfy the relation 
\begin{align}
\Rcal_s(W^\text{OCB})&<\Rcal_s(W^*)\\ 
&\text{and} \nonumber \\ 
\Rcal_c(W^\text{OCB})&<\Rcal_c(W^*),
\end{align}
making $W^*$ a good candidate for a process matrix that can be transformed into $W^\text{OCB}$ with free adapters. A data file containing the matrices that correspond to the processes $W^\text{OCB}$ and $W^*$, along with the values of their robustness, can be found in the online repository~\cite{github}.

In principle, the question of whether one process can be transformed to another by means of free adapters can be decided by the SDP
\begin{align}
\begin{split}\label{sdp::feasibility}
    \textbf{given} \ \ & W, W' \\ 
    \textbf{find} \ \ & \Upsilon\\ 
    \textbf{s.t.} \ \ &  \Upsilon\star W = W' \\
    & \Upsilon \in \Theta_{\texttt{NS}}\,.
\end{split}
\end{align}
Using this SDP, we have verified that the conversion of $W^*$ into $W^\text{OCB}$ with free adapters is \textit{not} possible. Although not conclusive, this result indicates that there might not be a causally non-separable process which represents the most valuable resource with respect to causal connection.

Considering the matter of a most resourceful process matrix on more general grounds, we may derive some insight from the analysis of a \textit{hypothetical} process.

Consider the matrix $Z \in \Bcal(\Hcal_{A_I}\otimes\Hcal_{A_O}\otimes\Hcal_{B_I}\otimes\Hcal_{B_O})$ consisting of an identity channel from $A_O$ to $B_I$ and another identity channel from $B_O$ to $A_I$, given by
\begin{gather}
Z\coloneqq \Phi^{+}_{A_OB_I}\otimes\Phi^{+}_{A_IB_O}.
\end{gather}

$Z$ is not a proper process matrix, as evidenced by the fact that $Z\neq L_V(Z)$, since by allowing for perfect communication in both directions it could give rise to causal loops. This non-valid process could be considered a natural extension of the most valuable process with respect to signalling robustness -- the fully-signalling bipartite ordered processes $W^{A\to B}$ or $W^{B\to A}$ [see Eq.~\eqref{eqn::fullsigprocess}] that can be transformed into any other ordered processes of the same order and subspace dimensions by means of free adapters -- to the realm of indefinite causal order. Indeed, \textit{all} ordered process matrices can be achieved from $Z$ by simply blocking one of the output spaces to create a one-way fully-signalling ordered process, and using a free adapter to transform it. It can also be checked that $W^\text{OCB}$ and $W^*$ can be generated by applying free adapters to $Z$, as well as all randomly sampled causally non-separable processes we tried. Hence, it is not unreasonable to expect that \textit{all} bipartite processes could be reached via free transformations on $Z$. 

Although such an unrealistic resource may be overkill when it comes to generating all bipartite process matrices, the evidence presented here suggests that it could actually be the required resource, in the sense that no proper process matrix might be able to fulfil such a role. 

In order to better understand the non-valid process $Z$, we investigate how far it is from the set $\texttt{Proc}$ of valid process matrices of the same dimension. One potential measure for this distance is the robustness of $Z$ with respect to the set of valid process matrices, i.e.,

\begin{gather}
\Rcal_{\texttt{Proc}}(Z) = \underset{T} \min \left\{ s\geq 0 \left| \frac{Z + sT}{1+s} = W \in \texttt{Proc} \right.\right\}.
\end{gather}

Here, we must carefully choose the characteristics of the noise $T$ against which we will be measuring robustness. If $T$ is set to be any valid process matrix, then $\Rcal_{\texttt{Proc}}(Z)$ will diverge, since $Z$ is not contained in the image of the projector $L_V$. A more viable option, although with an arguably more far-fetched interpretation, would be to set the noise to any $T\geq0$ of fixed trace $\tr T = d_{A_O}d_{B_O}$. The solution of such a problem for an all-qubit $Z$ is then $\Rcal_{\texttt{Proc}}(Z)=3$ and $W=W^\text{mix}$ [see Eq.~\eqref{eqn::Wmix}]. 

Perhaps surprinsingly, according to this notion of distance, the valid process matrix closest to $Z$ is not a causally non-separable process such as $W^*$ or $W_\text{OCB}$, but instead, an equally-weighted convex mixture of two one-way fully-signalling processes ordered in opposite directions, precisely the most robust process matrix with respect to signalling robustness.

While not a physical process, the discussion of the (non-valid) process $Z$ sheds some light on the properties a most resourceful two-party process matrix $W^\mathrm{max}$ would have to satisfy. On the one hand, it would have to contain perfect channels both from Alice to Bob, as well as from Bob to Alice, that can be `addressed' independently, such that one can degrade $W^\mathrm{\max}$ to both $W^{A\rightarrow B}$ \textit{and} $W^{B\rightarrow A}$. This, already, seems to exclude the existence of $W^\mathrm{max}$, since such a process matrix would most likely contain closed loops and thus lead to logical paradoxa. On the other hand, $W^\mathrm{max}$ would have to be causally non-separable, since causal separability is preserved under free adapters. Consequently, we conjecture that $Z$ (or variants thereof) is the only positive matrix of the correct trace that can be transformed to \textit{all} processes by means of free adapters, but, as already mentioned, $Z$ is not a valid process itself. 

\section{Conclusion}\label{sec::Conclusion}
In this paper, we have constructed a resource theory of causal connection, making  a first step towards a systematic resource-theoretic understanding of signalling in general (non-)causal structures. We derived the fundamental building blocks of this theory -- the sets of free objects and free transformations -- both on operational grounds as well as more axiomatic considerations. In turn, the resulting different levels of control of the involved parties over the adapter yield a strict hierarchy of sets of free transformations. 

The axiomatic perspective leads to the largest set of possible free transformations, which, despite being free, still allows for internal signalling as well as for the creation of causal non-separability. The guiding principle in the more operationally motivated settings were generalized non-signalling conditions between the involved parties via free transformations. Like in the case of many other resource theories -- most prominently that of entanglement -- all of these approaches lead to the same set of free objects, but differ with respect to the free transformations, thus giving rise to distinct resource theories of causal connection.

In order to quantify the causal connection of general processes, we introduced the signalling robustness as a faithful monotone of the resource theory of causal connection and provided a general upper bound for this monotone for the case of two parties, as well as tight upper bounds for many relevant multi-party scenarios. Somewhat surprisingly, the process that maximizes the signalling robustness for two parties is causally ordered, implying that causal non-separability does not increase `the amount of communication' that is possible between parties. Based on numerical evidence, we conjecture that this also holds true for an arbitrary numbers of parties, suggesting that the communication advantages displayed by causally indefinite processes are not rooted in an overall greater ability to communicate when the assumption of global causal order is dropped. For two and three parties, in addition, we provided the most resourceful causally ordered process, respectively, i.e., the process that can be transformed to all other processes of the same causal order by means of free transformations.  

Besides quantifying causal connection in general processes, the signalling robustness has a direct interpretation in terms of a witness of causal connection. Concurrently, this robustness can be understood as a measure of how many causal loops can be closed in a given process. 

Our results allow one to investigate and quantify -- under one common umbrella -- the ability to signal, both in causally ordered as well as causally disordered processes. Naturally, analogous considerations for the set of causally separable processes as the set of free objects lead to a resource theory of causal non-separability, which we introduced, showed uniqueness of, and compared to the resource theory of causal connection. In the same vein, using the framework we provided, one could also consider other sets of resources in process matrices, like, for example, the possibility of two-way signalling. Such considerations will be subject to future work.

As we showed by means of causal robustness, a second monotone of the resource theory of causal connection, causally non-separable processes display a layered structure with respect to free transformations, and there exist pairs of causally non-separable processes where no process can freely be transformed to the other. On the other hand, while causally ordered processes can only be transformed to processes of the same causal order by means of free adapters, for two parties, \textit{every} causally indefinite process can be transformed to processes of \textit{either} causal order by means of free transformations. 

While quantifying causal connection, it is currently unclear what operational task is naturally related to the signalling robustness. Using techniques provided in Refs.~\cite{rosset_resource_2018, uola_quantifying_2019,uola_quantification_2020}, for every process $W$ that satisfies $\Rcal_s(W)>0$, one can tailor a particular information-theoretic game for which the process $W$ will outperform every process with vanishing signalling robustness (and this outperformance is quantified by $\Rcal_s(W)$). However, no single operational task is known, for which the signalling robustness faithfully quantifies performance. In particular, such a task would have to involve bidirectional signalling, since it would have to highlight both signalling in the direction $A\rightarrow B$ as well as $B\rightarrow A$ (for example, a task that only requires signalling from Alice to Bob would not faithfully represent the resourcefulness of a process that allows for signalling from Bob to Alice). Besides this operational question, it is of interest to both find a general tight bound for the signalling robustness for any number of parties, and to investigate how the signalling robustness ties in with the complete set of monotones for the resource theory of causal connection that we provide in App.~\ref{app::Montones_allmonotones}.

Additionally, there are still many natural open questions with respect to the interconvertibility of processes under free transformations, both single shot as well as asymptotically, approximately, and catalytically~\cite{chitambar_quantum_2019}. In particular these latter points require detailed analysis since -- unlike in other resource theories like those involving states and channels as resources -- the tensor product of process matrices in general does \textit{not} yield a proper process matrix~\cite{guerin_composition_2019}.  

Finally, as mentioned, in the causally ordered case, there are processes from which all others can be reached by means of free transformations for the two- and three-party scenario. This suggests that these processes play a similar role as that played by the ebit in the resource theory of entanglement, giving rise to the question of distillability~\cite{bennett_concentrating_1996,bennett_purification_1996,horodecki_quantum_2009} of causal connection for general causally ordered processes. 
\\

\noindent All code developed for this work is openly available in the online repository~\cite{github}.

\begin{acknowledgments}
We thank Tristan Kraft, Kavan Modi, Felix A. Pollock, Benjamin Yadin and Marco T\'ulio Quintino for valuable discussions, and Giuseppe Vitagliano for providing the idea for the proof of Prop.~\ref{prop::Upbound}. The authors are also grateful to Princess Bubblegum for performing very demanding calculations. S. M. acknowledges funding from the European Union’s Horizon 2020 research and innovation programme under the Marie Sk{\l}odowska Curie grant agreement No 801110, and the Austrian Federal Ministry of Education, Science and Research (BMBWF). The opinions expressed in this publication are those of the authors, the EU Agency is not responsible for any use that may be made of the information it contains. JB acknowledges funding from the Austrian Science Fund (FWF) through the Zukunftskolleg ZK03 and the START project Y879-N27.  GC acknowledges financial support from the Hong Kong Research Grant Council through grant 17300918 and through the Senior Research Fellowship Scheme SRFS2021-7S02, by the Croucher Foundation, by the John Templeton Foundation through grant 61466, The Quantum Information Structure of Spacetime (qiss.fr).  Research at the Perimeter Institute is supported by the Government of Canada through the Department of Innovation, Science and Economic Development Canada and by the Province of Ontario through the Ministry of Research, Innovation and Science. The opinions expressed in this publication are those of the authors and do not necessarily reflect the views of the John Templeton Foundation.
\end{acknowledgments}

\bibliographystyle{apsrev4-1}
\bibliography{references}

\onecolumngrid
\appendix

\section*{APPENDIX}

\section{Characterisation of admissible adapters}
\label{app::Charc_LegalAdapters}

An admissible adapter $\Upsilon^{\texttt{A}}_{\Ads_I\Ads_O\Bds_I\Bds_O}$ should -- besides being positive -- map process matrices $W\in \Bcal(\Hcal_{A} \otimes \Hcal_{B})$ onto process matrices $W'\in \Bcal(\Hcal_{A'}\otimes \Hcal_{B'})$. Equivalently, as discussed, it can be understood as a map that transforms non-signalling maps $M'_{A'B'} \in \Bcal(\Hcal_{A'} \otimes \Hcal_{B'})$ to non-signalling maps $M_{AB} \in \Bcal(\Hcal_A\otimes \Hcal_B)$. Let us take the latter standpoint to derive the linear constraints on the set of adapters using techniques from~\cite{TulioMilz2021}. First, we see that for any non-signalling map $N' \in \Bcal(\Hcal_{A'} \otimes \Hcal_{B'})$ (whenever there is no risk of confusion, we will drop the subscripts from now on) we have 
\begin{gather}
    N' = N' - {}_{A_O'}N' + {}_{A_I'A_O'}N' \quad \text{and} \quad N' = N' - {}_{B_O'}N' + {}_{B_O'B_I'}N'\, .
\end{gather}
Combining these two conditions, we see that 
\begin{gather}
\begin{split}
    N' &= N' - {}_{A_O'}N' + {}_{A_I'A_O'}N' - {}_{B_O'}N' + {}_{A_O'B_O'}N' - {}_{A_I'A_O'B_O'}N' +  {}_{B_I'B_O'}N' - {}_{A_O'B_I'B_O'}N' + {}_{A_I'A_O'B_I'B_O'}N'  \\
    &=: L_{ns}'[N']
\end{split}
\end{gather}
holds for all non-signalling maps. Naturally, this equation is stronger than the requirement for general trace preserving maps, since it implies ${}_{A_O'B_O'}N' = {}_{A_I'A_O'B_I'B_O'}N'$, which is the TP requirement on the Choi matrices of maps $\Bcal(\Hcal_{A_{I}'}\otimes \Hcal_{B_I'}) \rightarrow \Bcal(\Hcal_{A_{O}'}\otimes \Hcal_{B_O'})$. The operator $L'_{NS}$ defined above is a trace-preserving self-dual projector. Now, demanding that $\Upsilon^{\texttt{A}}$ maps non-signalling maps to non-signalling maps implies
\begin{gather}
    L_{ns}[\Upsilon^{\texttt{A}} \star N'] = \Upsilon^{\texttt{A}} \star N'. 
\end{gather}
for all non-signalling maps $N'$. Since the span of the set of non-signalling maps is not the full space of Hermitian matrices, the above equation does not yet fully determine the linear constraints on $\Upsilon^{\texttt{A}}$. However, since $L_{ns}'$ is a projector, we see that every matrix $N' = L_{ns}'[R']$ satisfies $N' = L_{ns}'[N']$ for arbitrary $R' \in \Bcal(\Hcal_{A'} \otimes \Hcal_{B'})$, implying that the span of the set of non-signalling maps is given by $L'_{ns}[\Bcal(\Hcal_{A'}\otimes \Hcal_{B'})]$. With this, using the self-duality of $L_{ns}'$, we can transform the above equation to a full linear constraint on $\Upsilon^{\texttt{A}}$ such that it reads
\begin{gather}
    (L_{ns} \otimes L_{ns}')[\Upsilon^{\texttt{A}}] \star R' = (\Ical_{AB} \otimes L_{ns}')[\Upsilon^{\texttt{A}}] \star R'
\end{gather}
for \textit{all} $R' \in \Bcal(\Hcal_{A'}\otimes \Hcal_{B'})$. This, then, implies that \begin{gather}
    \Upsilon^{\texttt{A}} = \Upsilon^{\texttt{A}} - (\Ical_{AB} \otimes L_{ns}')[\Upsilon^{\texttt{A}}] + (L_{ns} \otimes L_{ns}')[\Upsilon^{\texttt{A}}]\,.
\end{gather}
Additionally, $\Upsilon^{\texttt{A}}$ is trace-rescaling on the set of non-signalling maps (since it maps channels to channels), i.e., 
\begin{gather}
    \tr(\Upsilon^{\texttt{A}} \star N') = \frac{d_{A_I}d_{B_I}}{d_{A'_I}d_{B'_I}} \tr(N')
\end{gather}
for all non-signalling maps $N'$. Using the same arguments as above, as well as the trace preservation of $L_{ns}'$, this translates to 
\begin{gather}
    \tr[(\Ical_{AB} \otimes L_{ns}'[\Upsilon^{\texttt{A}}]) (R^{\prime \mathrm{T}} \otimes \ident_{AB})] = \frac{d_{A_I}d_{B_I}}{d_{A'_I}d_{B'_I}} \tr(R')\, ,
\end{gather}
which implies 
\begin{gather}
    \tr_{AB}\{\Ical_{AB} \otimes L_{ns}'[\Upsilon^{\texttt{A}}]\} = \frac{d_{A_I}d_{B_I}}{d_{A'_I}d_{B'_I}} \ident_{A'B'}\, .
\end{gather}
With this, we obtain
\begin{gather}
    \tr(\Upsilon^{\texttt{A}}) =  d_{A_I}d_{B_I}d_{A_O'}d_{B_O'}. 
\end{gather}
Using this, we see that
\begin{gather}
    {}_{AB}(L_{ns}'[\Upsilon^{\texttt{A}}]) = {}_{ABA'B'}\Upsilon^{\texttt{A}}\, ,
\end{gather}
where we omitted the identity operator $\Ical_{AB}$. This gives us the conditions on an admissible adapter presented in Def.~\ref{def::LegalAdapters}.

\section{Free-preserving adapters}
\label{app:FreePres}
In the main text, we show that free-preserving adapters, besides being positive, have to satisfy the following constraints:
\begin{gather}
\begin{split}
   {}_{A_O'B_O'} (\Upsilon^\texttt{FP} \star W^{A||B}) &= (\Upsilon^\texttt{FP} \star W^{A||B}) \\
   \text{and} \quad \tr(\Upsilon^\texttt{FP} \star W^{A||B}) &= d_{A'_O}d_{B_O'}\, .
\end{split}   
\end{gather}
Following the same arguments as in the previous section, any free process matrix $W^{A||B}$ can be understood as $W^{A||B} = {}_{A_OB_O}R$ for some $R$, so that the first of the two above equations can be rewritten as  
\begin{gather}
    {}_{A_O'B_O'} (\Upsilon^\texttt{FP} \star {}_{A_OB_O}R) = (\Upsilon^\texttt{FP} \star {}_{A_OB_O}R)\, .
\end{gather}
Using the self-duality of ${}_{A_OB_O}\sbt$ and the fact that the above has to hold for all $R\in \Bcal(\Hcal_{A}\otimes \Bcal_{B})$, we obtain 
\begin{gather}
     {}_{A_O'B_O'A_OB_O}\Upsilon^{\texttt{FP}} = {}_{A_OB_O} \Upsilon^{\texttt{FP}} 
\end{gather}
as desired. Again, analogously to the proof in the previous section, the trace rescaling property of $\Upsilon^\texttt{FP}$ on the set of free process matrices can be rewritten as 
\begin{gather}
    \tr(\Upsilon^\texttt{FP} \star {}_{A_OB_O}R) = \frac{d_{A'_O}d_{B_O'}}{d_{A_O}d_{B_O}} \tr R\, 
\end{gather}
for all $R\in \Bcal(\Hcal_{A}\otimes \Bcal_{B})$, where we used the fact that ${}_{A_OB_O}\sbt$ is a trace preserving operator. This implies $\tr_{A'B'}({}_{A_OB_O}\Upsilon^{\texttt{FP}}) = \frac{d_{A'_O}d_{B_O'}}{d_{A_O}d_{B_O}} \ident_{AB}$, which, in turn, yields 
\begin{gather}
    \tr(\Upsilon^\texttt{FP}) =  d_{A_I}d_{B_I}d_{A_O'}d_{B_O'} \quad \text{and} \quad {}_{A'B'A_OB_O}\Upsilon^\textup{\texttt{FP}} = {}_{A'B'AB}\Upsilon^\textup{\texttt{FP}}\, ,
\end{gather}
as claimed in the main text.

\section{Admissible and free-preserving adapters can create causal non-separability}
\label{app::Maps_SeptoNonsep}

As we have seen in the main text, admissible and free-preserving adapters $\Upsilon^{AFP}$ can change the causal order of a non-free process and even create mixtures of causal orders. Now we are left to verify that admissible and free-preserving adapters can also map non-free causally separable processes to causally non-separable processes. We tackle this question numerically and we make use of the witnesses for causal non-separability that where introduced in~\cite{araujo_witnessing_2015}. These witnesses are such that, if for a witness $G$ and a process matrix $W$ we have $\tr(GW)<0$, then the process matrix $W$ is causally non-separable. With this, one can search for an adapter $\Upsilon^{\texttt{A}\texttt{FP}} \in \Theta_{{\texttt{A}\texttt{FP}}}$ that satisfies $\Upsilon^{{\texttt{A}\texttt{FP}}} \star W \notin \texttt{Sep}$ for some $W \in \texttt{Sep}$ via a see-saw SDP. In particular, if there is a causally separable $W$ such that $\Upsilon^{{\texttt{A}\texttt{FP}}} \star W \notin \texttt{Sep}$, then there must also be a causally ordered process matrix $W^{A\prec B}$ such that $\Upsilon^{{\texttt{A}\texttt{FP}}} \star W^{A\prec B}  \notin \texttt{Sep}$. As we see in Sec.~\ref{subsec::max_par_robust}, for two parties, the most valuable process (in the sense that it maximizes the causal robustness) with ordering $A\prec B$ is of the form 
\begin{gather}
    W^{A\rightarrow B} = \Psi_{A_I} \otimes \Phi^+_{A_OB_I} \otimes \ident_{B_O}\, , 
\end{gather}
where $\Psi_{A_I}$ is an arbitrary quantum state. 

Intuitively, this process matrix provides a good starting point when searching for causally non-separable processes resulting from the action of an adapter $\Upsilon^{{\texttt{A}\texttt{FP}}} \in \Theta_{{\texttt{A}\texttt{FP}}}$ on a causally ordered process. Employing this guess, we run the following program to find such an admissible and free-preserving adapter; first, we sample a random witness $G$ for causal non-separability. For this given witness, we run the SDP
\begin{align}
\begin{split}
 \textbf{given} \ \ & \text{Witness of causal non-separability} \ G \\ 
 \textbf{minimize} \ \ & \tr[(\Upsilon \star W^{A\rightarrow B})G] \\ 
 \textbf{subject to} \ \ & \Upsilon^{{\texttt{A}\texttt{FP}}} \in \Theta_{\texttt{A}\texttt{FP}}
\end{split}
\end{align}

If the resulting $\Upsilon^{{\texttt{A}\texttt{FP}}}$ is such that $\tr[(\Upsilon^{{\texttt{A}\texttt{FP}}} \star W^{A\rightarrow B})G] <0$, then the resulting process matrix $\Upsilon^{{\texttt{A}\texttt{FP}}} \star W^{A\rightarrow B}$ is causally non-separable and we have indeed found an admissible and free-preserving adapter that maps causally separable to causally non-separable process matrices. To obtain better and more robust results, the resulting adapter of the above SDP can be fed into a second SDP to optimize the respective witness $G$, i.e., 
\begin{align}
\begin{split}
 \textbf{given} \ \ & \text{Adapter} \ \Upsilon^{{\texttt{A}\texttt{FP}}} \in \Theta_{\texttt{A}\texttt{FP}} \\ 
 \textbf{minimize} \ \ & \tr[(\Upsilon^{{\texttt{A}\texttt{FP}}} \star W^{A\rightarrow B})G] \\ 
 \textbf{subject to} \ \ & G \ \text{is proper witness of causal non-separability.}
\end{split}
\end{align}
While running these two SDPs is not guaranteed to yield an adapter $\Upsilon^{{\texttt{A}\texttt{FP}}} \in \Theta_{\texttt{A}\texttt{FP}}$ with the desired properties, it actually provides adapters for which $\tr[(\Upsilon^{{\texttt{A}\texttt{FP}}} \star W^{A\rightarrow B})G] \approx -0.1661$ (as a comparison, for the best witness of $W^{\text{OCB}}$, one obtains $\tr(W^\text{OCB} G) \approx -0.1716$). While not necessarily the best achievable value, $-0.1661$ is well-beyond any potential numerical imprecision thresholds, implying that we indeed found an adapter with the desired properties. Since it does not appear to possess a particularly nice structure, we do not report it here explicitly, however, the code used to find this adapter and a data file containing the matrices that correspond to the adapter $\Upsilon^{{\texttt{A}}\texttt{FP}}$, the causally ordered, fully signalling process matrix $W^{A\rightarrow B}$, and the causally non-separable process matrix $\Upsilon^{{\texttt{A}\texttt{FP}}} \star W^{A\rightarrow B}$ can be found in the online repository~\cite{github}.

\section{Definition of non-signalling adapters}
\label{app::NSAdapt}
In Sec.~\ref{sssec::NS_Adapters} we defined the set $\Theta_\texttt{NS}$ of non-signalling adapters $\Upsilon^\texttt{NS}_{\Ads_I\Ads_O\Bds_I\Bds_O} \in \Bcal(\Hcal_{\Ads_I}\otimes \Hcal_{\Ads_O}\otimes\Hcal_{\Bds_I}\otimes\Hcal_{\Bds_O})$ as the set of positive semidefinite matrices that satisfy 
\begin{gather}
\label{eqn::freeAdApp}
\begin{split}
    &\Upsilon^\texttt{NS} = L_A[\Upsilon^\texttt{NS}], \quad \Upsilon^\texttt{NS} = L_B[\Upsilon^\texttt{NS}], \quad \\ 
    \text{and} \ &\tr[\Upsilon^\texttt{NS}] = d_{A_O'}d_{A_I} d_{A_O'}d_{A_I}\, ,
\end{split}
\end{gather}
where $L_X[\Upsilon^\texttt{NS}] = \Upsilon^\texttt{NS} - {}_{X_O}\Upsilon^\texttt{NS} + {}_{X_OX_O'}\Upsilon^\texttt{NS} - {}_{X_I'X_OX_O'}\Upsilon^\texttt{NS} + {}_{X_IX_I'X_OX_O'}\Upsilon^\texttt{NS}$. While positivity guarantees that adapters will lead to positive probabilities even when applied to process matrices with additional degrees of freedom (i.e., degrees of freedom the adapter does not act on), the above trace conditions and the normalization of $\Upsilon^\texttt{NS}$ ensure that $\Upsilon^\texttt{NS}$ does not allow for communication between Alice and Bob, and that the resulting process matrix is properly normalized, i.e., $\tr[\Upsilon^\texttt{NS} \star W] = d_{A_{O}'} d_{B_O'}$ for all $W\in \texttt{Proc}$ and all $\Upsilon^\texttt{NS} \in \Theta_\texttt{NS}$. 

Let us first show the former statement, namely that a non-signalling adapter that satisfies Eq.~\eqref{eqn::freeAdApp} does not allow for communication between Alice and Bob. To see this, consider a deterministic operation $\Omega_{\Ads_IA_O'} \otimes \ident_{A_O}$ that Alice can perform. Due to its causal ordering we have ${}_{A_O'}\Omega_{\Ads_IA_O'} = {}_{A_O'A_I'}\Omega_{\Ads_IA_O'}$ and $\tr[\Omega_{\Ads_IA_O'}] = d_{A_I'}$. It is easy to see that with this, we have $L_A[\Omega\otimes \ident_{A_O}] = {}_{A_IA_I'A_OA_O'}[\Omega\otimes \ident_{A_O}]$. We thus obtain
\begin{gather}
\begin{split}
    \Upsilon^\texttt{NS} \star (\Omega \otimes \ident_{A_O})&= (L_A \otimes L_B)[\Upsilon^\texttt{NS}] \star (\Omega \otimes \ident_{A_{O}})
    = (L_A \otimes L_B)[\Upsilon^\texttt{NS}] \star L_A[\Omega] = \\
    &= \Upsilon^\texttt{NS} \star {}_{A_IA_I'A_OA_O'}\Omega = \frac{1}{d_{A_I}d_{A_O'}} \tr_{\Ads_I\Ads_O}\Upsilon^\texttt{NS}, 
    \end{split}
\end{gather}
where we have dropped the labels for concise notation and we have used the fact that the operators $L_X$ are self-dual projections. Since the last line of the above equation is independent of the specific comb $\Omega$ Alice employs, Bob's local adapter is independent of Alice's operations, implying that she cannot signal to him via the non-signalling adapter $\Upsilon^\texttt{NS}$. Running the same argument for operations on Bob's side then shows that a non-signalling adapter indeed does not allow for signalling between the involved parties. This, then, implies that adapters that satisfy the properties laid out in Def.~\ref{def::free_adapters} automatically satisfy the requirements of Eqs.~\eqref{eqn::non_signal1} and~\eqref{eqn::non_signal2}.

Finally, it would remain to show that non-signalling adapters indeed satisfy $\tr(\Upsilon^\texttt{NS}_{\Ads_{I}\Ads_O \Bds_I \Bds_O}\star W_{AB}) = d_{A_O'}d_{B_O'}$, i.e., they transform process matrices to correctly-normalized objects. This has already been done in the main text. As we saw, the free adapters form a proper subset of all valid adapters, i.e., $\Theta_\texttt{NS} \subset \Theta_{\texttt{A}}$, and all admissible adapters map proper process matrices to proper process matrices, i.e., they also preserve the correct normalization.

\section{Shared entanglement and local operations}
\label{app::pre_shared}

In the main text, we claimed that $\Theta_\texttt{LOSE} \subset \Theta_\texttt{NS}$, implying that adapters from local operations and shared entanglement satisfy not only Requirement R$3$, i.e., they do not allow for signalling between Alice and Bob, but also Requirements R$1$ and R$2$, by mapping proper processes to proper processes and free processes to free processes. We state this here in more detail as the following Lemma: 

\begin{lemma}
An adapter of the form 
\begin{gather}
    \Upsilon^{\textup{\texttt{LOSE}}} = \rho_{\widetilde A_I \widetilde A'_O \widetilde B_I \widetilde B'_O} \star \Lambda_{\widetilde \Ads_I \Ads_I} \star \Gamma_{\widetilde A' \Ads_O} \star \Lambda_{\widetilde \Bds_I \Bds_I} \star \Gamma_{\widetilde B' \Bds_O}\, ,
\label{eqn::presharedApp}
\end{gather}
stemming from shared entanglement and local operations, maps proper process matrices to proper processes and does not allow for signalling between Alice and Bob.
\end{lemma}

\begin{proof}
We start by proving the former. Since the link product of positive objects is again a positive object, we see that for an adapter of the form of Eq.~\eqref{eqn::presharedApp} we have $\Upsilon^\texttt{LOSE} \geq 0$ and thus $W'_{A'B'} = \Upsilon^\texttt{LOSE} \star W_{AB} \geq 0$ for all proper process matrices $W_{AB}$. It remains to show that $W'_{A'B'}$ is a proper process matrix. This can be done by checking its action on products $M_{A'} \otimes M_{B'}$ of CPTP maps. We have 
\begin{align}
    \begin{split}
        &W'_{A'B'} \star M_{A'} \star M_{B'}= \rho_{\widetilde A_I \widetilde A'_O \widetilde B_I \widetilde B'_O} \star \Lambda_{\widetilde \Ads_I \Ads_I} \star \Gamma_{\widetilde A' \Ads_O} \star \Lambda_{\widetilde \Bds_I \Bds_I} \star \Gamma_{\widetilde B' \Bds_O} \star W_{AB} \star M_{A'} \star M_{B'} \\
        &=: N_{AB} \star W_{AB} 
    \end{split}
    \label{eqn::ProperProcOp}
\end{align}
It is easy to see that the map $N_{AB}$ satisfies $\tr_{A_O}N_{AB} = \ident_{A_I} \otimes N_{B}$ and $\tr_{B_O}N_{AB} = \ident_{B_I} \otimes N_{A}$ with CPTP maps $N_A$ and $N_B$, implying that $N_{AB}$ is non-signalling between Alice and Bob.

\begin{figure}
    \centering
    \begin{minipage}{.45\textwidth}
        \centering
        \includegraphics[width=0.9\linewidth]{Shared_initial.png}
        \caption{\textbf{Free adapters from shared entanglement and local operations.} For better orientation, here, we reproduce Fig.~\ref{fig::Shared_initial} from the main text.}
        \label{fig::Shared_initial2}
    \end{minipage}%
    \hspace*{10mm}
    \begin{minipage}{0.45\textwidth}
        \centering
        \includegraphics[width=0.8\linewidth]{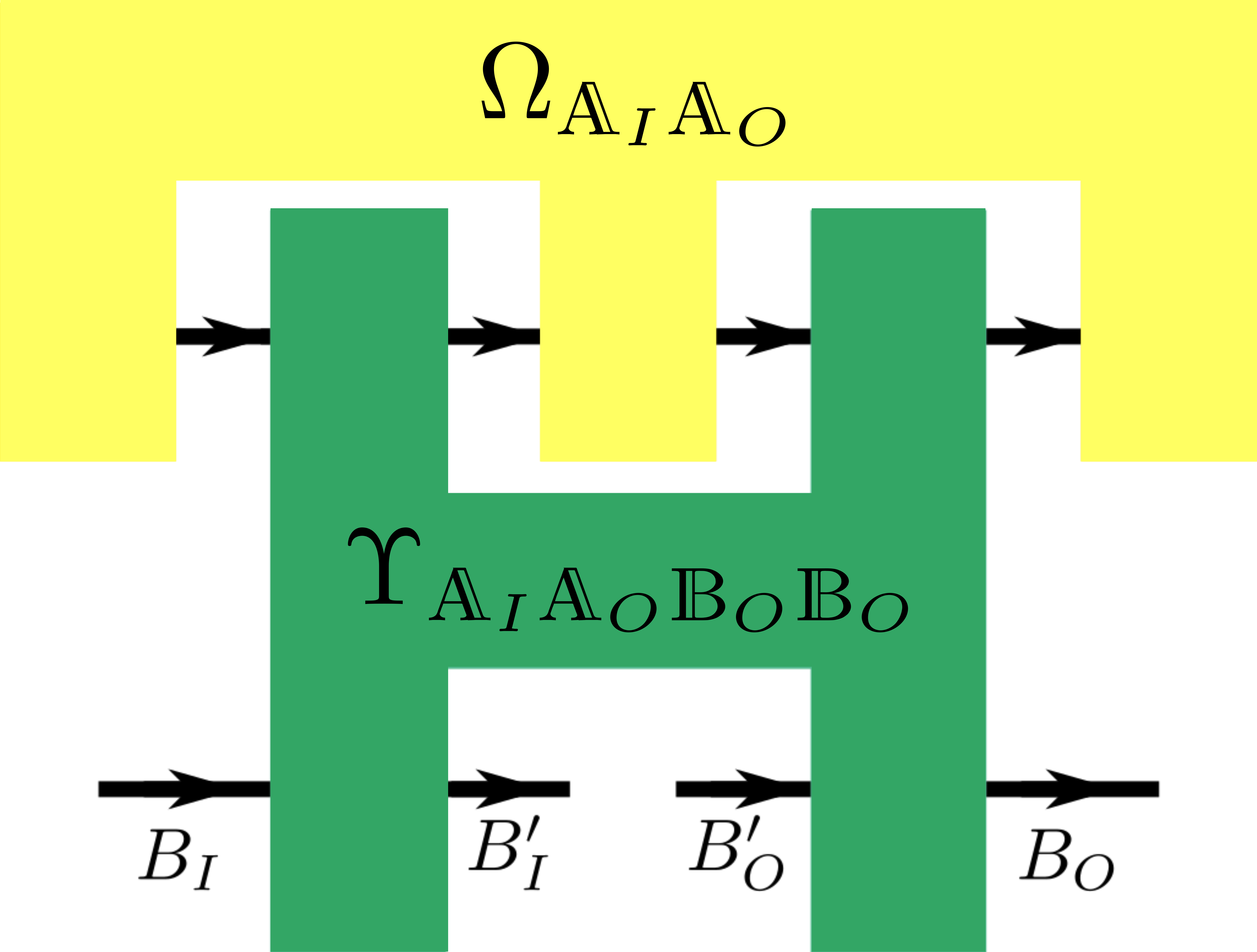}
        \caption{\textbf{No signalling for adapters.} Given the adapter $\Upsilon_{\Ads_I\Ads_O\Bds_I\Bds_O}$, the most general thing Alice could do to send information to Bob would be to perform a causally ordered comb $\Omega_{\Ads_I\Ads_O}$. If the adapter is non-signalling, then the resulting comb on Bobs side (defined on $B_I, B_I', B_O'$, and $B_O$) must be independent of $\Omega_{\Ads_I\Ads_O}$.}
        \label{fig:Non_sig_Adapt}
    \end{minipage}
\end{figure}

Concretely, we have (see Fig.~\ref{fig::Shared_initial2} for better orientation):
\begin{gather}
    \begin{split}
        &\tr_{A_O}N_{AB} = \rho_{\widetilde A_I \widetilde A'_O \widetilde B_I \widetilde B'_O} \star \Lambda_{\widetilde \Ads_I \Ads_I} \star \tr_{A_O}\Gamma_{\widetilde A' \Ads_O} \star \Lambda_{\widetilde \Bds_I \Bds_I} \star \Gamma_{\widetilde B' \Bds_O} \star M_{A'} \star M_{B'}\\
        &= \tr_{\widetilde  A'_O}\rho_{\widetilde A_I \widetilde A'_O \widetilde B_I \widetilde B'_O} \star \tr_{\widetilde A_I'} \Lambda_{\widetilde \Ads_I \Ads_I} \star \Lambda_{\widetilde \Bds_I \Bds_I}  \star \Gamma_{\widetilde B' \Bds_O} \star \tr_{A_O'}M_{A'} \star M_{B'} \\
        &= \rho_{\widetilde A_I \widetilde B_I \widetilde B'_O} \star \tr_{\widetilde A_I' A_I'} \Lambda_{\widetilde \Ads_I \Ads_I} \star \Lambda_{\widetilde \Bds_I \Bds_I} \star \Gamma_{\widetilde B' \Bds_O} \star M_{B'}\\
        &= \ident_{A_I} \star \rho_{\widetilde{B}_I\widetilde B_O'} \star \Lambda_{\widetilde \Bds_I \Bds_I} \star \Gamma_{\widetilde B' \Bds_O} \star M_{B'}\\
        &=: \ident_{A_I} \otimes N_B
    \end{split}
\end{gather}
where we have used $\tr_{A_O}\Gamma_{\widetilde A' \Ads_{O}} = \ident_{\widetilde A' A_O'}$,  $\tr_{A_O'}M_{A'} = \ident_{A_I'}$, $\tr_{\widetilde A_I' A_I'} \Lambda_{\widetilde \Ads_I \Ads_I}$,  and the fact that $\ident_X$ is the Choi matrix of $\tr_X$. In a similar vein, it is easy to see that $N_B$ satisfies $\tr_{B_{O}}N_B = \ident_{B_I}$, i.e., it is CPTP. The corresponding relations for $\tr_{B_O}N_{AB}$ follow analogously. $N_{AB}$ is thus a non-signalling map. Invoking the decomposition~\eqref{eqn::DecompNonSig} of general non-signalling maps in terms of products of CPTP maps implies that $\Upsilon^\texttt{LOSE}$ maps non-signalling maps on `the primed degrees of freedom' to non-signalling maps `on the unprimed degrees of freedom'. As we have seen in the main text, this implies that it maps proper process matrices to proper process matrices, which concludes the first part of the proof.

To show that the above $\Upsilon^\texttt{LOSE}$ does not allow for signalling between Alice and Bob, we first note that, in the above reasoning, we have already shown that $\Upsilon^\texttt{LOSE}$ maps all non-signalling maps to non-signalling maps (i.e., $\Upsilon^\texttt{LOSE} \in \Theta_{\texttt{A}\texttt{FP}}$). However, as discussed in detail in Sec.~\ref{sssec::NS_Adapters}, this is not sufficient to guarantee that the adapter does not enable any signalling between Alice and Bob. In addition, we have to show that, given the adapter $\Upsilon^\texttt{LOSE}$, no matter what Alice `inserts' into it, Bob locally `sees' the same comb (and vice versa). In the language of Def.~\ref{def::free_adapters}, where we defined non-signalling adapters, this means that $\Upsilon^\texttt{LOSE}$ satisfies $L_A[\Upsilon^\texttt{LOSE}] = L_B[\Upsilon^\texttt{LOSE}] = \Upsilon^\texttt{LOSE}$. Here, for concreteness, we will not use the projectors $L_A$ and $L_B$ to prove the non-signalling conditions, but rather show them directly (yet equivalently) by proving that Alice cannot send signals to Bob by performing deterministic operations on her side of the adapter. 

In detail, this means that $\Upsilon^\texttt{LOSE} \star \Omega_{\Ads_I\Ads_O}$ is independent of any comb $\Omega_{\Ads_I\Ads_O}$ with ordering $A_I\prec A_I' \prec A_O' \prec A_O$ that Alice could `insert' into the adapter (see Fig.~\ref{fig:Non_sig_Adapt}).

This can be checked by direct insertion into Eq.~\eqref{eqn::presharedApp} and using the causality constraints on $\Omega_{\Ads_I\Ads_O} = \ident_{A_O} \otimes \Omega_{\Ads_I A_O'}$:  
\begin{gather}
\begin{split}
    &\Upsilon^\texttt{LOSE} \star \Omega_{\Ads_I\Ads_O} = \ident_{A_O} \star \Gamma_{\widetilde A' \Ads_O} \star \Omega_{\Ads_I A_O'} \star \rho_{\widetilde A_I \widetilde A'_O \widetilde B_I \widetilde B'_O} \star \Lambda_{\widetilde \Ads_I \Ads_I}  \star \Lambda_{\widetilde \Bds_I \Bds_I} \star \Gamma_{\widetilde B' \Bds_O} \\
    &= \ident_{\widetilde{A}_O'A_O'\widetilde{A}_I'} \star \Omega_{\Ads_I A_O'} \star \rho_{\widetilde A_I \widetilde A'_O \widetilde B_I \widetilde B'_O} \star \Lambda_{\widetilde \Ads_I \Ads_I} \star \Lambda_{\widetilde \Bds_I \Bds_I} \star \Gamma_{\widetilde B' \Bds_O} \\
    &= \rho_{\widetilde B_I \widetilde B'_O} \star \Lambda_{\widetilde \Bds_I \Bds_I} \star \Gamma_{\widetilde B' \Bds_O}\, ,
\end{split}
\end{gather}
which is independent of Alice's operations. As for the previous derivation, here, we have alternatingly used the properties of CPTP maps and causally ordered combs, in particular $\tr_{A_O}\Gamma_{\widetilde A' \Ads_O} = \ident_{\widetilde{A}_O'A_O'\widetilde{A}_I'}$, $\tr_{A_I'\widetilde{A}_I'}\Lambda_{\widetilde \Ads_I \Ads_I} = \ident_{\widetilde{A}_I A_I}$, and $\tr_{A_O'}\Omega_{\Ads_I A_O'} = \ident_{A_I'} \otimes \Omega_{A_I}$. Naturally, the same independence can be shown with respect to Bob's operations, implying that $\Upsilon^{\texttt{LOSE}}_{\Ads_I\Ads_O\Bds_I\Bds_O}$ does not enable any signalling between Alice and Bob. Consequently, we have $\Theta_\texttt{LOSE} \subseteq \Theta_\texttt{NS}$. The fact that the inclusion is strict is proven in the main text. 
\end{proof}

\section{Equivalence of \texorpdfstring{$\Theta_{\texttt{CA}}$}{} and \texorpdfstring{$\Theta_{\texttt{NS}}$}{}}
\label{app::CV_NS}
Here, we provide the prove of Prop.~\ref{thm::comp_leg_ad}, where we stated that the set $\Theta_{\texttt{CA}}$ of completely admissible adapters coincides with the set $\Theta_{\texttt{NS}}$ of non-signalling adapters. As laid out in Sec.~\ref{sec::comp_Leg}, a completely admissible adapter $\Upsilon^{\texttt{CA}}_{\Ads_I\Ads_O\Bds_I\Bds_O}$ must -- besides being a positive semidefinite matrix -- map any non-signalling map $\bar M_{A'\bar A B'\bar B}$ with $A_I'\bar A_I \not \rightarrow B_O'\bar B_O$ and $B_I'\bar B_I \not \rightarrow A_O'\bar A_O$ to a non-signalling map $\bar N_{A\bar A B \bar B} =  \Upsilon^{\texttt{CA}}_{\Ads_I\Ads_O\Bds_I\Bds_O} \star \bar M_{A'\bar A B'\bar B}$, where $\bar N_{A\bar A B \bar B}$ is non-signalling $A_I\bar A_I \not \rightarrow B_O\bar B_O$ and $B_I\bar B_I \not \rightarrow A_O\bar A_O$. 

As was the case for admissible adapters, complete admissibility is equivalent to mapping any non-signalling map of the form $M_{A'\bar{A}} \otimes M_{B'\bar{B}}$ -- where $M_{X'\bar{X}}$ is an arbitrary CPTP map $X_I'\bar{X}_I \rightarrow X_{O}'\bar{X}_O$ -- to a non-signalling map $N_{A\bar{A}B\bar{B}}$ (see Fig.~\ref{fig::CL}).
\begin{figure}
    \centering
    \includegraphics[width = 0.5\linewidth]{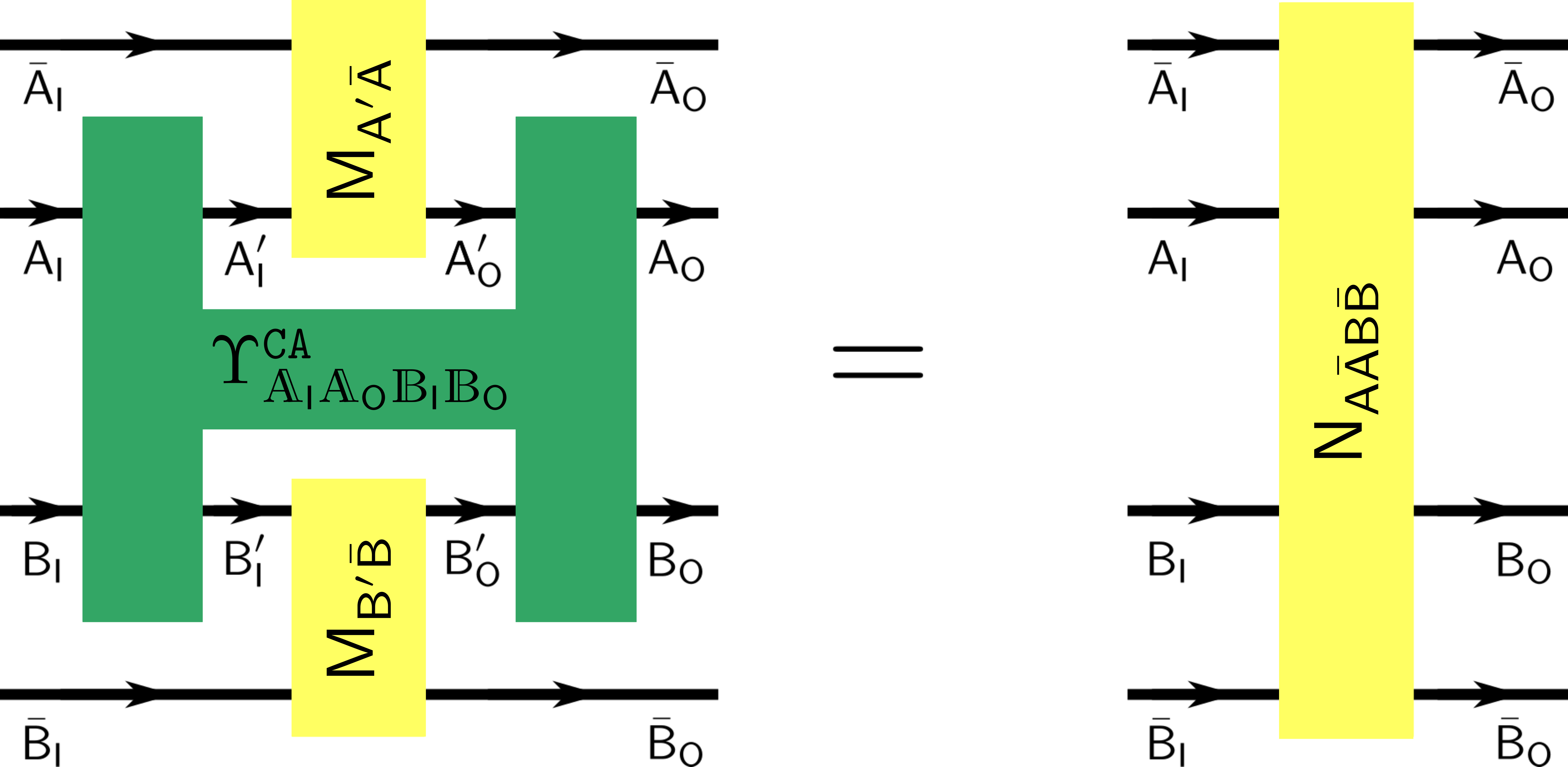}
    \caption{\textbf{Completely admissible adapters.} An adapter $\Upsilon^\texttt{CA}$ is completely admissible if it maps any non-signalling map to a non-signalling map, even when only acting on parts of the map. As for the case of admissible adapters, this is equivalent to mapping arbitrary tensor products of CPTP maps to a non-signalling map, i.e., if the adapter in the figure is completely admissible, then $N_{A\bar{A}B\bar{B}}$ is non-signalling $A_I\bar{A}_I \not \rightarrow B_O\bar{B}_O$ and $B_I\bar{B}_I \not \rightarrow A_O\bar{A}_O$ for all CPTP maps $M_{A'\bar{A}}$ and $M_{B'\bar{B}}$.}
    \label{fig::CL}
\end{figure}

With this, one could derive the linear constraints on the set $\Theta_\texttt{CA}$ of completely admissible adapters in terms of a projector $L_{CA}$ and show that its properties correspond to those of the projectors that define $\Theta_\texttt{NS}$ [given in Eqs.~\eqref{eqn::defFree} and~\eqref{eqn::trFree}]. However, here, we take a more explicit route and directly show that both $\Theta_\texttt{NS} \subseteq \Theta_\texttt{CA}$ and $\Theta_\texttt{CA} \subseteq \Theta_\texttt{NS}$ hold.

We start by showing $\Theta_{\texttt{NS}} \subseteq \Theta_\texttt{CA}$. Recall that for any $\Upsilon^\texttt{NS}_{\Ads_I\Ads_O\Bds_I\Bds_O} \in \Theta_{\texttt{NS}}$ we have ${}_{B_O}\Upsilon^{\texttt{NS}} = {}_{B_OB_O'}\Upsilon^{\texttt{NS}}$ and ${}_{B_OB_O'B_I'}\Upsilon^{\texttt{NS}} = {}_{B_OB_O'B_I'B_I}\Upsilon^{\texttt{NS}}$. With this, we can show that $N_{A\bar{A}B\bar{B}}:= \Upsilon^\texttt{NS} \star (M_{A'\bar{A}} \otimes M_{B'\bar{B}})$ is non-signalling $B_I\bar{B}_I \not \rightarrow A_O\bar{A}_O$, i.e., ${}_{B_O\bar{B}_O} N_{A\bar{A}B\bar{B}} = {}_{B_O\bar{B}_O B_I \bar{B}_I} N_{A\bar{A}B\bar{B}}$. Concretely, we have 
\begin{gather}
\label{eqn::eqCL1}
\begin{split}
    &{}_{B_O\bar{B}_O} N_{A\bar{A}B\bar{B}} \\
    &= {}_{B_O\bar{B}_O}(\Upsilon^\texttt{NS}_{\Ads_I\Ads_O\Bds_I\Bds_O} \star M_{A'\bar{A}} \otimes M_{B'\bar{B}}) =  {}_{B_OB_O'}\Upsilon^\texttt{NS}_{\Ads_I\Ads_O\Bds_I\Bds_O} \star M_{A'\bar{A}} \star {}_{\bar{B}_O}M_{B'\bar{B}} \\
    &= 
    {}_{B_O}\Upsilon^\texttt{NS}_{\Ads_I\Ads_O\Bds_I\Bds_O} \star M_{A'\bar{A}} \star {}_{\bar{B}_OB_O' B_I' \bar{B}_I} M_{B'\bar{B}} = {}_{B_OB_O'B_I'B_I} \Upsilon^{\texttt{NS}}_{\Ads_I\Ads_O\Bds_I\Bds_O} \star M_{A'\bar{A}} \star {}_{\bar{B}_O\bar{B}_I} M_{B'\bar{B}} \\
    & = {}_{B_O\bar{B}_OB_I\bar{B}_I}(\Upsilon^\texttt{NS}_{\Ads_I\Ads_O\Bds_I\Bds_O} \star M_{A'\bar{A}} \star {}_{\bar{B}_OB_O'} M_{B'\bar{B}}) \,
\end{split}
\end{gather}
where we have both used the fact that the operators ${}_X\sbt$ can be moved around freely in the link product and the property ${}_{\bar{B}_OB'_O}M_{B'\bar{B}} =  {}_{\bar{B}_OB'_O\bar{B}_I B_I'}M_{B'\bar{B}}$ of CPTP maps. In the same vein (and using ${}_X({}_X\sbt) = {}_X\sbt$), we see that  
\begin{gather}
\label{eqn::eqCL2}
    {}_{B_O\bar{B}_OB_I \bar{B}_I} N_{A\bar{A}B\bar{B}} =  {}_{B_O\bar{B}_OB_I\bar{B}_I}(\Upsilon^\texttt{NS}_{\Ads_I\Ads_O\Bds_I\Bds_O} \star M_{A'\bar{A}} \star {}_{\bar{B}_OB_O'} M_{B'\bar{B}})\, .
\end{gather}
Comparison of Eqs.~\eqref{eqn::eqCL1} and~\eqref{eqn::eqCL2} implies ${}_{B_O\bar{B}_O} N_{A\bar{A}B\bar{B}} = {}_{B_O\bar{B}_O B_I \bar{B}_I} N_{A\bar{A}B\bar{B}}$. The fact that $M_{A\bar{A}B\bar{B}}$ is non-signalling from $A_I\bar{A}_I$ to $B_O\bar{B}_O$ follows in the same way. Non-signalling adapters thus map non-signalling maps to non-signalling maps, even when only acting on a part of them, implying $\Theta^{\texttt{NS}} \subseteq \Theta^\texttt{CA}$.

To prove $\Theta^{\texttt{NS}} \supseteq \Theta^\texttt{CA}$, let us consider two particular CPTP maps on Alice's and Bob's side, respectively, namely a swap $\bar{A}_I \leftrightarrow A_O', A'_I \leftrightarrow \bar{A}_O$ on Alice's side, and a CPTP map $M_{B'\bar{B}} = L_{B'} \otimes \Phi^+_{\bar{B}_I\bar{B}_O}$ (where $L_{B'}$ is CPTP and $\Phi^+_{\bar{B}_I\bar{B}_O}$ corresponds to an identity channel from $\bar B_I$ to $\bar B_O$) on Bob's side (see Fig.~\ref{fig::CLProof}).
\begin{figure}
    \centering
    \includegraphics[width=0.7\linewidth]{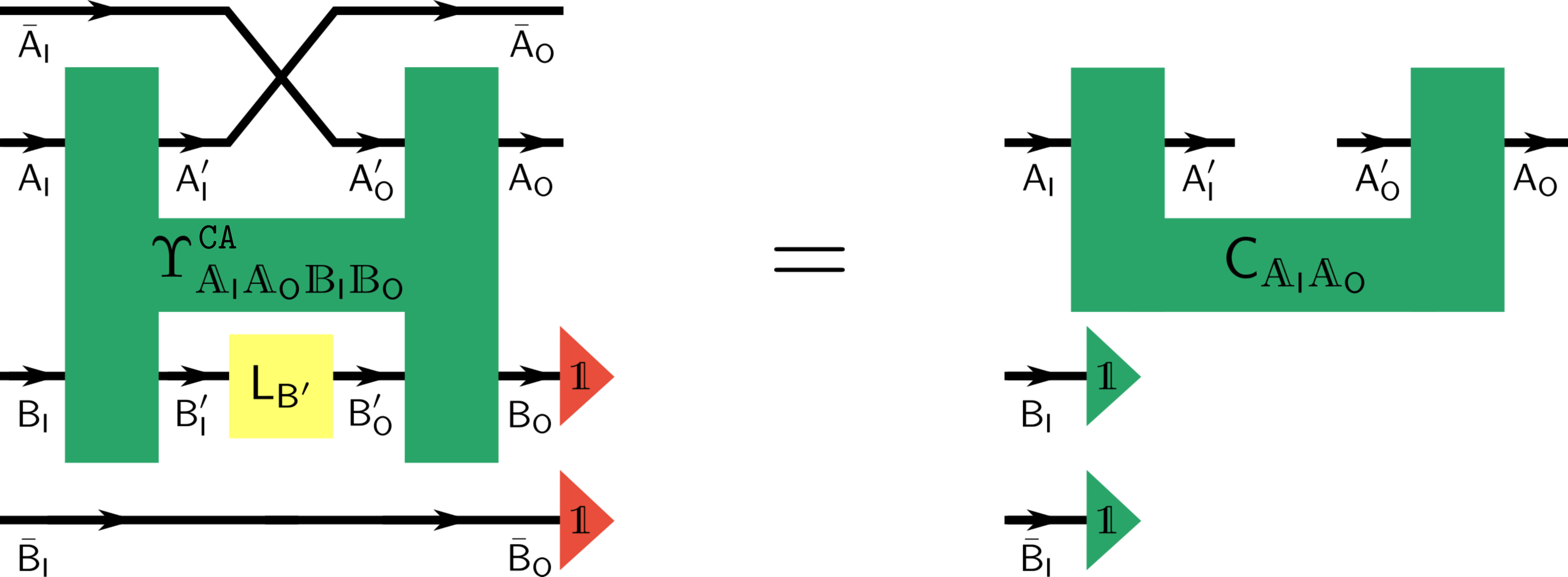}
    \caption{\textbf{Complete admissibility and non-signalling.} Taking the swap and a product map as the two maps on Alice's and Bob's side, respectively, if the adapter $\Upsilon^\texttt{CA}$ is completely admissible, then the resulting channel should be non-signalling (here, this requirement is shown for non-signalling from Alice to Bob).}  
    \label{fig::CLProof}
\end{figure}
Since the swap on Alice's side is merely a relabelling, we drop the usage of bars on the labels for Alice from here on. Demanding that, for this particular choice of CPTP maps, the resulting channel is non-signalling $B_I\bar{B}_I \not \rightarrow A_O A_I'$ is equivalent to the requirement (see Fig.~\ref{fig::CLProof})
\begin{gather}
\label{eqn::CL3}
    \tr_{B_O\bar B_O}(\Upsilon_{\Ads_I\Ads_O\Bds_I\Bds_O}^\texttt{CA} \star L_{B'}) = \ident_{B_I\bar B_I} \otimes C_{\Ads_I \Ads_O}\, , 
\end{gather}
where $C_{\Ads_I\Ads_O}$ is a causally ordered comb (since $\Upsilon_{\Ads_I\Ads_O\Bds_I\Bds_O}^\texttt{CA}$ is an admissible adapter), but its concrete form is not important for the following argument. We now show that $C_{\Ads_I\Ads_O}$ must be \textit{independent} of $L_B'$ if $\Upsilon_{\Ads_I\Ads_O\Bds_I\Bds_O}^\texttt{CA}$ is completely admissible. To see this, let us assume the contrary, that there exist two different CPTP maps $L_{B'}$ and $K_{B'}$ such that 
\begin{gather}
    \tr_{B_O\bar B_O}(\Upsilon_{\Ads_I\Ads_O\Bds_I\Bds_O}^\texttt{CA} \star L_{B'}) := \ident_{B_I\bar B_I} \otimes C_{\Ads_I \Ads_O} \neq  \ident_{B_I \bar B_I} \otimes C^\#_{\Ads_I \Ads_O} := \tr_{B_O\bar B_O}(\Upsilon_{\Ads_I\Ads_O\Bds_I\Bds_O}^\texttt{CA} \star K_{B'})\, .
\end{gather}
Then, Bob can perform a CPTP map $H_{B'\bar{B}}$ that classically controls the CPTP  maps $L_{B'}$ (for measurement outcome $0$ on $\bar B_I$) and $K_{B'}$ (for measurement outcome 1 on $\bar B_I$), i.e., 
\begin{gather}
    H_{B'\bar{B}'} = \ketbra{00}{00}_{\bar{B}_I\bar B_O} \otimes L_{B'} + \ketbra{11}{11}_{\bar{B}_I\bar B_O} \otimes K_{B'}\, , 
\end{gather}
and we assume that the degrees of freedom labeled with bars correspond to qubits. Now, contracting this map with the adapter $\Upsilon_{\Ads_I\Ads_O\Bds_I\Bds_O}^\texttt{CA}$ (and tracing out the final degrees of freedom $B_O$ and $\bar B_O$) yields [cf.~\eqref{eqn::CL3}]
\begin{gather}
    \tr_{B_O\bar B_O}(\Upsilon_{\Ads_I\Ads_O\Bds_I\Bds_O}^\texttt{CA}) \star H_{B'\bar{B}} = \ketbra{0}{0}_{\bar{B}_I} \otimes \ident_{B_I} \otimes C_{\Ads_I \Ads_O} + \ketbra{1}{1}_{\bar{B}_I} \otimes   \ident_{B_I} \otimes C_{\Ads_I \Ads_O}^\#\, .
\end{gather}
However, since $H_{B'\bar{B}}$ is itself a CPTP map, the above must -- due to complete admissibility of $\Upsilon_{\Ads_I\Ads_O\Bds_I\Bds_O}^\texttt{CA}$ -- also factorize as 
\begin{gather}
    \tr_{B_O\bar B_O}(\Upsilon_{\Ads_I\Ads_O\Bds_I\Bds_O}^\texttt{CA}) \star H_{B'\bar{B}} = \ident_{B_I\bar B_I} \otimes D_{\Ads_I \Ads_O}\, ,
\end{gather}
for some causally ordered comb $D_{\Ads_I \Ads_O}$. This, in turn, shows that $C_{\Ads_I \Ads_O} = C_{\Ads_I \Ads_O}^\#$ holds, i.e., Alice's part of the adapter is independent of the CPTP map that Bob performs (if Bob's final degree of freedom is discarded). The same argument can be employed to show that Bob's part of the adapter $\Upsilon_{\Ads_I\Ads_O\Bds_I\Bds_O}^\texttt{CA}$ is independent of the CPTP map that Alice performs (if Alice's final degrees of freedom are discarded). In fact, this is already sufficient to conclude that $\Theta^{\texttt{NS}} \supseteq \Theta^\texttt{CA}$, since it is exactly this property that we used to derive the non-signalling conditions of $\Theta_\texttt{NS}$. Combining this with the fact that $\Theta^{\texttt{NS}} \subseteq \Theta^\texttt{CA}$ (shown above), we see that the requirements on completely admissible adapters are equivalent to those on non-signalling adapters and we have $\Theta_{\texttt{NS}} = \Theta_{\texttt{CA}}$.

\section{Convexity of \texorpdfstring{$\Rcal_s(W)$}{}}
\label{app:conv}

In order to prove the convexity of $\Rcal_s(W)$, assume that for a given $W = pW_1 + (1-p)W_2$, we have found the process matrices $T^*\in \texttt{Proc}$ and $C^*\in \texttt{Free}$ that lead to the minimal value in Eq.~\eqref{eqn::ParRobust}. Then, 
\begin{gather}
    W = (1+\Rcal_s(W))C^* - \Rcal_s(W)T^*. 
\end{gather}
Importantly, any other such \textit{pseudo-mixture} representation $W = (1+\Rcal')C' - \Rcal'T'$, where $C'\in \texttt{Free}$ and $T'\in \texttt{Proc}$ satisfies (by definition) $\Rcal' \geq \Rcal_s(W)$. 

Let $C_1^*, C_2^*, T_1^*,$ and $T_2^*$ be the corresponding process matrices for $W_1$ and $W_2$, \textit{i.e.}, $W_1 = (1+\Rcal_s(W_1))C_1^* - \Rcal_s(W_1)T_1^*$ and $W_2 = (1+\Rcal_s(W_2))C_2^* - \Rcal_s(W_2)T_2^*$.
For convenience, we will set $\Rcal_i := \Rcal_s(W_i)$. Consequently, we have 
\begin{gather}
W = p(1+\Rcal_1)C_1^* + (1-p)(1+\Rcal_2)C_2^* - p\Rcal_1T_1^* - (1-p)\Rcal_2T_2^*. 
\end{gather}
Up to normalization, the first two terms together form the convex combination of two free processes, while the second two terms form a convex combination of proper process matrices. Thus, we can express $W$ as 
\begin{gather}
W = [p(1+\Rcal_1) + (1-p)(1+\Rcal_2)]C' -[p\Rcal_1 + (1-p)\Rcal_2]T', 
\end{gather}
where $T' \in \texttt{Proc}$ and $C' \in  \texttt{Free}$. This implies 
\begin{gather}
\Rcal_s(pW_1 + (1-p)W_2) \leq p\Rcal_s(W_1) + (1-p)\Rcal_s(W_2), 
\end{gather}
\textit{i.e.}, $\Rcal_s(W)$ is a convex function.

\section{Three-party processes from \texorpdfstring{$W^{A\rightarrow B \rightarrow C}$}{}}
\label{app::tri_part_trans}

In the main text we showed that all processes $W^{A\prec B}$ that satisfied the dimensional constraints we laid out can be obtained from the fully signalling process matrix $W^{A\rightarrow B}$ by means of free adapters. Here, we show that the analogous statement holds for processes of the form $W^{A\prec B\prec C}$, which can all be obtained from the fully signalling three-party process matrix \begin{gather}
    W^{A\rightarrow B \rightarrow C} = \rho_{A_I} \otimes \Phi^+_{A_OB_I} \otimes \Phi^+_{B_OC_I} \otimes \ident_{C_I}
\end{gather} by means of free adapters. 

\begin{figure}
    \centering
    \begin{minipage}{.48\textwidth}
        \centering
        \includegraphics[width=\linewidth]{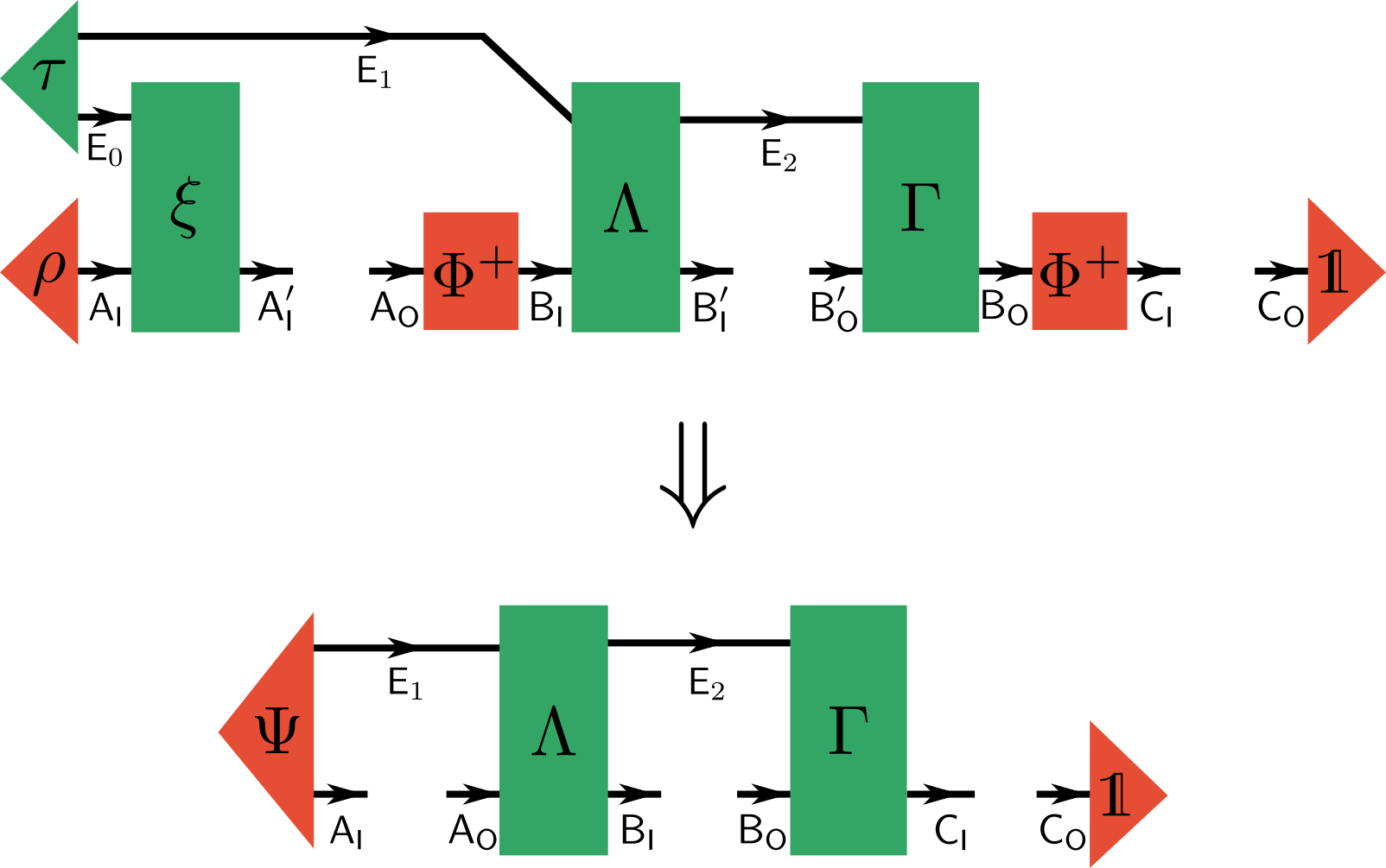}
        \caption{\textbf{General $W^{A\prec B\prec C}$ from $W^{A\rightarrow B \rightarrow C}$.} By performing the correct local adapters and sharing an entangled state, Alice and Bob can implement the representation of general three-step combs. For convenience, we have relabelled the spaces in the resulting comb (bottom panel).}
        \label{fig::Three_stepTrans}
    \end{minipage}%
    \hspace*{10mm}
    \begin{minipage}{0.48\textwidth}
        \centering
        \includegraphics[width=\linewidth]{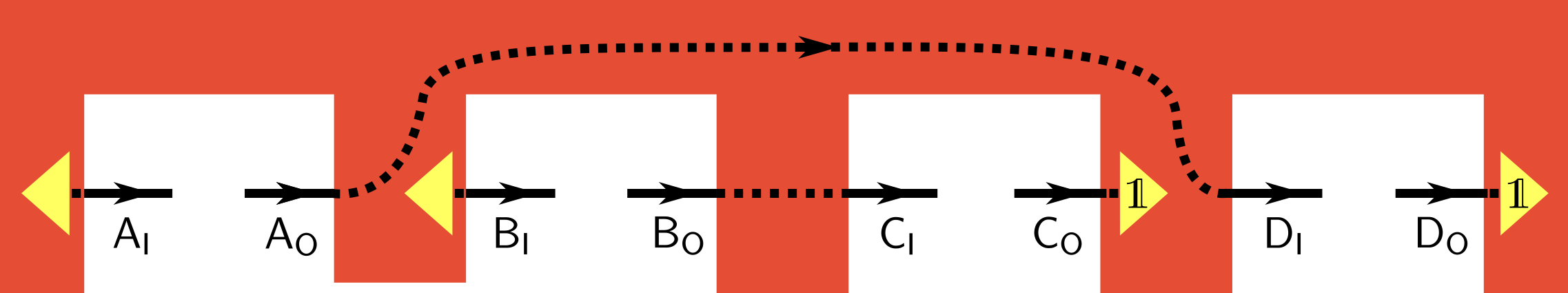}
        \caption{\textbf{Conjectured unreachable comb.} Starting from $W^{A\rightarrow B\rightarrow C\rightarrow D}$, we conjecture that the above comb, which basically corresponds to an identity channel from $A_O$ to $C_I$ and an identity channel from $B_O$ to $C_I$, cannot be reached by free adapters.}
        \label{fig::Four_Conj}
    \end{minipage}
\end{figure}

As for the two-party case, let us motivate this statement graphically (see Fig.~\ref{fig::Three_stepTrans} for a graphical representation). Any causally ordered comb $W^{A\prec B\prec C}$ can be represented as a concatenation of an initial system-environment state $\Psi_{A_IE_1}$ and two CPTP mappings $A_OE_1 \rightarrow B_IE_2$ and $B_OE_2 \rightarrow C_I$ with corresponding Choi matrices $\Lambda_{E_1A_O E_2 B_I}$ and $\Gamma_{E_2B_O C_I}$. It is then easy to see that, by only using shared entantglement and local operations, starting from $W^{A\rightarrow B\rightarrow C}$ Alice and Bob can implement any such $\Psi, \Lambda$ and $\Gamma$, and thus reach all $W^{A\prec B\prec C}$ (see Fig.~\ref{fig::Three_stepTrans}). In particular, let $\tau_{E_O E_1} = \Psi_{E_O E_1} \cong \Psi_{A_IE_1}$ be the initially shared entangled state, and let $\xi_{A_IE_OA_I'} = \ident_{A_I} \otimes \Phi^+_{E_OA_I'}$ be the Choi matrix of a local map in Alice's laboratory that traces out the system on $A_I$ and relabels $E_O \mapsto A_I'$. Then $\rho_{A_I} \star  \tau_{E_O E_1} \star \xi_{A_IE_OA_I'} = \Psi_{A_I'E_1} \cong \Psi_{A_IE_1}$. Since the degree of freedom $E_1$ is shared with Bob's laboratory, it is then sufficient for Bob to locally perform CPTP maps with corresponding Choi matrices $\Lambda_{B_IE_1E_2B_I'}$ and $\Gamma_{E_2B_O'B_O}$. The identity maps in the fully signalling comb $W^{A\rightarrow B \rightarrow C}$ then lead to a relabelling $B_I \mapsto A_O$ and $B_O \mapsto C_I$. Finally then, identifying $A_I'$ with $A_I$, $B_I'$ with $B_I$ and $B_O'$ with $B_O$, we find that 
\begin{gather}
\begin{split}
    W^{A\rightarrow B \rightarrow C} \star \Upsilon &=  (\rho_{A_I} \otimes \Phi^+_{A_OB_I} \otimes \Phi^+_{B_OC_I} \otimes \ident_{C_I}) \star (\tau_{E_O E_1} \star \xi_{A_IE_OA_I'} \star \Lambda_{B_IE_1E_2B_I'} \star \Gamma_{E_2B_O'B_O}) \\
    &\cong \Psi_{A_IE_1} \star \Lambda_{E_1A_O E_2 B_I} \star \Gamma_{E_2B_O C_I}\, ,
\end{split}
\end{gather}
and thus all combs $W^{A\prec B\prec C}$ can be reached from $W^{A\rightarrow B\rightarrow C}$ by means of free adapters. 

While the above derivation works for the case of three parties, it cannot be employed in the four-party scenario, where we conjecture the existence of ordered processes that cannot be reached from $W^{A\rightarrow B\rightarrow C\rightarrow D}$ by means of free adapters. In particular, we expect the following process to not be reachable (see Fig.~\ref{fig::Four_Conj}): 
\begin{gather}
    W_{4\text{parties}} = \Psi_{A_I} \otimes \Phi^+_{A_OD_I} \otimes \eta_{B_I} \otimes \Phi^+_{B_OC_I} \otimes \ident_{C_OD_O}
\end{gather}
This process consists of an identity channel between Alice and Dave, and identity channel between Bob and Alice, and either states that are fed in, or degrees that are traced out on the remaining degrees of freedom. The reason why we conjecture that this process cannot be reached from $W^{A\rightarrow B \rightarrow C \rightarrow D}$ by means of an adapter that does not enable any communication between the parties is as follows: If an adapter $\Upsilon$ does not enable direct communication between the respective parties, all signalling between Alice and Dave, as well as the signalling between Bob and Charlie must have come from $W^{A\rightarrow B \rightarrow C\rightarrow D}$. Consequently, signals that go from Alice to Dave must have gone via Bob (since there is no direct channel between Alice and Dave in $W^{A\rightarrow B \rightarrow C\rightarrow D}$). However, Bob has an identity channel to Charlie, but, at the same time, needs to feed forward signals from Alice to Dave, so that the resulting $W_{4\text{parties}}$ can contain an identity channel between Alice and Dave. This, then, should not be possible, since Bob cannot, at the same time, send his information undisturbed to Charlie, but concurrently also perfectly feed Alice's signals forward towards Charlie. 

While somewhat intuitive, these arguments are not sufficient to exclude the possibility to obtain $W_{4\text{parties}}$ from $W^{A\rightarrow B \rightarrow C\rightarrow D}$ by means of free adapters. As we have seen, there are free adapters that cannot be obtained from shared entanglement and local operations. They thus require internal signalling for their implementation, however, this signalling is such that it cannot be employed by the respective parties to signal to each other. In principle, though, this signalling might be \textit{activated} by concatenation with $W^{A\rightarrow B \rightarrow C\rightarrow D}$. Nonetheless, we conjecture that $W_{4\text{parties}}$ cannot be obtained from $W^{A\rightarrow B \rightarrow C\rightarrow D}$ by means of free adapters.

\section{Signalling robustness of multipartite causally ordered processes}\label{app::robust_multipartite}

Here, we provide proofs and generalizations of the results on the bounds on the signalling robustness in the multi-party scenario mentioned in the main text. In Sec.~\ref{subsec::robust_multipartite} we claimed the following result: \newline \newline
\textbf{Proposition 5.}
{\it
Let $\textup{\texttt{Proc}}_{1:N}$ be the set of all causally ordered processes on $N$ parties with causal order $X^{(1)}\prec X^{(2)} \prec \cdots \prec X^{(N)}$. For any $W\in \textup{\texttt{Proc}}_{1:N}$ we have 
\begin{gather}
    \Rcal_s(W) \leq \prod_{i=1}^{N-1} d_{X^{(i)}_O}^{2} - 1:= d_{\bar{O}}^2-1,
\end{gather} 
where  $d_{X^{(i)}_O}$ is the output dimension of party $X^{(i)}_O$. } \newline

\begin{proof}
The proof follows a similar line as that of Thm.~B3 in Ref.~\cite{nery_simple_2021}. First, consider the quantum channel $\widetilde{\Lambda}[\rho] = \tr(\rho) \ident/d_X$, where $d$ is the dimension of the space $X$ on which $\rho$ is defined. This channel can be written as $\widetilde{\Lambda}[\rho] = \frac{1}{d^2}\sum_{j=0}^{d^2-1} U_j \rho U_j^\dagger$, where the $d^2$ unitary matrices $\{U_j\}$ form an orthogonal basis of the space spanned by unitary matrices and $U_0 = \ident$. Using this decomposition of $\widetilde \Lambda$, we can write its action as 
\begin{gather}
\label{eqn::decompOrthU}
    \widetilde{\Lambda}[\rho] = \frac{1}{d^2} \rho + \frac{1}{d^2} \sum_{j=1}^{d^2-1} U_j \rho U_j^\dagger =: (\frac{1}{d^2} \Ical + \frac{d^2 - 1}{d^2} \widetilde{\Gamma})[\rho]\, ,
\end{gather}
where $\widetilde{\Gamma}$ is a CPTP map and $\Ical$ is the identity map. The map $\widetilde{\Lambda}$ traces out the state it acts on and replaces it by a maximally mixed state. Consequently, applying it to all output spaces $X_O^{(i)}$ (except the last one) of $W\in \textup{\texttt{Proc}}_{1:N}$ yields a process of the form $\rho_{X_I} \otimes \ident_{X_O}$. In detail, due to its causal ordering, we have $W = W_{X_I^{(1)}X_O^{(1)}\cdots X_I^{(N)}}\otimes \ident_{X_O^{(N)}}$. Applying $\widetilde \Lambda$ to each output space of $W$ except for the last one yields 
\begin{gather}
    {}_{X_O^{(1)}\cdots X_O^{(N-1)}}W 
    = 1/(d_{X^{(1)}_O} \cdots d_{X^{(N-1)}_O})\tr_{X^{(1)}_O \cdots X^{(N-1)}_O}W \otimes \ident_{X_O} 
    =:\rho_{X_I} \otimes \ident_{X_O}\, ,
\end{gather}
where $X_I$ $(X_O)$ denotes \textit{all} input (output) spaces. Now, employing Eq.~\eqref{eqn::decompOrthU} and setting $d_{\bar{O}}^2:= \prod_{i=1}^{N-1} d_{X^{(i)}_O}^2$, we obtain 
\begin{gather}
\label{eqn::OmegaEq}
    {}_{X_O^{(1)}\cdots X_O^{(N-1)}}W 
    = \prod_{i=1}^{N-1} (\frac{1}{d_{X^{(i)}_O}^2} \Ical_{X_O^{(i)}} + \frac{d_{X^{(i)}_O}^2 - 1}{d_{X^{(i)}_O}^2} \widetilde{\Gamma}_{X_O^{(i)}})[W]   =: (\frac{1}{d_{\bar{O}}^2} \Ical_{X_O} + (1- \frac{1}{d_{\bar{O}}^2}) \widetilde{\Omega}_{X_O})[W]\, , 
\end{gather}
where $\widetilde{\Omega}_{X_O}$ is a CPTP map. Let us assume that $T:= \widetilde{\Omega}_{X_O}[W]$ is a proper process matrix. In this case, using the above considerations, we have $1/d_{\bar{O}}^2 W + (1- 1/d_{\bar{O}}^2) T \in \texttt{Free}$. Recalling the definition of the signalling robustness 
\begin{gather}
\Rcal_s(W) = \underset{T \in \texttt{Proc}} \min \left\{ s\geq 0 \left| \frac{W + sT}{1+s} = C \in \texttt{Free} \right.\right\}\, ,
\end{gather}
we see that $\Rcal_s(W)\leq \prod_{i=1}^{N-1} d_{X^{(i)}_O}^2 -1$, which corresponds to the claim of the Proposition. It remains to show that $\widetilde{\Omega}_{X_O}[W]$ is indeed a proper process matrix. To this end, we note from Eq.~\eqref{eqn::OmegaEq} that $\widetilde \Omega_{X_O}$ corresponds to a convex combination of local CPTP maps that act on the respective output spaces of each party. Since each such tensor product of maps simply corresponds to a post-processing of the local output spaces, it transforms a proper process matrix to a proper process matrix, and analogously so does a convex combination of such maps. This concludes the proof.  
\end{proof}
For the case where all involved dimensions are equal (the more general case will be discussed below), it is now straightforward to show that a process of the form
\begin{gather}
\label{eqn::MostResMulti}
    W^{X^{(1)} \rightarrow \cdots \rightarrow X^{(N)}}
    = \frac{\ident_{X^{(1)}_I}}{d_{X^{(1)}_I}} \otimes \Phi^+_{X^{(1)}_OX^{(2)}_I} \otimes \cdots\otimes \Phi^+_{X^{(N-1)}_OX^{(N)}_I} \otimes \ident_{X^{(N)}_O}
\end{gather}
maximizes the signalling robustness on $\textup{\texttt{Proc}}_{1:N}$, as claimed in the main text. In clear analogy to the two-party case, the matrix $S\geq 0$ used in the dual SDP~\eqref{sdp::parallel_dual} for the computation of the signalling robustness in the multi-party case satisfies 
\begin{gather}
    \tr_{X_O^{(1)}\cdots X_O^{(N)}}S = \ident_{X_I^{(1)}\cdots X_I^{(N)}}. 
\end{gather}
Then, choosing 
\begin{gather}
\label{eqn::witnessMulti}
    S = \ident_{X^{(1)}_I} \otimes \Phi^+_{X^{(1)}_OX^{(2)}_I} \otimes \cdots \otimes \Phi^+_{X^{(N-1)}_OX^{(N)}_I} \otimes \ident_{X^{(N)}_O}/d_{X^{(N)}_O}\, ,
\end{gather}
we see that $\tr(SW^{X^{(1)}\rightarrow \cdots \rightarrow X^{(N)}}) = d_{\bar{O}}^2-1$. This proofs Cor.~\ref{cor::maxProcOrdered} of the main text.

Evidently, the proof of Prop.~\ref{prop::maxRordered} does not work for general, causally indefinite, processes, since we explicitly used the fact that $W$ factorizes, i.e., it satisfies $W = {}_{X_O^{(N)}}W$. Nonetheless, we can slightly generalize the above statement without changing the proof. In particular, even though we focused on causally ordered processes, the only actual assumption we employ in the above proof is that there is a definite last party ($X^{(N)}$ in the notation of the proof). Whether or not the remainder of the process is causally ordered did not factor in. Consequently, we have the following corollary:
\begin{corollary}
Let $\textup{\texttt{Proc}}_{\prec N}$ be the set of all processes that have $X^{(N)}$ as the definite last party, i.e., if $W \in \textup{\texttt{Proc}}_{\prec N}$, then ${}_{X^{(N)}_O}W = W$ and ${}_{X^{(N)}_OX^{(N)}_I}W = {}_{X^{(N)}_OX^{(N)}_IX^{(N-1)}_O}W$. For all $W \in \textup{\texttt{Proc}}_{\prec N}$ we have $\Rcal_s(W) \leq \prod_{i=1}^{N-1} d_{X^{(i)}_O}^{2} - 1$, where  $d_{X^{(i)}_O}$ is the output dimension of party $X^{(i)}$.  
\end{corollary}
The proof of this statement follows along the same lines as that of Prop.~\ref{prop::maxRordered} for causally ordered processes. 

Having found these upper bounds for the signalling robustness on the sets $\textup{\texttt{Proc}}_{1:N}$ and $\textup{\texttt{Proc}}_{\prec N}$ of causally ordered processes and processes with a definite last party, respectively, then allows one to find an analogous bound on their convex hulls.
\begin{corollary}
\label{cor::MaxRobSep}
Let $\textup{\texttt{Proc}}_{\pi(N)}^\textup{\texttt{conv}}$ be the convex hull of all processes that have a definitive last party. Then, for all $W \in \textup{\texttt{Proc}}_{\pi(N)}^\textup{\texttt{conv}}$ we have $\Rcal_s(W) \leq \underset{\bar{O}}{\max} \{d_{\bar{O}}^2 - 1\}$, where $d_{\bar{O}}^2$ is defined as in Prop.~\ref{prop::maxRordered}, i.e., the dimension of all spaces except for the last one, and the maximization is over all possible causal orders -- given by the permutation $\pi$ -- of $N$ parties. 
\end{corollary}
Since $\textup{\texttt{Proc}}_{1:N} \subset \textup{\texttt{Proc}}_{\prec N}$, the above Corollary also applies to the convex hull of $\textup{\texttt{Proc}}_{1:N}$. 
\begin{proof}
If $W\in \textup{\texttt{Proc}}_{\pi(N)}^\textup{\texttt{conv}}$ then it can be written as $W = \sum_{\pi} p_{\pi} W^{\pi(N)}$, where $\pi$ is a permutation of $N$ elements, $\{p_{\pi}\}$ is a probability distribution with $\sum_{\pi} p_{\pi}\ = 1$, and $W^{\pi(N)}$ denotes a process with final party $X^{(\pi(N))}$. Using the convexity of the signalling robustness, we obtain
\begin{gather}
    \Rcal_s(W) = \Rcal_s(W)\Big(\sum_{\pi} p_{\pi} W^{\pi(N)}\Big) \leq \sum_{\pi} p_{\pi} \Rcal_s(W^{\pi(N)}) \leq \underset{\pi}{\max} \Rcal_s(W^{\pi(N)}) = \underset{\bar{O}}{\max} \{d_{\bar{O}}^2 - 1\}.
\end{gather}
\end{proof}

A priori, this bound is not tight. However, under the assumption that there is a causal ordering of parties such that $d_{\bar O}^2$ is maximized and the dimensions of all input spaces are lower bounded by the dimension of their respective preceding output spaces (i.e., $d_{X^{(i)}_O} \leq d_{X^{(i+1)}_I}$ for all $i = 1,\dots, N-1$), we have the following Proposition: 
\begin{proposition}
\label{prop::maxRob}
Let $W$ be a process matrix defined on the $N$ laboratories $\{X^{(i)}\}$. Let $X^{(1)} \prec X^{(2)} \prec \cdots \prec X^{(N)}$ be an order in which $d_{\bar O}^2$ is maximized, and let $d_{X^{(i)}_O} \leq d_{X^{(i+1)}_I}$ for all $i = 1,\dots, N-1$. Then there exists a causally ordered Markovian process $W^{X^{(1)} \prec \cdots \prec X^{(N)}}$ that maximizes the signalling robustness on $\textup{\texttt{Proc}}_{\pi(N)}^\textup{\texttt{conv}}$, i.e., it satisfies 
\begin{gather}
    \Rcal_s(W^{X^{(1)} \prec \cdots \prec X^{(N)}}) = \underset{\bar{O}} \max \{d_{\bar O}^2 - 1\}\, .
\end{gather}
\end{proposition}
\begin{proof}
First, assume that the dimension of neighboring input and output spaces is the same, i.e., $d_{X^{(i)}_O} = d_{X^{(i+1)}_I}$ for all $i = 1, \dots, N-1$. Then the above statement is directly implied by from Cor.~\ref{cor::maxProcOrdered}

The case where $d_{X^{(i)}_O} < d_{X^{(i+1)}_I}$ for some (or possibly all) adjacent laboratories follows in a similar vein: First, we define a `generalized identity map' $\Ical_{X^{(i)}_O \rightarrow X^{(i+1)}_I}$ by its action on all states $\rho\in \Bcal(\Hcal_{X^{(i)}_O})$: 
\begin{gather}
    \Ical_{X^{(i)}_O \rightarrow X^{(i+1)}_I}[\rho] = \begin{pmatrix} \rho & 0 \\ 0 & 0 \end{pmatrix} \in \Bcal(\Hcal_{X^{(i+1)}_I})\, ,
\end{gather}
i.e., $\Ical_{X^{(i)}_O \rightarrow X^{(i+1)}_I}[\rho]$ transmits $\rho$ and merely pads it out with zeros such that the output lives in the right space. The corresponding Choi matrix of $\Ical_{X^{(i)}_O \rightarrow X^{(i+1)}_I}$ is given by 
\begin{gather}
    \Phi_{X^{(i)}_O \rightarrow X^{(i+1)}_I}^+ := \begin{pmatrix} \Phi_{X^{(i)}_O X^{(i)'}_O}^+ & 0 \\ 0 & 0 \end{pmatrix}\, ,
\end{gather}
where $\Hcal_{X^{(i)}_O} \cong \Hcal_{X^{(i)'}_O}$. Using this generalized identity map, the corresponding witness $S$ and the causally ordered process $W^{X^{(1)} \rightarrow \cdots \rightarrow X^{(N)}}$ can be defined in the same vein as in Eqs~\eqref{eqn::MostResMulti} and~\eqref{eqn::witnessMulti} with the same result for the signalling robustness.
\end{proof}

Since the signalling robustness of causally non-separable processes does not seem to exceed that of the processes of Eq.~\eqref{eqn::MostResMulti}, allowing for communication in several directions beyond probabilistic mixing does not appear increase the amount of causal connection. The fact that the upper bound cannot necessarily be achieved in the case where the input dimension of an input space is not larger than or equal to its preceding output space then also possesses a straightforward explanation. In this case, there is no identity channel from one laboratory to the other, leading to lower causal connection and an unattainability of the bound $\underset{\bar{O}} \max \{d_{\bar O}^2 - 1\}$. 

\section{All monotones of the resource theory of causal connection}
\label{app::Montones_allmonotones}

As we have seen, neither the signalling nor the causal robustness are sufficient to decide whether a process can be converted into another by means of free adapters. Intuitively though, if one knew \textit{all} monotones of the resource theory of causal connection, one could unambiguously decide whether or not two processes can be connected via free adapters. Here, using results from Ref.~\cite{gour_dynamical_res_2020}, we make this intuition concrete and provide \textit{all} monotones of the resource theory of causal connection. These monotones, in turn, provide a necessary and sufficient condition for the convertibility of processes. 

In Ref.~\cite{gour_dynamical_res_2020}, the complete set of monotones for a wide class of quantum resource theories of processes has been derived. There, the authors consider sets of free maps that map CPTP maps onto CPTP maps, and, under weak assumptions on the set of free CPTP maps, derive all monotones of such a resource theory. Here, we follow these ideas to provide a characterization of all monotones in a resource theory where the free processes are given by $\texttt{Free} = \{\rho_{A_IB_I} \otimes \ident_{A_OB_O} \in \Bcal(\Hcal_{AB})\}$. Unlike in the previous sections, we will \textit{not} work with the Choi matrices of adapters, but consider the free adapters as maps $\widetilde \Upsilon$ that map Choi matrices $W$ on $\Bcal(\Hcal_{AB})$ to Choi matrices $W'$ on $\Bcal(\Hcal_{A'B'})$. Concretely, with this, we get the set of free transformations
\begin{gather}
    \bar \Theta_{\texttt{NS}}(AB \rightarrow A'B') = \{\widetilde \Upsilon| \Upsilon_{\Ads_I\Ads_O\Bds_I\Bds_O} \in \Theta_{\texttt{NS}}\}\, ,
\end{gather}
where  $\Upsilon_{\Ads_I\Ads_O\Bds_I\Bds_O}$ is the Choi matrix of $\widetilde \Upsilon$.

Considering the map $\widetilde \Upsilon$ instead of its Choi matrix $\Upsilon_{\Ads_I\Ads_O\Bds_I\Bds_O}$ will lead to some notational simplifications in what follows. Naturally, everything can be phrased entirely in terms of Choi matrices (or entirely in terms of maps, as is done in Ref.~\cite{gour_dynamical_res_2020}). In what follows, whenever there is no risk of confusion, we will drop the explicit input and output spaces and simply write $\bar \Theta_{\texttt{NS}}$ instead of $\bar \Theta_{\texttt{NS}}(AB \rightarrow A'B')$. 

Following Ref.~\cite{gour_dynamical_res_2020}, we can define the following monotones for our resource theory: 
\begin{gather}
\label{eqn::ResMons}
 f_Q(W) = \underset{\widetilde \Upsilon \in \bar{\Theta}_{\texttt{NS}}} \max \tr(Q \widetilde \Upsilon(W))\,,
\end{gather}
where $W$ is a proper process matrix on $\Bcal(\Hcal_{AB})$ and $Q$ is a proper process matrix on $\Bcal(\Hcal_{A'B'})$. Note that we can use the maximum instead of the supremum in the above equation, as the set of free adapters we consider is compact, i.e., the maximum exists. Since the concatenation of free adapters still yields a free adapter, it is easy to see that Eq.~\eqref{eqn::ResMons} indeed defines a monotone of the resource theory. With this, we obtain an analogous characterization of a \textit{complete} set of monotones as the one found in~\cite{gour_dynamical_res_2020}:

\begin{proposition}
\label{thm::Monotones}
Let $W_{AB}$ be a proper process matrix on $\Bcal(\Hcal_{AB})$ and $W'_{A'B'}$ be a proper process matrix on $\Bcal(\Hcal_{A'B'})$. Then there exists a free adapter $\widetilde \Upsilon \in \bar \Theta_{\textup{\texttt{NS}}}$ such that $W'_{A'B'} = \widetilde \Upsilon(W_{AB})$ iff 
\begin{gather}
\label{eqn::All_mono}
   f_Q(W_{AB}) \geq f_Q(W'_{A'B'})  \quad \forall Q \in \textup{\texttt{Proc}}_{A'B'}\, , 
\end{gather}
where $\textup{\texttt{Proc}}_{A'B'}$ is the set of valid process matrices on $A'B'$.
\end{proposition}

The proof proceeds analogous to the original proof of Thm.~4 in Ref.~\cite{gour_dynamical_res_2020}. First, using convexity and closedness of the set of all free adapters, we will show that Eq.~\eqref{eqn::All_mono} provides all monotones if $Q$ is \textit{Hermitian}. In the second step, we show that it is sufficient to restrict $Q$ to be a proper process matrix on $\Bcal(\Hcal_{A'B'})$. 

\begin{proof}
\textbf{Step 1}: We denote the orbit of a process matrix $W_{AB}$ on $\Bcal(\Hcal_{AB})$ under free transformations by $\Ocal_W = \{\widetilde \Upsilon(W_{AB})| \widetilde \Upsilon \in \bar \Theta_{\texttt{NS}}\}$. As $\bar \Theta_{\texttt{NS}}$ is convex and closed, so is $\Ocal_W$. Consequently, any process matrix $W'_{A'B'}$ is not an element of $\Ocal_W$ iff there exists a Hermitian matrix $Q_{A'B'}$ on $\Bcal(\Hcal_{A'B'})$ such that 
\begin{gather}
    \tr(Q_{A'B'}W'_{A'B'}) > \underset{\Upsilon \in \bar \Theta_{\texttt{NS}}} \max \tr(Q_{A'B'} \widetilde \Upsilon[W_{AB}])\, ,
\end{gather}
where the above follows from the hyperplane separation theorem for convex sets. On the other hand, $W'_{A'B'}$ is \textit{inside} the set $\Ocal_W$ iff for \textit{all} Hermitian matrices $Q_{A'B'}$ on $ \Bcal(\Hcal_{A'B'})$, we have 
\begin{gather}
\label{eqn::Membership}
    \tr(Q_{A'B'}W'_{A'B'}) \leq \underset{\Upsilon \in \bar \Theta_{\texttt{NS}}} \max \tr(Q_{A'B'} \Upsilon(W_{AB}))\, .
\end{gather}

Naturally, the `only if' direction of the above equation holds trivially for any set, while the `if' direction requires the convexity of the set. Eq.~\eqref{eqn::Membership} allows one to check whether or not there exists an $\Upsilon \in \bar \Theta_{\texttt{NS}}$ such that $\Upsilon(W_{AB}) = W'_{A'B'}$; this is exactly the case when $W'_{A'B'}$ satisfies Eq.~\eqref{eqn::Membership}. Now, we can show that Eq.~\eqref{eqn::Membership} is equivalent to Eq.~\eqref{eqn::All_mono} when we allow $Q_{A'B'}$ to be Hermitian (instead of a proper process matrix). 

First, if Eq.~\eqref{eqn::All_mono} holds (for all Hermitian $Q_{A'B'}$), then we have 
\begin{gather}
    \underset{\widetilde \Upsilon' \in \bar \Theta_{\texttt{NS}}(A'B' \rightarrow A'B')} \max \tr[Q_{A'B'} \widetilde \Upsilon'(W'_{A'B'})]
    \leq \underset{\Upsilon \in \bar \Theta_{\texttt{NS}} (AB \rightarrow A'B')} \max \tr[Q_{A'B'} \widetilde \Upsilon(W_{AB})] \quad \forall Q_{A'B'}\, .
\end{gather}
Choosing $\widetilde\Upsilon'$ to be the identity map (which is a free adapter), we see that Eq.~\eqref{eqn::Membership} is satisfied, and consequently $W'_{A'B'} \in \Ocal_W$. 

Conversely, if Eq.~\eqref{eqn::Membership} holds, then for any $\widetilde \Upsilon' \in \bar \Theta_{\texttt{NS}}(A'B' \rightarrow A'B')$ we have 
\begin{gather}
\begin{split}
    \tr(Q_{A'B'}\widetilde \Upsilon'(W'_{A'B'})) &= \tr(\widetilde\Upsilon^{\prime \dagger}(Q_{A'B'})W'_{A'B'}) \leq \underset{\Upsilon \in \bar\Theta_{\texttt{NS}}(AB \rightarrow A'B')} \max\tr[\widetilde\Upsilon^{\prime \dagger}(Q_{A'B'}) \widetilde \Upsilon(W_{AB})] \\
    &= \underset{\Upsilon \in \bar \Theta_{\texttt{NS}}(AB \rightarrow A'B')}\max\tr[Q_{A'B'} \widetilde \Upsilon^\prime \circ \widetilde \Upsilon(W_{AB})] \\
    &\leq  \underset{\Upsilon \in \bar \Theta_{\texttt{NS}}(AB \rightarrow A'B')}\max\tr[Q_{A'B'} \widetilde \Upsilon(W_{AB})] = f_Q(W)\, ,
\end{split}
\end{gather}
where we used Eq.~\eqref{eqn::Membership} in the second line, $\widetilde \Upsilon^{\prime \dagger}$ is the dual map of $\widetilde \Upsilon$, and we have employed the fact that $\widetilde \Upsilon' \circ \widetilde \Upsilon$ is still a free operation. As the above equation holds for all $Q_{A'B'}$ and free adapters $\widetilde \Upsilon'$, it implies $f_{Q}(W_{AB}) \geq f_Q(W'_{A'B'})$. Now, it remains to show that we can restrict the set of matrices $Q_{A'B'}$ for which we have to check the monotones to the set of proper process matrices. 

\textbf{Step 2}: To show that proper process matrices are sufficient, we consider the matrix
\begin{gather}
    \bar{Q}_{A'B'} = (1-\varepsilon)\frac{\ident_{A'B'}}{d_{A_I'B_I'}} + \varepsilon L_V(Q_{A'B'}).
\end{gather}
This matrix satisfies $L_V(\bar{Q}_{A'B'}) = \bar{Q}_{A'B'}$, and for $\varepsilon$ small enough, we have $\bar{Q}_{A'B'}\geq 0$, i.e., $\bar{Q}_{A'B'}$ is -- up to normalization -- a proper process matrix. For this matrix, we obtain
\begin{gather}
\begin{split}
    \tr(\bar{Q}_{A'B'}W'_{A'B'}) &= d_{A_O'B_{O}'}^2(1-\varepsilon) + \varepsilon \tr(L_V(Q_{A'B'})W'_{A'B'}) \\
    &= d_{A_O'B_{O}'}^2(1-\varepsilon) + \varepsilon \tr(Q_{A'B'}W'_{A'B'})\, ,
\end{split}
\end{gather}
where we have used the self-duality of $L_V$ and the fact that $W'_{A'B'}$ is a proper process matrix, i.e., it satisfies $L_V(W'_{A'B'}) = W'_{A'B'}$. From the above equation we see that up to a constant offset that only depends on the dimension $d_{A_O'B_O'}$, and a positive constant pre-factor, $f_{\bar Q}$ yields the same value as $f_Q$. Consequently, checking Eq.~\eqref{eqn::All_mono} for proper process matrices $Q_{A'B'}$ is sufficient. This concludes the proof.
\end{proof}
Since we did not rely in the proof on the fact that there are only two parties, the above straightforwardly generalizes to more parties. 

While in principle providing a complete characterization of all the monotones of the theory of causal connection, the monotones provided in Prop.~\ref{thm::Monotones} are rather abstract, which obstructs their intuitive interpretation, and makes it hard to deduce how the monotones we already encountered in the main text (the signalling and the causal robustness) fit in with this complete family of monotones. While first steps towards an operational interpretation of such a family of monotones have been taken for the case of supermaps acting on CPTP maps~\cite{brandsen_what_2021}, giving the monotones we provided here a physical interpretation remains an open problem. 
\end{document}